\documentclass[11pt]{article}
\usepackage{amsfonts}
\usepackage{mathrsfs}
\usepackage{amsmath,amssymb,amsthm}
\usepackage{latexsym}
\usepackage{indentfirst}
\usepackage[top=1in, bottom=1in, left=1in, right=1in]{geometry}
\usepackage{color}
\usepackage{footmisc}
\usepackage{endnotes}
\usepackage{algpseudocode}
\usepackage{algorithm}
\usepackage{float}
\usepackage{graphicx}
\usepackage{subfig}
\usepackage{epstopdf}
\usepackage{enumitem}
\usepackage{comment}
\usepackage[round,authoryear]{natbib}
\usepackage{longtable}
\usepackage{authblk}
\input{./Definitions}

\newcommand{\vertiii}[1]{{\left\vert\kern-0.25ex\left\vert\kern-0.25ex\left\vert #1
    \right\vert\kern-0.25ex\right\vert\kern-0.25ex\right\vert}}

\def\T{{ \mathrm{\scriptscriptstyle T} }}

\def\vcm{VCM}

\def\vv{\mathrm{v}}

\allowdisplaybreaks

\DeclareMathOperator*{\Cov}{Cov}
\DeclareMathOperator*{\var}{Var}

\DeclareMathOperator*{\diag}{diag}

\DeclareMathOperator*{\tr}{tr}
\DeclareMathOperator*{\Tr}{Tr}

\def\T{{ \mathrm{\scriptscriptstyle T} }}

\newtheorem{theorem}{Theorem}
\newtheorem{lemma}{Lemma}

\theoremstyle{definition}

\begin{document}

\title{\bf Distributed Bayesian Varying Coefficient Modeling Using a Gaussian Process Prior}


\author[1]{Rajarshi Guhaniyogi \thanks{rajguhaniyogi@tamu.edu}}
\author[2]{Cheng Li \thanks{stalic@nus.edu.sg}}
\author[3]{Terrance D. Savitsky \thanks{savitsky.terrance@bls.gov}}
\author[4]{Sanvesh Srivastava \thanks{sanvesh-srivastava@uiowa.edu}}

\affil[1]{Department of Statistics, Texas A \& M University}
\affil[2]{Department of Statistics and Applied Probability, National University of Singapore}
\affil[3]{U.S. Bureau of Labor Statistics, Office of Survey Methods Research}
\affil[4]{Department of Statistics and Actuarial Science, University of Iowa}

\date{}
\maketitle

\begin{abstract}
  Varying coefficient models (VCMs) are widely used for estimating nonlinear regression functions in functional data models. Their Bayesian variants using Gaussian process (GP) priors on the functional coefficients, however, have received limited attention in massive data applications. This is primarily due to the prohibitively slow posterior computations using Markov chain Monte Carlo (MCMC) algorithms. We address this problem using a divide-and-conquer Bayesian approach that operates in three steps. The first step creates a large number of data  subsets with much smaller sample sizes by sampling without replacement from the full data. The second step formulates VCM as a linear mixed-effects model and develops a data augmentation (DA)-type algorithm for obtaining MCMC draws of the parameters and predictions on all the subsets in parallel. The DA-type algorithm appropriately modifies the likelihood such that every subset posterior distribution is an accurate approximation of the corresponding true posterior distribution. The third step develops a combination algorithm for aggregating MCMC-based estimates of the subset posterior distributions into a single posterior distribution called the Aggregated Monte Carlo (AMC) posterior. Theoretically, we derive minimax optimal posterior convergence rates for the AMC posterior distributions of both the varying coefficients and the mean regression function. We provide quantification on the orders of subset sample sizes and the number of subsets according to the smoothness properties of the multivariate GP. The empirical results show that the combination schemes that satisfy our theoretical assumptions, including the one in the AMC algorithm, have better nominal coverage, shorter credible intervals, smaller mean square errors, and higher effective sample size than their main competitors across diverse simulations and in a real data analysis.
\end{abstract}


\section{Introduction}
\label{sec:introduction}

We first introduce the motivation of studying Bayesian varying coefficient models with a Gaussian process prior for massive data applications. Then, we outline our main contributions in this work and discuss the related literature.

\subsection{Varying Coefficient Models Using a GP Prior}

VCMs are a flexible and popular extension of the linear regression model \citep{HasTib93}, in which the regression coefficients can be smooth functions that capture nonlinear dependence of the response function on the covariates. VCMs are extensively used in practice, including time-series \citep{CheTsa93,Caietal00}, longitudinal \citep{Wuetal98,Rupetal03}, spatial \citep{Geletal03}, and spatiotemporal data analysis \citep{Luetal09}. Bayesian VCMs combine the flexibility of nonparametric models and the interpretability of parametric models and provide uncertainty estimates in inference and predictions via MCMC draws from the posterior distribution; therefore, they are well-suited for the Bayesian analysis of massive time-series, healthcare, and spatial/spatiotemporal databases.

We focus on Bayesian VCMs in which the varying coefficients are assigned a multivariate GP prior. Without loss of generality, we assume that all functional variables (responses or covariates) are defined on the $d$-dimensional indexing space $[0,1]^d$ ($d \in \NN$). The index is time and $d=1$ in purely times series applications, whereas $d=2$ and the index is a spatial location in purely spatial applications. The two indices are combined in spatiotemporal applications, where $d=3$ and the index is a space-time tuple. More generally, for a sample of indexes $\{u_i:i=1,\ldots,n\}$ ($n \in \NN$) from $[0,1]^d$, we observe the $s_i$-dimensional $i$th response vector $y(u_i) \in \RR^{s_i}$ ($s_i \in \NN$) and the matrix of $i$th covariate functions $X(u_i) \in \RR^{s_i \times p}$, where $p \in \NN$ is the number of covariates. We consider a \vcm\ with the form
\begin{align} \label{eq:mdl1}
y(u_i) = X(u_i) \beta(u_i) + \epsilon(u_i), \quad \epsilon(u_i) \stackrel{\text{ind}}{\sim} N \left(0, \tau^2 I_{s_i} \right), \quad i = 1, \ldots, n,
\end{align}
where $\beta(u) = \{\beta_1(u), \ldots, \beta_p(u)\}^\T \in \RR^p$ for $u \in [0, 1]^d$ is the vector of varying regression coefficients, 0 and $I_s$ are a zero vector and an identity matrix of dimension $s$, and $\epsilon(u_i) \in \RR^{s_i}$ ($i=1, \ldots, n$) are idiosyncratic normal errors. The responses in \eqref{eq:mdl1} are allowed to have different dimensions, but $s_i = s$ for every $i$ ($s \in \NN$) in a typical scientific application.

The VCM setup in \eqref{eq:mdl1} has advantages over its peers in the literature. For example, the varying coefficients $\beta(u)$ provide a more flexible and realistic modeling of responses and predictors with space or space-time indices, so they perform better in practice than fitting deterministic trends in covariates, such as polynomial regression \citep{Geletal03}. There are some methods for modeling the varying coefficients $\beta(u) $ for $u \in [0, 1]^d$ in \eqref{eq:mdl1}, but we use a multivariate GP prior distribution on $\beta(u)$; see Section 2 for a detailed description. In fact, VCMs in the existing literature typically rely on basis expansions using local polynomials, P-splines, and trees for modeling the varying coefficients $\beta(u)$ \citep{LiRac07,Mar10,Beretal19}. The specification of the number of basis functions or the height of the tree and the choice of knots or split locations is usually difficult in practice. In comparison, Bayesian inference using GP priors only requires a tuning-free prior specification for the covariance parameters. Even low-dimensional structures in the data are conveniently modeled using a GP prior projected on a moderately large number of inducing points \citep{QuiRas05}.

While the VCM in \eqref{eq:mdl1} with a GP prior on $\beta(u)$ has such advantages in the modeling of structured data, there are several practical considerations, including inefficient posterior computations, that have severely restricted its application in massive data settings. Posterior sampling involving GPs are already prohibitively slow if the sample size is large. In fact, the cost per MCMC iteration for updating $\beta(u)$ scales as $O(p^2n^2)$ for storage and $O(p^3n^3)$ for computations. As a result, the simple posterior sampling scheme for inference in \eqref{eq:mdl1} using Gibbs and slice sampling as originally proposed in \citet{Geletal03} becomes infeasible in practice. Even with the low-rank GP approximation techniques using $r$ inducing points \citep{QuiRas05,Alvetal12}, one can only reduce the cost per MCMC iteration from $O(n^3p^3)$ to $O(npr^2)$ \citep{AlvLaw11}. Furthermore, $r$ is chosen to be sufficiently large, typically of the order $O\{(\log n)^d \}$, to achieve {satisfactory} approximation accuracy \citep{Buretal19}. The use of deep GP priors in \eqref{eq:mdl1} further worsens the computational burden and cannot be used in practice \citep{DamLaw13,Duvetal14}. Finally, variational inference has been widely used in machine learning for inference in applications based on VCMs with multivariate GPs involving big data, but MCMC-based inference has remained relatively unexplored in this context \citep{AlvLaw11,Alvetal19,Youetal19}. MCMC based inference has the natural advantage of accurately characterizing the uncertainty of inference and prediction in VCMs with functional data, having strong local features. This is crucial in the spatio-temporal application of interest in Section~\ref{sec:real-data-analysis}, which aims at understanding local features in the space-time varying relationship between sea surface temperature and sea salinity for the Atlantic ocean based on large functional data. Our posterior inference algorithm fills this gap, providing a scalable MCMC-based alternative.

Addressing the computational bottlenecks for the VCM in \eqref{eq:mdl1} with a GP prior on $\beta(u)$, {we develop a three stage distributed Bayesian inferential approach for efficient computation with functional response and covariates obtained at a large number of indices. The first stage of the algorithm constructs $k$ subsets by randomly selecting $m$ samples without replacement from the full data, where $k$ is large and posterior computations with $m$ is tractable. The second step obtains $k$ MCMC-based approximations of the full data posterior distribution by fitting the VCM in \eqref{eq:mdl1} with a GP prior on $\beta(u)$ on all the subsets in parallel. This step has two main novelties. First, we compensate for the missing $(1-m/n)$-fraction of the full data in each subset by appropriately modifying the subset likelihood. Second, we reformulate the VCM in \eqref{eq:mdl1} with a GP prior on $\beta(u)$ as a linear mixed-effects model using parameter expansion. This leads to an MCMC algorithm that has closed-form full conditional distributions for all the parameters, except those used for defining the covariance function of the GP prior. We draw these parameters using elliptical slice sampling \cite[ESS,][]{Nisetal14}, which bypasses the proposal tuning problems of Metropolis-Hastings algorithm. The parameter expanded DA with the ESS step constitutes our DA-type algorithm for posterior inference and predictions on the subsets.

The subset posterior computations are tractable because $m \ll n$ and parameter updating is efficient due to the closed-from full conditionals; however, posterior computations a subset condition on $m$ samples only. The third stage of the algorithm develops a combination scheme that aggregates MCMC-based approximations of the true posterior distribution from the $k$ subsets into the AMC posterior, which uses information from all the $n$ samples. This step has several theoretical novelties. First, we identify regularity assumptions under which the AMC posterior distributions of both the varying coefficients and the mean regression function have minimax optimal posterior convergence rates in the $L_2$ norm toward their truth. Development of such guarantees in VCMs with multivariate latent GPs remains an open problem since their proposal in \citet{Geletal03}. Second, our results provide quantification on the orders of the subset size $m$, the number of subsets $k$, and modification of the subset likelihood according to the underlying smoothness of varying coefficients. Finally, our theory only requires a weak condition on the combination scheme, so it encompasses a few existing combination methods, including the AMC posterior proposed in this paper as well as the double parallel Monte Carlo \citep[DPMC,][]{XueLia19}, Wasserstein posterior \citep[WASP,][]{Srietal15}, and posterior interval estimation \citep[PIE,][]{Lietal17} algorithms. The minimax optimality of the AMC posterior distribution implies that it can be used for principled Bayesian inference in massive data settings with very large $n$ if $m$ and $k$ are chosen appropriately.

\subsection{Related Work}

The theoretical and computational properties of frequentist estimation methods for VCMs have been studied extensively. The theoretical results focus mainly on VCMs that use local polynomial smoothing, regularized basis expansions, and boosted trees \citep{HasTib93,FanZha99,huang2002varying,ZhoHoo19}; see \citet{Paretal15} for a recent review. The software for fitting VCMs is also well-developed \citep{Woo17}. On the other hand, Bayesian VCMs have been widely applied to different types of data \citep{Geletal03,Baketal15,Ham15,Datetal16}, but the literature on their theoretical properties is sparsely populated. Recently, \citet{Baietal19} have studied the theoretical properties of Bayesian VCMs based on regularized basis expansions; however, their model is different from the VCM with multivariate GPs considered in this paper, and their main focus is on the high dimensional variable selection problem, which is essentially different from our focus on applications with massive $n$.  Furthermore, frequentist properties of the posterior distribution of regression function obtained using a univariate GP prior are known \citep{VarZan11}, but their extensions to a  multivariate GP prior, similar to the one used for Bayesian inference in \eqref{eq:mdl1}, are non-trivial and have not been studied.

Furthermore, Bayesian VCMs with multivariate response functions have not been studied extensively in the literature. There are some extensions of factor models based on independent GP priors that are used for modeling multivariate responses. One such example is a spatial factor model \citep{RenBan13} that is defined as
\begin{align}
  \label{eq:fact-mdl1}
  y(u_i) = \beta(u_i) + \epsilon(u_i), \quad \beta(u_i) = L \nu(u_i), \quad y(u_i) \in \RR^s, \; s \in \NN, \quad u_i \in [0, 1]^2,
\end{align}
for $i = 1, \ldots, n$, where $L$ is a $p$-by-$q$ factor loading matrix and $\nu(\cdot) = \{\nu_1(\cdot), \ldots, \nu_q(\cdot)\}^\T$ is a vector of $q$ spatial factors, all following mutually independent  univariate GPs $\nu_1(\cdot), \ldots, \nu_q(\cdot)$. \citet{GuShe20} and \citet{RenBan13} also specify identifiability constraints on $L$ in \eqref{eq:fact-mdl1} for valid frequentist estimation and Bayesian inference using MCMC, respectively. Compared to \eqref{eq:mdl1}, every $y(u_i)$ in (\ref{eq:fact-mdl1}) has the same dimension and the covariate matrix $X(u_i)$ is unobserved due to the unsupervised nature of the model.  {While we also propose a very similar formulation of the varying coefficients $\beta(u)$ using the linear model of co-regionalization (LMC) approach, our main focus is posterior inference on the regression coefficient $\beta(u)$ and prediction of $y(u)$ in \eqref{eq:mdl1}, which does not require inference on $L$ or $\nu(\cdot)$.

We now turn our focus to distributed Bayesian inference in  \eqref{eq:mdl1}. The strategy of modifying the subset likelihoods for obtaining better uncertainty characterization in distributed Bayesian inference for parametric models has been discussed in \citet{Minetal17}. Each subset contains only $m/n$-fraction of the full data, so the posterior distribution computed from the usual likelihood overestimates the uncertainty relative to the true posterior distribution. Thus, the modification of the subset likelihood is essential for accurate uncertainty quantification in parametric models \citep{Minetal17}; however, the subset likelihood modification strategy for parametric models cannot be straightforwardly applied for Bayesian VCMs, due to the lack of any supporting theoretical result. One of our main contributions is to identify the subset likelihood modification and to justify it through rigorous theoretical results for VCMs; see Sections~\ref{sec:sampling-step} and \ref{sec:theory}. If we use the asymptotic posterior $L_2$-risk of a combined posterior distribution for quantifying its performance, then the likelihood modifications required for asymptotic optimality are different in the parametric models and VCMs based on an ``appropriately tuned'' multivariate GP prior.

The AMC algorithm belongs to the class of divide-and-conquer (or distributed) methods for Bayesian inference. These methods have been studied extensively for scalable Bayesian inference in parametric models \citep{Scoetal16, Entetal17,Minetal17,Lietal17,Srietal18,XueLia19,Joretal19} and nonparametric regression using univariate GP priors \citep{Zhaetal15,CheSha17,ShaChe19,SzaVan19,SzaVan20,ZhaWil19,guhaniyogi2017divide}. Unfortunately, the literature fails to address distributed Bayesian inference in \eqref{eq:mdl1} using multivariate GP priors, which is our main focus. All these methods consist of three main steps: dividing the massive data set into smaller computationally manageable subsets, performing statistical estimation on the subsets in parallel, and combining the subset estimates into a global estimate, which is used as an alternative to the true posterior distribution. Existing distributed Bayesian inference methods differ mainly in the third step that computes the global estimate. Given some minimal requirements are met in the combination step, our theoretical results can be used to obtain posterior convergence rates for any of these combination schemes used for distributed Bayesian inference in \eqref{eq:mdl1}. We demonstrate that such requirements are indeed met by the combination schemes used in the AMC, DPMC, PIE, and WASP algorithms.

\section{Model Setup and Prior Specification}
\label{sec:vary-coef-full}

In this section, we describe the varying-coefficient model setup and its equivalent formulation as a linear mixed-effects model. We then provide the prior specification and outline the main data augmentation algorithm that is used to fit the model.

\subsection{Model Reformulation}
\label{subsec:modelreform}

Consider a general VCM setup based on \eqref{eq:mdl1} with $p$ predictors out of which $q$ predictors have varying coefficients such that $p\geq q$. Without loss of generality, assume that the first $q$ predictors have varying coefficients, so that $\beta(u) = \{\beta_{\text{va}}(u), \beta_{\text{nv}}\}^\T$ in \eqref{eq:mdl1}, where $\beta_{\text{va}}(u) \in \RR^q$ for every $u \in [0,1]^d$ and  $\beta_{\text{nv}} \in \RR^{p-q}$ are the varying and non-varying coefficients blocks. For performing Bayesian inference on $\beta(\cdot)$, a typical strategy is to assign multivariate GP prior and Gaussian prior distributions on $\beta_{\text{va}}(\cdot)$ and $\beta_{\text{nv}}$, respectively, and obtain MCMC draws from the posterior distribution of $\beta(\cdot)$ using \eqref{eq:mdl1}.

The most important part of the prior specification is to choose a cross-covariance function for the multivariate GP prior on $\beta_{\text{va}}(\cdot)$ that is flexible and leads to simple posterior computations. Versatile constructions exist for specifying the cross covariance of $\beta_{\text{va}}(u)$ \citep{GasCoh99,MajGel07,Wac06,Zha07,genton2015cross,bourotte2016flexible}. We, however, adopt the LMC technique \citep{Alvetal12} for inducing correlation among the components of $\beta_{\text{va}}(u)$ due to its simplicity and relatively efficient computation. Under the LMC framework, we set $\beta_{\text{va}}(\cdot) = \alpha_{\text{va}} + \Gamma \nu(\cdot)$, where $\alpha_{\text{va}} \in \RR^q$, $\Gamma \in \RR^{q \times q}$, and $\nu(u)=\{\nu_1(u), \ldots ,\nu_q(u)\}^\T$ is a vector of $q$ independent GPs indexed by $[0,1]^d$ with mean functions 0 and correlation functions $\rho_1(\cdot,\cdot), \ldots, \rho_q(\cdot,\cdot)$ with parameters $\theta_1, \ldots, \theta_q$, respectively. The independent GP priors on $\nu_1(\cdot), \ldots ,\nu_q(\cdot)$ induces a multivariate GP prior on $\beta_{\text{va}}(\cdot)$. Specifically, given  $\alpha_{\text{va}}$, $\Gamma$, and  $\theta_1, \ldots, \theta_q$, $\beta_{\text{va}}(\cdot) = \alpha_{\text{va}}+\Gamma \nu(\cdot)$ is a $q$-variate GP with mean function $\alpha_{\text{va}}$ and covariance function $C(u, u')$ defined as
\begin{equation}\label{eq:cor-fun}
   C(u, u') = \Cov\{\Gamma \nu(u), \Gamma \nu(u')\}  = \sum_{a=1}^{q} \Gamma_a \rho_a(u, u') \Gamma_a^\T = \sum_{a=1}^{q} \Gamma_a \{R(u, u')\}_{aa} \Gamma_a^\T ,
\end{equation}
where $u, u' \in [0, 1]^d$, $\Gamma_a$ is the $a$th column of $\Gamma$, and $R(u, u')= \diag\{\rho_1(u, u'), \ldots, \rho_q(u, u')\}$ is a $q$-by-$q$ diagonal matrix of correlations determined by $\theta^\T = (\theta_1^\T, \ldots, \theta^\T_q)$.

We now reformulate the VCM in \eqref{eq:mdl1} with a GP prior imposed on $\beta_{\text{va}}(\cdot)$ using the LMC technique as linear mixed-effects model. Define $\alpha=(\alpha_{\text{va}},\beta_{\text{nv}})^\T \in \RR^p$ and $Z(u_i) \in \RR^{s_i \times q}$ to be the matrix that includes the first $q$ columns of $X(u_i)$ ($i=1, \ldots, n$). Reformulate \eqref{eq:mdl1} as
\begin{align} \label{eq:mdl11}
y(u_i) = X(u_i) \alpha + Z(u_i) \Gamma \nu(u_i) + \epsilon(u_i), \quad \nu(\cdot) \sim \text{GP}\{0, R(\cdot, \cdot)\}, \quad i = 1, \ldots, n,
\end{align}
where $R(\cdot, \cdot) = \diag\{\rho_1(\cdot, \cdot), \ldots, \rho_q(\cdot, \cdot)\}$ is the correlation ``function'' for $\nu(\cdot)$. The models in \eqref{eq:mdl1} and \eqref{eq:mdl11} are equivalent if we let $\beta_{\text{va}}(u) = \alpha_{\text{va}} + \Gamma \nu(u)$ for all $u\in [0,1]^d$. The parameters $\alpha$ and $\Gamma$ in \eqref{eq:mdl11} cannot be estimated uniquely from the data $\{y(u_i),X(u_i):i=1,\ldots,n\}$ {but the vector $\{\beta(u_1),\ldots,\beta(u_n)\}$ is still estimable if the design matrix formed by $\{X(u_i):i=1,\ldots,n\}$ as the row blocks is of full column rank. The prior distributions on the unknown parameters $\alpha, \Gamma, \tau^2, \theta$ are spelled out in Section \ref{subsec:prior}.

Many widely used models are obtained as special cases of \eqref{eq:mdl11}. If $\rho_a (u, u') = 1_{u = u'} $ for every $a$, where $1_{u = u'}$ equals 1 if $u=u'$ and 0 otherwise, then we recover the linear-mixed effects model using \eqref{eq:mdl11}, where $\Gamma \Gamma^\T$ equals the covariance matrix of the random effects. If $s_i=1$, $p=q$, $\Gamma $ is a diagonal matrix, and $X(u_i) = Z(u_i)$, then \eqref{eq:mdl1} reduces to
\begin{align} \label{eq:subm1}
y(u_i) = X(u_i) \{\alpha + \Gamma \nu(u_i)\} + \epsilon(u_i) \equiv X(u_i) \beta(u_i) + \epsilon(u_i), \quad i = 1, \ldots, n,
\end{align}
where $\alpha$ and $\Gamma \nu(\cdot)$ model the global and local effects, respectively, the diagonal entries of $\Gamma$ determine the scale of local effects, and $\beta(\cdot)$ is the $p$-by-1 varying coefficients vector. The spatiotemporal varying coefficient model is a special case of \eqref{eq:subm1} when $u \in [0, 1]^3$ \citep{Geletal03,gelfand2010multivariate}. Finally, assuming $u$ to be the time domain in \eqref{eq:subm1} yields a regression model for longitudinal data analysis.

\subsection{Prior Specification}
\label{subsec:prior}

The parameters $(\alpha, \Gamma, \tau^2)$ are jointly assigned a noninformative prior with density  $p(\alpha,\Gamma,\tau^2)\propto 1 / \tau^2$. If $\gamma$ represents the $q^2$-dimensional vector formed by stacking the columns of $\Gamma$, then this prior is a limiting case of the normal-inverse-gamma prior distribution on $\{(\alpha, \gamma), \tau^2\}$, where $(p + q^2)$-variate normal prior distribution is assigned on $(\alpha, \gamma)$.  We are also not concerned with the identifiability of $\Gamma$ or $\alpha$ since they are intermediate latent variables enabling efficient estimation of $\beta_{\text{va}}(u)$ for every $u \in [0, 1]^d$.

As far as the choice of $\rho_a(\cdot,\cdot)$ ($a=1,..,q$) is concerned, two types of correlation functions are used in this paper. The first one is the exponential correlation function defined as $\rho_a(u, u') = e^{-\phi_a \| u - u ' \|_2 }$ for any $u, u' \in [0, 1]^d$, where $\phi_a > 0$, $\| \cdot \|_2$ is the Euclidean norm, and $\theta_a = \{\phi_a\}$ ($a=1, \ldots, q$). We also use Gneiting's correlation function for varying coefficient modeling of spatiotemporal data presented in Section~\ref{sec:real-data-analysis}, which is defined as
\begin{align}
    \label{eq:corr-fun-1}
    \rho_a(u , u') = \frac{1}{(\psi_a |t - t'|^2 + 1)^{\kappa_a}}
    e^{- \frac{\phi_a \|h - h' \|_2} {(\psi_a |t - t'|^2 + 1)^{\kappa_a/2}}}, \quad u, u' \in [0, 1]^3,
\end{align}
where $u=(h, t), u'=(h', t')$ are space-time tuples, $h, h' \in[0,1]^2$, $t, t' \in[0,1]$, $\phi_a > 0$, $\psi_a > 0$, $\kappa_a \in [0, 1]$, and $\theta_a = (\phi_a, \psi_a, \kappa_a)^\T$ \citep{gneiting2002nonseparable}. For the exponential correlation function, we put a Uniform($\underline c_{0a}, \overline c_{0a}$) prior on $\phi_a$. The parameters $\phi_a$, $\psi_a$, and $\kappa_a$ are assigned Uniform($\underline c_{1a}, \overline c_{1a}$),  Uniform($\underline c_{2a}, \overline c_{2a}$), and  Uniform($\underline c_{3a}, \overline c_{3a}$) priors respectively for the Gneting's correlation function. The parameters for the uniform priors satisfy $0 < \underline c_{ia} < \overline c_{ia} $ for $i=0,1,2$ and $0 < \underline c_{3a} < \overline c_{3a} \leq 1$.

\subsection{The DA-type Algorithm}
\label{da-type-algo}

The DA-type algorithm for posterior inference on $\beta(\cdot), \tau^2$ and prediction of $y(\cdot)$ has six parts. Let $\Ucal^*$ be a given subset of $[0, 1]^d$ where the draws of $\beta(\cdot)$  and $y(\cdot)$ are required, $\Dcal$ be the training data, and $\nu_n=\{\nu(u_1), \ldots, \nu(u_n)\}$, where $\nu(u_i)$ is defined in \eqref{eq:mdl11}. The first part of the DA-type algorithm is the Imputation (I) step that draws $\nu_n$ given $\Dcal$ and $(\alpha, \Gamma, \tau^2, \theta)$. The second part of the DA-type algorithm is the Prediction (P) step that has five sub parts. It uses the $\nu_n$ to draw $(\alpha, \Gamma, \tau^2, \theta)$ given $\Dcal$ and $\{\beta(u^*), y(u^*): u^* \in \Ucal^*\}$ given $(\alpha, \Gamma, \tau^2, \theta)$.} The I and P steps are repeated until convergence to the stationary distribution of the Markov chain for $(\tau^2, \{\beta(u^*), y(u^*): u^* \in \Ucal^*\})$; see Appendix \ref{da-full-deriv} for derivation of the six parts and their analytic forms.

This DA-type algorithm is slow in moderately large data sets. The computational complexity of the I step is $O(n^3p^3)$ if we update multivariate GPs. Low rank GP methods provide some computational relief, though the computational gain is not substantial if one needs to maintain the inferential accuracy. The sparse iterative methods for sampling from GPs lead to only marginal improvements because the number of iterations have to be relatively large for guaranteeing accurate approximation \citep{ChoSaa14}. Due to the slow I step, the DA-type algorithm using the full data is extremely inefficient in applications with a large $n$. The next section presents an extension of the DA-type algorithm using divide-and-conquer technique that overcomes these inefficiencies while retaining its simplicity and numerical stability.

\section{Distributed Varying Coefficient Modeling Using a GP Prior}
\label{sec:algorithm}

Our distributed model fitting of the Bayesian VCM consists of three steps described below.

\subsection{First Step: Constructing Training Data Subsets}
\label{sec:partitioning-step}

The first step of the distributed extension of the DA-type algorithm in Section \ref{da-type-algo} constructs $k$ subsets from the training data. The default scheme for constructing subsets is to randomly sub-sample without replacement from the training data, ensuring that each subset provides a reliable representation of the full data and all observations specific to a sample are on the same subset. The size of a subset $m$ is set to be moderately large so that $p,q \ll m$ and posterior computations are efficient on any subset. The ``optimal choice" of $k$ depends on the smoothness of the regression function, which we study in Section \ref{sec:theory}. Let $\Dcal_j$ be the training data on subset $j$ ($j=1,\ldots,k$), $u_{ji}$ be the $i$th  index in subset $j$, and $y(u_{ji})$, $X(u_{ji})$, $Z(u_{ji})$ be the corresponding observations with dimensions $s_{ji}$, $s_{ji}$-by-$p$, $s_{ji}$-by-$q$, for $i=1, \ldots, m$. Similarly, the subset $j$ versions of parameters $(\beta, \alpha, \Gamma, \tau^2, \theta)$ are denoted by $(\beta_j, \alpha_j, \Gamma_j, \tau^2_j, \theta_j)$. The correlation function $\rho_{ja}$ equals $\rho_a$ but replaces $\theta_a$ by $\theta_{ja}$.  The GP with with correlation function $\rho_{ja}$ is denoted as $\tilde \nu_{ja}(\cdot)$, and $\tilde \nu_j(\cdot) = \{\tilde \nu_{j1}(\cdot), \ldots, \tilde \nu_{jq}(\cdot)\}^{\T}$.

The VCM in \eqref{eq:mdl11} has a natural extension to subset $j$. For $i=1, \ldots, m$,
\begin{align} \label{eq:cor-sparse-sub}
& y(u_{ji}) = X(u_{ji}) \alpha_j + Z(u_{ji}) \Gamma_j \tilde \nu_j(u_{ji}) + \epsilon(u_{ji}), \nonumber\\
& \tilde \nu_{ja}(\cdot) \sim \text{GP}\{0, \rho_{ja}(\cdot, \cdot)\}, ~~  a=1,\ldots,q,\quad
\epsilon(u_{ji}) \sim N(0, \tau^2_j),
\end{align}
which reduces to the subset $j$ extension of \eqref{eq:mdl1} if $\{\beta_j(u)\}_{\text{va}} = (\alpha_j)_{\text{va}} + \Gamma_j \tilde \nu_j(u)$, $u \in [0,1]^d$, where $\beta_j(u) = [\{\beta_j(u)\}_{\text{va}}, (\beta_j)_{\text{nv}}]$ and $\alpha_j = \{(\alpha_j)_{\text{va}}, (\beta_j)_{\text{nv}}\}$ are represented in terms of their varying and non-varying coefficients blocks. The prior distributions for $(\alpha_j, \Gamma_j, \tau^2_j)$ and $\theta_j$ in \eqref{eq:cor-sparse-sub} are the same as defined in Section \ref{subsec:prior} for $(\alpha, \Gamma, \tau^2)$ and $\theta$, respectively. 
If we obtain MCMC draws of the parameters and predictions using the likelihood in \eqref{eq:cor-sparse-sub} on each subset directly, then we condition on an $(m/n)$-fraction of the full data, resulting in wider credible intervals for parameters than those obtained using the full data posterior distribution. The next step fixes this problem by using a modified likelihood based on \eqref{eq:cor-sparse-sub} that compensates for the missing  $(1 - m/n)$-fraction of the full data.

\subsection{Second Step: Posterior Sampling on the Subsets}
\label{sec:sampling-step}

We now consider the inference on each subset using the DA-type algorithm based on Section \ref{da-type-algo} and \eqref{eq:cor-sparse-sub}.
The GP realizations $\tilde \nu_{j}(u_{j1}), \ldots, \tilde \nu_{j}(u_{jm})$ in \eqref{eq:cor-sparse-sub} are looked upon as the ``missing'' data, and marginalizing over them recovers the subset $j$ version of \eqref{eq:mdl1} with a GP prior on $[\{\beta_j(u_{j1})\}_{\text{va}}, \ldots, \{\beta_j(u_{jm})\}_{\text{va}}]$.
For $a=1,\ldots,q$, we let $\tilde s_j=\sum_{i=1}^m s_{ji}$ and define
\begin{align}\label{eq:tilde.nu.Z}
& \tilde \nu_{ja} = \{\tilde \nu_{a}(u_{j1}), \ldots, \tilde \nu_{a}(u_{jm})\}^\T \in \RR^{m}, \quad \tilde \nu_{j} = \left(\tilde \nu_{j1}^\T,\ldots,\tilde \nu_{jq}^\T \} \right)^\T \in \RR^{mq}, \nonumber \\
& \Gamma_j = (\Gamma_{j1},\ldots,\Gamma_{jq}), ~~ \Gamma_{j1},\ldots,\Gamma_{jq} \in \RR^q, \nonumber \\
& \tilde Z_{ja} = \diag\left\{Z(u_{j1}) \Gamma_{ja}, \ldots,  Z(u_{jm}) \Gamma_{ja} \right\}, \nonumber \\
&y_j = \left\{y(u_{j1})^\T, \ldots, y(u_{jm})^\T \right\}^\T, ~~ X_j = \left\{X(u_{j1})^\T, \ldots, X(u_{jm})^\T\right\}^\T .
\end{align}
If $\tilde \nu_{j}$ is known, then the full conditional for drawing $(\alpha_j, \Gamma_j, \tau_j^2)$ given $\tilde \nu_{j}$, $\Dcal_j$ is available in closed-from; therefore, $\tilde \nu_{j}$ is an auxiliary variable that simplifies the forms of the full conditionals if it is known and its marginalization preserves the Bayesian VCM with a GP prior on $\{\beta_j(\cdot)\}_{\text{nv}}$.

The I step of the DA-type algorithm on subset $j$ uses this property of $\tilde \nu_{j}$ for simplifying the form of the modified likelihood. Assume that $(\alpha_j, \Gamma_j, \tau_j^2, \theta_j)$ are given. Let $p(\tilde \nu_{j} \mid \Dcal_j, \alpha_j, \Gamma_j, \tau_j^2, \theta_j)$ denote the conditional density of $\tilde \nu_{j}$ given $\Dcal_j$ and $(\alpha_j, \Gamma_j, \tau_j^2, \theta_j)$ based on \eqref{eq:cor-sparse-sub}. Then, the I step
\begin{enumerate}[leftmargin=5mm]
\item[(a)] draws $\tilde \nu_{j} $ given $\Dcal_{j}$ and $(\alpha_j, \Gamma_j, \tau^2_j, \theta_j)$ from $N_{mq}(\mu_{\tilde \nu_j}, \Sigma_{\tilde \nu_j})$, where $\mu_{\tilde \nu_j}$ and $\Sigma_{\tilde \nu_j}$ are defined in terms of their blocks as
\begin{align}  \label{eq:istep-mu-sig}
(\mu_{\tilde \nu_j})_{a} &=  R^\T_{ja} \tilde Z_{ja}^\T \left( \sum_{c=1}^q \tilde Z_{jc} R_{jc} \tilde Z^\T_{jc}  + \tau_j^2 I_{\tilde s_j}  \right)^{-1} (y_j - X_j \alpha_j), \quad a = 1, \ldots, q, \nonumber \\
(\Sigma_{\tilde \nu_j})_{aa} &=  R_{ja} - R^\T_{ja} \tilde Z_{ja}^\T \left( \sum_{c=1}^q \tilde Z_{jc} R_{jc} \tilde Z^\T_{jc}  + \tau_j^2 I_{\tilde s_j}  \right)^{-1} \tilde Z_{ja} R_{ja}, \quad a = 1, \ldots, q, \nonumber \\
(\Sigma_{\tilde \nu_j})_{ab} &=  - R^\T_{ja} \tilde Z_{ja}^\T \left( \sum_{c=1}^q \tilde Z_{jc} R_{jc} \tilde Z^\T_{jc}  + \tau_j^2 I_{\tilde s_j}  \right)^{-1} \tilde Z_{jb} R_{jb}, \quad a \neq b \in \{1, \ldots, q\}, \nonumber \\
(R_{ja})_{ii'} &= \rho_{ja}(u_{ji}, u_{ji'}), \quad i, i' = 1, \ldots, m, \quad  a = 1, \ldots, q.
\end{align}
\end{enumerate}
If we substitute the I step draw of $\tilde \nu_{j}$ in \eqref{eq:cor-sparse-sub} and compute the likelihood of $(\alpha_j, \Gamma_j, \tau_j^2, \theta_j)$ given $(\Dcal_j, \tilde \nu_{j})$, then this is equivalent to computing the likelihood of  $(\alpha_j, \Gamma_j, \tau_j^2, \theta_j)$ after marginalizing over $\tilde \nu_{j}$ in \eqref{eq:cor-sparse-sub} using Monte Carlo. Denote this Monte Carlo based likelihood as $L_j$, and we use it as the likelihood of $(\alpha_j, \Gamma_j, \tau_j^2, \theta_j)$ given $\Dcal_j$ in the P step.

The subset $j$ draws $(\alpha_j, \Gamma_j, \tau_j^2, \theta_j)$ given $\Dcal_j$ using a modified version of  $L_j$. Since $\Dcal_j$ contains an $(m/n)$-fraction of the full data, we raise $L_j$ to a power of $\delta_n$, where $\delta_n$ is a deterministic sequence dependent on $n$. Let $L_j^{\delta_n}$ be this modified likelihood, and the modification is equivalent to replicating $\Dcal_j$ for $\delta_n$-times. The power $\delta_n$ is chosen such that $L_j^{\delta_n}$ compensates for the missing $(1- m/n)$-fraction of the full data on subset $j$. The modified posterior density for drawing  $(\alpha_j, \Gamma_j, \tau_j^2, \theta_j)$ given $\Dcal_j$ is defined as
\begin{align}\label{eq:fullsub-j}
  \pi_{m} (\alpha_j, \Gamma_j, \tau_j^2, \theta_j \mid \Dcal_j) = \frac{L_j^{\delta_n} \; p(\alpha_j, \Gamma_j, \tau^2_j) \; p(\theta_j)} {\int L_j^{\delta_n} \; p(\alpha_j, \Gamma_j, \tau^2_j) \; p(\theta_j) \; d\alpha_j\, d\Gamma_j \, d\tau^2_j \, d\theta_j}, \quad j = 1,\ldots, k,
\end{align}
where the denominator is finite due to the choice of prior distributions. The method of raising subset likelihoods to a power is known as the \emph{stochastic approximation} \citep{Minetal14}. We choose $\delta_n=n/m$ following the same choice in parametric models \citep{Minetal14,Entetal17,Lietal17,Srietal18}, which is equivalent to replicating the subset data for $n/m$ times, such that the subset posterior variances of parameters are comparable to the full data posterior variance. The theoretical impact from $\delta_n$ will be further discussed in Section \ref{sec:theory}.

The P step draws $\alpha_j, \Gamma_j, \tau^2_j$, and $\theta_j$  given $(\Dcal_j, \tilde \nu_{j})$ in a sequence of three steps using \eqref{eq:fullsub-j}. Define $\gamma_j = (\Gamma_{j1}^\T, \ldots, \Gamma_{jq}^\T)^\T$, the column-wise vectorization of $\Gamma_j$, $b_j= (\alpha_j^\T, \gamma_j^\T)^\T$, and
\begin{align} \label{eq:xx}
W_j= \left( W_{j1}^\T, \ldots, W_{jm}^\T\right)^\T , \quad  W_{ji} = [X_{ji}  \;\; \{\tilde \nu_{1}(u_{ji}), \ldots, \tilde \nu_{q}(u_{ji})\} \otimes Z_{ji} ] \in \RR^{ s_{ji} \times (p + q^2)},
\end{align}
for $i = 1, \ldots, m$, where $W_j \in \RR^{\tilde s_j \times (p + q^2)}$ and $\otimes$ is the Kronecker product. Then, the P step draws $\alpha_j, \Gamma_j, \tau^2_j$, and $\theta_j$ as follows:
\begin{enumerate}[leftmargin=5mm]
\item[(b)] draw $\tau_j^2$ given $\tilde \nu_{j}$ and $\Dcal_j$ as
  \begin{align}
    \label{eq:12}
    \tau^2_j \sim \frac{\delta_n \left\|y_j - \hat y_j \right\|_2^2 }{\chi^2_{\delta_n \tilde s_j - p - q^2}}, \quad \hat y_j = W_j (W_j^\T  W_j )^{-1} W_j^{\T} y_j,
  \end{align}
where $\chi^2_{\delta_n \tilde s_j - p - q^2}$ is a chi-square random variable with $\delta_n \tilde s_j - p - q^2$ as its degrees of freedom.
\item[(c)] draw $b_j = (\alpha_j^\T, \gamma_j^\T)^\T$ given $\tau^2_j$, $\tilde \nu_{j}$, and $\Dcal_j$ from $N\{(W_j^\T  W_j )^{-1} W_j^{\T} y_j, \tau^2_j (W_j^\T W_j)^{-1}\}$; and
\item[(d)] draw $\theta_{j1}, \ldots, \theta_{jq}$ given $\tilde \nu_{j}$ and $\Dcal_j$ using ESS (Algorithm 1 in \citealt{Nisetal14}) with the modified log-likelihood for $\theta_1, \ldots, \theta_q$ defined as
\begin{align}
  \label{eq:ess-llk}
  \log L(\theta_{j1}, \ldots, \theta_{jq}) = - \frac{\delta_n mq}{2} \log 2 \pi - \frac{\delta_n}{2} \sum_{a=1}^q \log \text{det} (R_{ja})
  - \frac{\delta_n}{2} \sum_{a=1}^q \tilde \nu_{ja}^\T R_{ja} ^{-1} \tilde \nu_{ja},
\end{align}
where $\tilde \nu_{ja}$ is defined in \eqref{eq:tilde.nu.Z}, $R_{ja}$ is a $m$-by-$m$ matrix defined in \eqref{eq:istep-mu-sig} and depends on $\theta_{ja}$. The form of the likelihood of $\theta_{j1}, \ldots, \theta_{jq}$ depends on the correlation functions of the univariate GPs; see Appendix \ref{da-sub-deriv} for the exact details of the likelihood for the two correlation functions used in this paper.
\end{enumerate}

In most applications, the goal is to perform inference on $\beta(u^*)$ and predict $y(u^*)$ for $u^* \in \Ucal^*$, where $\Ucal^*=\{u_1^*,\ldots,u_l^*\}$ is a known subset of $[0,1]^d$, also known as the testing set. This is done by using the parameter draws from parts (b)--(d) in the P step as follows:
\begin{enumerate}[leftmargin=5mm]
\item[(e)] draw $\nu_{ja}^* = \{\nu_{a}(u_1^*), \ldots, \nu_{a}(u_l^*)\}^\T$ given $\tilde \nu_{j}$, $\theta_j$, and $\Dcal_j$ from $N(\mu_{ja}^*, \Sigma_{ja}^*)$, where
  \begin{align}
    \label{eq:s-nu-mu-sig}
    &\mu_{ja}^* = R_{ja*}^\T R_{ja}^{-1} \tilde \nu_{ja}, \quad \Sigma_{ja}^* = R_{ja**} - R^\T_{ja*} R_{ja}^{-1} R_{ja*}, \nonumber\\
    &(R_{ja**})_{i'i''} = \rho_{ja}(u^*_{i'}, u^*_{i''}), \quad
                         (R_{ja*})_{ii'} = \rho_{ja}(u_i, u_{i'}^*),
  \end{align}
  for $a=1, \ldots, q$, $i', i'' = 1, \ldots, l$, and $i = 1, \ldots, m$, and set $\{\beta_j(u^*)\}_{\text{nv}} = (\alpha_j)_{\text{nv}} $ and $\{\beta_j(u^*)\}_{\text{va}} = (\alpha_j)_{\text{va}} + \Gamma_j  \nu_j(u^*)$,  $u^* \in \Ucal^*$; and
\item[(f)] draw $y_j(u^*)$ given $\alpha_j$, $\Gamma_j$, $\tau_j^2$, $X(u^*)$, $\beta(u^*)$ independently from $N(\mu_{y_j}^*, \tau_j^2 I_{s^*})$ for every $u^* \in \Ucal^*$, where $\mu_{y_j}^* = X(u^*) \beta_j(u^*)$ and $s^*$ is the dimension of $y$ at $u^*$.
\end{enumerate}
The I and P steps, including the parts (a)--(f), are run in parallel on the $k$ subsets until convergence of the Markov chain for $(\tau^2_j, \{\beta_j(u^*), y_j(u^*): u^* \in \Ucal^*\})$  to its stationary distribution; see Appendix \ref{da-sub-deriv} for derivation of the six parts and their analytic forms.

The AMC sampler cycles through steps (a)--(f) on subset $j$ to obtain posterior draws of $\beta_j(u^*)$, $\tau_j^2$, and $y_j(u^*)$, $u^* \in \Ucal^*$ ($j=1, \ldots, k$). Let $T$ be the number of post-burnin draws collected on every subset. Denote the parameter and prediction samples obtained from subset $j$ at the $t$th iteration as $\{\beta_{j}^{(t)}(u^*), \tau^{2(t)}_{j}, y_{j}^{(t)}(u^*)\}$ ($t=1, \ldots, T$; $u \in \Ucal^*$), which are called the $j$th \emph{subset posterior} draws. We assume that the marginal $j$th subset posterior draws for $\beta$, $\tau^2$, and $y^*$ follow their invariant distributions denoted as $\Pi_{\beta}(\cdot \mid \Dcal_j)$,  $\Pi_{\tau^2}(\cdot \mid \Dcal_j)$, and $\Pi_{y^*}(\cdot \mid \Dcal_j)$, which are called $j$th \emph{subset posterior} distributions and their densities  are obtained using the joint density in \eqref{eq:fullsub-j}. We develop next an algorithm for combining the collection of $k$ subset posterior draws such that the combined draw follows the AMC posterior distribution that conditions on the full data.

\subsection{Third Step: Aggregation of Subset Posterior Draws}
\label{sec:comb-step}

We aggregate the subset posterior draws for $\beta(\cdot)$, $y(\cdot)$, and $\tau^2$ using centering and scaling operations. Let $\beta_j^{*(t)} = \{\beta_{j}^{(t)}(u^*_1), \ldots, \beta_{j}^{(t)}(u^*_l)\}$, $y_j^{*(t)} = \{y_{j}^{(t)}(u^*_1), \ldots, y_{j}^{(t)}(u^*_l)\}$, and $\log \tau_j^{2 (t)}$ be the $t$th draws for $\beta(\cdot)$, $y(\cdot)$, and $\log \tau^2$ on subset $j$ ($j=1, \ldots, k$). Let $\xi \in \{\beta(\cdot), y(\cdot), \log \tau^2\}$ and  $\xi^{(t)}_j$ be its $t$th draw on subset $j$. Define the  empirical mean vector and covariance matrix of $\xi$ draws on subset $j$ as
\begin{align}
  \label{eq:mean-cov-j}
  \mu_{j \xi} &= \frac{1}{T} \sum_{t=1}^T \xi_{j}^{(t)}, \quad
  \Sigma_{j\xi} = \frac{1}{T} \sum_{t=1}^T \left( \xi_{j}^{(t)} - \mu_{j\xi}\right) \left( \xi_{j}^{(t)} - \mu_{j\xi}\right)^{\T}, \quad j =1, \ldots, k.
\end{align}

We now summarize the algorithm for obtaining draws from the AMC posterior using the subset posterior draws. First, define the combined empirical mean and covariance matrix for $\xi$ draws using the subset posterior empirical means and covariance matrices in \eqref{eq:mean-cov-j}  as
\begin{align}
  \label{eq:mean-cov-comb}
  \mu_{\xi} = \frac{1}{k} \sum_{j=1}^k \mu_{j \xi} , \quad \Sigma_{\xi} = \frac{1}{k} \sum_{j=1}^k \Sigma_{j \xi}.
\end{align}
Second, center and scale the $j$th subset posterior draws of $\xi$ as
\begin{align}
  \label{eq:cent-scale}
  q_{j \xi}^{(t)} = \Sigma_{j \xi}^{-1/2} \left( \xi_{j}^{(t)} - \mu_{j \xi} \right), \quad t=1, \ldots, T; \; j=1, \ldots, k.
\end{align}
Third, rescale and recenter the $\xi$ draws from all the subsets in \eqref{eq:cent-scale} as
\begin{align}
  \label{eq:re-cent-scale}
  \xi_{t'} = \mu_{\xi} + \Sigma_{\xi}^{1/2} q_{j \xi}^{(t)}, \quad t' = t + (j-1)T; \; j=1, \ldots, k,
\end{align}
to obtain $t'$th draws from the AMC posterior distribution of $\xi$. The $\tau^2$ draws are obtained by taking the exponential of draws from the AMC posterior of $\log \tau^2$.

The AMC aggregation algorithm for subset posteriors bears close connections to a few recently devised combination methods such as the Double Parallel Monte Carlo \citep[DPMC,][]{XueLia19} and Wasserstein posterior \citep[WASP,][]{XuSri20}. All three algorithms agree on the combination of the subset posterior means but differ in their approach to combining subset posterior covariance matrices. The scaling and re-scaling steps are absent in DPMC because it relies on the asymptotic normality of subset posterior distributions. On the other hand, the combination algorithms of AMC and WASP have the same three steps, except the former and latter compute the combined covariance matrix as the arithmetic and geometric means of subset posterior covariance matrices. The computation of the geometric mean in WASP requires an iterative algorithm \citep{Alvetal16}, so the AMC algorithm is computationally simpler. Finally, AMC, WASP and DPMC are applicable for aggregating posterior distributions of multivariate quantities, whereas the PIE algorithm is developed for posterior distributions of scalar quantities only. %

An important contribution of this article is to theoretically establish that any aggregation step of subset posteriors constructed in Step 2 with the modified data likelihood that satisfies a simple assumption (see Assumption \ref{combine_assumption} in Section~\ref{sec:theory}) will lead to an optimal estimation of the regression function. We establish that this assumption is satisfied not only by AMC, but also by DPMC, WASP and PIE aggregation methods for divide-and-conquer Bayesian inference mentioned above.

\section{Theoretical Properties of AMC}
\label{sec:theory}

This section derives the posterior convergence rates for the varying coefficients and the mean regression function for the three stage AMC framework under certain regularity assumptions on the smoothness of the latent GPs. This setup allows the study of VCMs with coefficients modeled using GPs with full-rank and low-rank covariance functions from a common framework.

We first make the following assumption on $n, k, m$ and the sampling scheme of index $u$.
\begin{enumerate}[label=(A.\arabic*)]
\item \label{sampling} $c_1n\leq km\leq c_2n$ for some constants $0<c_1\leq 1\leq c_2$. The sampled indices in the full data $\{u_{1},\ldots,u_{n}\}$  and a single testing index $u^*$ are drawn independently from the Lebesgue measure on $[0,1]^d$. The subset indexes $\{u_{j1},\ldots,u_{jm}:~j=1,\ldots,k\}$ are drawn independently without replacement from $\{u_{1},\ldots,u_{n}\}$.
\end{enumerate}

For the VCM in \eqref{eq:mdl11}, we simplify the model setup by first assuming that $p=q$ and $X(\cdot)\equiv Z(\cdot)$; that is, every covariate function has a varying coefficient. The frequentist minimax rates for Gaussian process model with $s>1$ are unknown, so we only consider the case where $s=1$ because it enables direct comparison of our posterior convergence rates with existing frequentist results. A fully Bayesian asymptotic theory involves the full posterior distribution of $\alpha,\Gamma,\tau^2,\theta$ and the latent GP realizations $\nu_1(\cdot),\ldots,\nu_q(\cdot)$; however, $\alpha$ and $\Gamma$ are not identifiable, so their posteriors do not contract to any point mass as the sample size $n$ increases to infinity. Additionally, given that our indexing space $[0,1]^d$ is fixed and bounded, it is known that the length-scale parameters in $\theta$ are also not identifiable in fixed-domain asymptotics \citep{Zha04}. To alleviate the technical difficulties from such non-identifiable parameters, we make the following assumption:
\begin{enumerate}[label=(A.\arabic*)]
\setcounter{enumi}{1}
\item \label{prior_assumption} $s=1$, $p=q$ and $X(\cdot)\equiv Z(\cdot)$. $\alpha$, $\Gamma$, $\tau^2$, and $\theta$ are all fixed at their true values $\alpha_0=0$, $\Gamma_0$, $\tau^2_0$, and $\theta_0$. $\Gamma_0$ is a full-rank $q\times q$ matrix. The observed response function satisfies $y(u)=Z(u) \Gamma_0 \nu_0(u) + \epsilon(u)$, $\EE[\epsilon(u)]=0$, $\var[\epsilon(u)]=\tau^2_0$, for all $u\in [0,1]^d$, where $\nu_0(\cdot)=\{\nu_{01}(\cdot),\ldots,\nu_{0q}(\cdot)\}^\T$ are the true latent functions.
\end{enumerate}
Assuming $\alpha\equiv \alpha_0=0$ is equivalent to assuming that $\alpha$ is fixed at any true value $\alpha_0$, since we can always redefine $y'(u)=y(u)-X(u)\alpha_0$ and call $y'(u)$ the response function. While $\tau^2$ is assumed fixed in Assumption \ref{prior_assumption}, it is possible to generalize our technical proofs such that $\tau^2$ has a prior in a bounded interval $[\underline \tau^2, \overline \tau^2]$ for some constants $0<\underline\tau^2 <\overline \tau^2<\infty$ \citep{VarZan08a}. For our theory on convergence rates, we only require the error to have mean zero and variance $\tau^2_0$, but do not require the true error distribution to be normal. In other words, our convergence theory also works when our model of normal error is misspecified.

We now define some notations for the subset data. For $i=1,\ldots,m$ and $j=1,\ldots,k$, let $y_j = \{y(u_{j1}), \ldots, y(u_{jm})\}^\T$ and $\epsilon_j = \{\epsilon(u_{j1}), \ldots, \epsilon(u_{jm})\}^\T$. For $a=1,\ldots,q$, let $\tilde \nu_{ja}=\{\nu_a(u_{j1}),\ldots,\nu_a(u_{jm})\}^\T$. We have assumed that $\theta$ is known, so $\theta_{ja} = \theta_a = \theta_{0a}$ and $\rho_{a} = \rho_{ja}$ ($a=1, \ldots, q$). Since $s=1$ and $p=q$, $Z(\cdot)$ is a $q$-dimensional row vector of functions. For $a=1,\ldots,q$, let $\Gamma_{0a}$ be the $a$th column of $\Gamma_0$ and $\tilde Z_a(\cdot) = Z(\cdot)\Gamma_0$. Let $\tilde Z(\cdot) = Z(\cdot) \Gamma_0 = \{\tilde Z_1(\cdot),\ldots,\tilde Z_q(\cdot)\}$, which is still a $q$-dimensional row vector of functions, and $\tilde Z_{ja}=\{\tilde Z_{a}(u_{j1}),\ldots, \tilde Z_{a}(u_{jm})\}$ for $a=1,\ldots,q$. With the stochastic approximation described in Section \ref{sec:sampling-step}, our ``working model" of VCM on the $j$th subset data under Assumption \ref{prior_assumption} can be written as
\begin{align}\label{subset_model}
& y_j =  \sum_{a=1}^q \tilde Z_{ja} \tilde \nu_{ja} + \epsilon_j, \quad \epsilon_j\sim N\left(0,\frac{\tau^2_0}{\delta_n}I_m\right), \quad \tilde \nu_{ja} \sim N\left (0, \lambda_n^{-1} R_{ja} \right),
\end{align}
where $R_{ja}$ is a $m$-by-$m$ matrix with entries $(R_{ja})_{ii'}=\rho_{a}(u_{ji},u_{ji'})$, $i,i'\in \{1,\ldots,m\}$, $a=1,\ldots,q$, and $\rho_{a}$ has its parameter $\theta_a$ fixed at $\theta_{0a}$. $\delta_n=n/m$ is the same as used in the stochastic approximation definition in Section \ref{sec:sampling-step}. The stochastic approximation with $\delta_n=n/m$ is crucial for ensuring that the AMC posterior of varying coefficients $\beta(\cdot)$ converges to the truth at a polynomial rate of $n$ rather than $m$, such that the AMC posterior can be a valid approximation to the full data posterior, which converges at a polynomial rate of $n$. In \eqref{subset_model}, we have also added an additional tuning parameter $\lambda_n>0$ that only depends on $n$ and is only used for theory development. The value of $\lambda_n$ helps offering minimax optimal rate and is specified later in Theorem \ref{thm:main}. In practice, we simply set $\lambda_n=1$ that provides a nearly optimal rate, and the model \eqref{subset_model} becomes the same model as the VCM in \eqref{eq:mdl11}. This can also be seen in Theorem~\ref{thm:main} and in the ensuing discussion.

We focus on the posterior convergence behavior of the varying coefficients $\beta(\cdot)=\alpha_0+ \Gamma_0\nu(\cdot)$ towards the truth $\beta_0(\cdot)=\alpha_0+ \Gamma_0\nu_0(\cdot)$, as well as the mean function $w(\cdot) = Z(\cdot) \beta(\cdot)$ towards the truth $w_0(\cdot)=Z(\cdot) \beta_0(\cdot)$. We introduce some concepts for reproducing kernel Hilbert space (RKHS) that will be used for stating the assumptions on $\nu$'s. Let $L_2(\dd u)$ be the class of all square-integrable functions on $[0,1]^d$ with respect to the Lebesgue measure, with the inner product given by $\langle f,g \rangle_{L_2(\dd u)} = \int_{[0,1]^d} f(u)g(u)\dd u $ and the $L_2(\dd u)$-norm given by $\|f\|_2^2 = \langle f,f \rangle_{L_2(\dd u)}$, for any generic $f,g\in L_2(\dd u)$. For the correlation function $\rho_a(\cdot,\cdot)$ with parameters $\theta_{0a}$, we assume that $\sup_{u,u'\in[0,1]^d} \rho_a(u,u')<\infty$  ($a=1,\ldots,q$), which means that all correlation functions are trace class kernels. For each $a=1,\ldots,q$, by the Mercer's theorem, there exists an orthonormal sequence of eigenfunctions $\{\varphi_{ah}\}_{h=1}^{\infty}$ in $L_2(\dd u)$ with eigenvalues $\mu_{a1}\geq \mu_{a2} \geq \ldots \geq 0$, such that $ \int_{[0,1]^d} \rho_a(\cdot,u')\varphi_{ah}(u') \dd u'=\mu_{ah}\varphi_{ah}(\cdot)$ for all $h=1,2,\ldots$, and $\rho_a(u, u') = \sum_{h=1}^{\infty} \mu_{ah} \varphi_{ah}(u) \varphi_{ah}(u')$ for any $u, u' \in [0,1]^d$. The RKHS $\HH_a$ attached to the correlation function $\rho_a$ is the space of all functions $f \in L_2(\dd u)$ such that the $\HH_a$-norm $\|f\|_{\HH_a}^2 = \sum_{h=1}^{\infty} \langle f, \varphi_{ah} \rangle_{L_2(\dd u)}^2 / \mu_{ah}<\infty$, for $a = 1, \ldots, q$.

For two positive sequences $\{a_n\}_{n\geq 1}$ and $\{b_n\}_{n\geq 1}$, the relation $\limsup_{n\to\infty} a_n/b_n\leq c$ for some constant $c>0$ is denoted by $a_n\lesssim b_n$, or $b_n\gtrsim a_n$. If $a_n\lesssim b_n$ and $b_n\lesssim a_n$, then we say that $a_n \asymp b_n$.

We impose the following assumption on the eigenfunctions and eigenvalues of $\rho_a$ ($a=1,\ldots,q$) as well as $Z(\cdot)$:
\begin{enumerate}[label=(A.\arabic*)]
\setcounter{enumi}{2}
\item \label{eigen_assumption}
(i) There exists a constant $C_{\varphi}>0$ such that $|\varphi_{ah}(u)| \leq C_{\varphi}$ for all $u \in [0,1]^d$, $a=1,\ldots,q$, and $h =1,2,\ldots$.\\
(ii) The largest $h$th eigenvalue of $\rho_1,\ldots,\rho_q$, defined by $\mu_{h*}=\max_{a=1,\ldots,q}\mu_{ah}$, satisfies $\mu_{h*} \lesssim h^{-2\vv/d}$ for every $h=1,2,\ldots$ and some constant $\vv > d/2$.
\item \label{w0_assumption}
For $a=1,\ldots,q$, the true latent functions satisfy $\nu_{0a} \in \HH_a$.
\item \label{z_assumption}
$|Z_a(u)|\leq C_Z$ for some finite constant $C_Z>0$ for all $u \in [0,1]^d$ and $a=1,\ldots,q$. Let $\bar h=\lceil n^{3d/(2\vv-d)} \rceil$ with $\vv$ given in Assumption \ref{eigen_assumption}. For any $u \in [0,1]^d$, define the $q\bar h$-variate function
\begin{align*}
W(u)={}& \big\{Z(u)\Gamma_{01}\varphi_{11}(u),\ldots, Z(u)\Gamma_{01}\varphi_{1\bar h}(u),\ldots, \nonumber\\
& ~~Z(u)\Gamma_{0q}\varphi_{q1}(u),\ldots, Z(u)\Gamma_{0q}\varphi_{q\bar h}(u)\big\}^\T \in \RR^{q \bar h},
\end{align*}
and the matrix $\Omega = \EE_{u}\left\{W(u) W(u)^\T\right\} \in \RR^{q\bar h \times q\bar h}$, where $\EE_{u}$ is the expectation with respect to the Lebesgue measure on $[0,1]^d$. Then, the smallest and the largest eigenvalues of $\Omega$ are bounded away from zero and infinity by constants.
\end{enumerate}
Consider the covariance function
$$C(u,u')=\sum_{h=0}^n a_h \cos(h\pi|u-u'|)=\sum_{h=0}^n a_h\big\{\cos(h\pi u)\cos(h\pi u') + \sin(h\pi u)\sin(h\pi u')\big\},$$
for $u,u'\in [0,1]$ with $a_h\geq 0$ for all $h=1,2,\ldots$ and $\sum_{h=1}^{\infty}a_h<\infty$, then by the Mercer's theorem, $\{\varphi_0(u)=1,\varphi_{2h-1}(u)=\cos(h\pi u),\varphi_{2h}(u)=\sin(h\pi u):h=1,2,\ldots\}$ are the eigenfunctions and $\{\mu_0=a_0,\mu_{2h-1}=\mu_{2h}=a_h: h=1,2,\ldots\}$ are the eigenvalues. In this example, Assumption \ref{eigen_assumption} (i) is satisfied since the trigonometric functions are uniformly bounded by 1. The commonly used Mat\'ern covariance function with smoothness parameter $\kappa$ takes the form $C(u,u')=\tfrac{2^{1-\kappa}}{\Gamma(\kappa)}\left(\sqrt{2\kappa}\theta\|u-u'\|\right)^{\kappa} K_{\kappa}\left(\sqrt{2\kappa}\theta\|u-u'\|\right)$ for $u,u'\in [0,1]^d$, where $\Gamma(\cdot)$ is the gamma function and $K_{\kappa}$ is the modified Bessel function of the second kind. Then Assumption \ref{eigen_assumption} (i) is satisfied for Mat\'ern with $d=1$ and $\kappa=1/2$ as the eigenfunctions are again the trigonometric functions as shown in Section 3.4.1 of \citet{VT01}. It is also known in the literature that the decay rate of eigenvalues for the Mat\'ern covariance function on $[0,1]^d$ with smoothness parameter $\vv$ satisfies Assumption \ref{eigen_assumption} (ii) with $\vv=\kappa+d/2$ \citep{Ritetal95,SchWen02,SanSch16}.

Assumption \ref{w0_assumption} assumes the smoothness of the true underlying functions $\nu_{01}(\cdot),\ldots,\nu_{0q}(\cdot)$. Given the eigenvalue condition in Assumption \ref{eigen_assumption}, for any given set of nonzero constants $c_1,\ldots,c_q$, the RKHS attached to the covariance function $\sum_{a=1}^q c_a\rho_a(\cdot,\cdot)$ is norm equivalent to the $\vv$-smooth Sobolev space on $[0,1]^d$.

Assumption \ref{z_assumption} is a technical condition that makes $\nu(\cdot)$ estimable from subset data. Similar conditions have been used in varying-coefficient modeling literature. For example, in the VCMs based on basis expansions where the dimension increases with $n$, the bounded eigenvalue condition in \ref{z_assumption} is comparable to Condition (C1) in \citet{Weietal11} and Assumption (A5) in \citet{Baietal19}, both of which have imposed bounded eigenvalue conditions on the covariance matrices involving the products of regressors and basis functions.

Given the fixed full-rank $\Gamma_0$ as in Assumption \ref{prior_assumption}, combining subset draws of $\beta(\cdot)$ using the method described in Section \ref{sec:comb-step} is equivalent to combining subset draws of $\nu(\cdot)$. For $a=1,\ldots,q$, let $\overline \nu_a(\cdot)$ be a random function drawn from the AMC posterior of $\nu_a(\cdot)$. We need the following assumption for the combination scheme in AMC.
\begin{enumerate}[label=(A.\arabic*)]
\setcounter{enumi}{5}
\item \label{combine_assumption}
For any $u \in [0,1]^d$, for each $a=1,\ldots,q$, the AMC posterior mean and variance of $\overline \nu_a(u)$ satisfy
\begin{align}\label{comb_mean_var}
& \EE_{\overline \nu_a |y,u}\{\overline \nu_a(u)\} = \frac{1}{k} \sum_{j=1}^k \EE_{ \tilde \nu_{ja} |y_j,u_j }\{\nu_{ja}(u)\} + O_p\left(n^{-1/2}\right), \nonumber \\
& {\var}_{\overline \nu_a |y,u} \{\overline \nu_a(u)\} \leq \frac{\overline c}{k} \sum_{j=1}^k {\var}_{\tilde \nu_{ja}|y_j,u_j}\{\nu_{ja}(u)\},
\end{align}
for some constant $\overline c>0$, where $\EE_{\tilde \nu_{ja}|y_j,u_j}$ and ${\var}_{\tilde \nu_{ja}|y_j,u_j}$ are the $j$th subset posterior mean and variance of $\nu_{a}(\cdot)$ given the $j$th subset data from \eqref{subset_model}, $\EE_{\overline \nu_a|y,u}$ and ${\var}_{\overline \nu_a|y,u}$ denote the AMC posterior mean and variance of $\nu_a(\cdot)$ given the full data, and the term $O_p\left(n^{-1/2}\right)$ holds uniformly over all $u\in [0,1]^d$ and all $a=1,\ldots,q$ in the probability of the observed data.
\end{enumerate}

Assumption \ref{combine_assumption} imposes very weak conditions on the combination method for the aggregated Bayesian posterior. It only requires that the AMC posterior mean is roughly unbiased compared to the average of subset posterior means, and the AMC posterior variance to be upper bounded by the average of subset posterior variances. These relations can be verified for many existing combination methods in the divide-and-conquer Bayes literature for parametric models. In particular, the $O_p(n^{-1/2})$ term in \eqref{comb_mean_var} is exactly zero in parametric models for the PIE algorithm \citep{Lietal17}, the Wasserstein posterior \citep{XuSri20}, the DPMC posterior \citep{XueLia19}, and our proposed AMC method. Furthermore, if the model is parametric, then in all four methods, the subset posterior and the combined posterior variances satisfy ${\var}_{\overline \nu_a |y,u} \{\overline \nu_a(u)\} = n^{-1} \Ical_0^{-1} + o_p(n^{-1})$, where $\Ical_0$ is a fixed information matrix that does not depend on $n$; see the theory in \citet{Lietal17}, \citet{XuSri20} and \citet{XueLia19}. Therefore, for parametric models, the combined posterior using either of these methods can recover the exact asymptotic variance of the true posterior distribution after setting $\overline c = 1$ and changing the inequality to equality in \eqref{comb_mean_var}. Based on these observations about the combined posterior means and variances, it is expected that in the VCM setup, the rates of convergence of the AMC, Wasserstein, and DPMC posterior distributions to the true posterior distribution should be similar to each other.

The following theorem is our main result on the convergence rate in $L_2$ norm of the AMC posterior distribution for the varying coefficients and the mean regression function; see Appendix \ref{app:proof} for the proof. Although the convergence results are presented for the AMC posterior distribution, they are not unique to the AMC posterior. Rather, they hold for any other combined posterior distribution that is built under the three step framework described in Sections~\ref{sec:partitioning-step}, \ref{sec:sampling-step} and \ref{sec:comb-step}, with the combination step following Assumption~\ref{combine_assumption}, including the Wasserstein posterior and the DPMC posterior.

\begin{theorem} \label{thm:main}
Suppose that Assumptions \ref{sampling}--\ref{combine_assumption} hold for the VCM in \eqref{subset_model}. Let $\overline\beta(\cdot)= \alpha_0+\Gamma_0\overline \nu(\cdot)$ and $\overline w(\cdot) = Z(\cdot) \overline \beta(\cdot)$, where $\overline \nu(\cdot)=\{\overline \nu_1(\cdot),\ldots,\overline \nu_q(\cdot)\}^\T$ is a $q$-variate random function drawn from the AMC posterior of $\nu(\cdot)$. Let $\EE_{u^*}$, $\EE_{y,u}$, and $\EE_{\overline \beta \mid y,u}$ be the expectations with respect to the distribution of testing point $u^*$, the true data generating distribution (with randomness from $y_{ji}$ and $u_{ji}$, $j=1,\ldots,k$ and $i=1,\ldots,m$), and the AMC posterior distribution of the varying coefficients $\overline \beta(\cdot)$ given the full data.
\vspace{2mm}

\noindent (i) If $\lambda_n = 1$ and $m \gtrsim n^{(d/\vv)+\eta}$ for some constant $\eta \in \left(0, \tfrac{\vv-d}{\vv}\right]$, then the AMC posterior satisfies
\begin{align*}
& \EE_{u^*}\EE_{y,u} \EE_{\overline \beta \mid y,u} \left\| \overline \beta(u^*) - \beta_0(u^*) \right \|_2^2 \lesssim n^{-(2\vv-d)/(2\vv)}, \\
\text{ and } & \EE_{u^*}\EE_{y,u} \EE_{\overline \beta \mid y,u} \left\{ \overline w(u^*) - w_0(u^*) \right \}^2 \lesssim n^{-(2\vv-d)/(2\vv)},
\end{align*}
\noindent (ii) If $\lambda_n \asymp n^{d/(2\vv+d)}$ and $m \gtrsim n^{2d/(2\vv+d)+\eta}$ for some constant $\eta \in \left(0, \tfrac{2\vv-d}{2\vv+d}\right]$, then the AMC posterior satisfies
\begin{align*}
& \EE_{u^*}\EE_{y,u} \EE_{\overline \beta \mid y,u} \left\|\overline \beta(u^*) - \beta_0(u^*) \right \|_2^2 \lesssim n^{-2\vv/(2\vv+d)}, \\
\text{ and } & \EE_{u^*}\EE_{y,u} \EE_{\overline \beta \mid y,u} \left\{ \overline w(u^*) - w_0(u^*) \right \}^2 \lesssim n^{-2\vv/(2\vv+d)}.
\end{align*}
\end{theorem}

Theorem \ref{thm:main} gives the upper bounds for the posterior convergence rates of both the $q$-dimensional varying coefficients and the mean regression function in the $L_2$ norm. To the best of our knowledge, such convergence result for varying coefficients $\beta(\cdot)$ is new in the literature of Bayesian varying-coefficient models with multivariate latent GPs. Theorem \ref{thm:main} provides theoretical guarantees for the distributed extension of VCM proposed in \citet{Geletal03}, which is developed in Section \ref{sec:algorithm} and scales to massive data settings.

When the tuning parameter $\lambda_n$ is chosen appropriately as in Theorem \ref{thm:main} (ii), the AMC posterior of the varying coefficients $\beta(\cdot)$ converges to the underlying truth in the $L_2$ norm at the rate $n^{-\vv/(2\vv+d)}$. This rate is known as the minimax optimal posterior convergence rate in the $L_2$ norm for the simple Gaussian process regression \citep{VarZan11}. In the extreme case of $m=n$ and $k=1$, Theorem \ref{thm:main} also implies that the rate $n^{-\vv/(2\vv+d)}$ is the convergence rate of the full data posterior distribution of $\beta(\cdot)$; therefore, we have shown that the AMC posterior from our distributed Bayesian method can quantify the posterior uncertainty in the same order as the full data posterior. If tuning from $\lambda_n$ is not available (that is, $\lambda_n = 1$), then Theorem \ref{thm:main} (i) shows that the AMC posterior converges in the $L_2$ norm at least at the rate $n^{-(2\vv-d)/(4\vv)}$, which is slightly slower than the optimal rate in part (ii). Furthermore, Theorem \ref{thm:main} also gives sufficient conditions for the subset size $m$ in the two scenarios. For part (i), the order of $m$ is $m \gtrsim n^{(d/\vv)+\eta}$, which is meaningful when $\vv>d$ since $m\leq n$. Since Assumption \ref{sampling} says that $mk$ and $n$ have the same order, this implies that the number of subsets $k$ can increase no faster than $n^{(\vv-d)/\vv}$. For part (ii), the order of $m$ is $m \gtrsim n^{2d/(2\vv+d)+\eta}$ and this works for all $\vv>d/2$ as in Assumption \ref{eigen_assumption}. As a result, the number of subsets $k$ can increase no faster than $n^{(2\vv-d)/(2\vv+d)}$.

Our convergence rates in Theorem \ref{thm:main} are also comparable to similar theoretical results on distributed Bayesian inference in non-parametric regression models using an univariate GP (without a VCM formulation). This includes the recent works of \citet{guhaniyogi2017divide} and \citet{SzaVan19}. While \citet{Baietal19} have also shown the posterior contraction rates for the varying coefficients, their Bayesian model is based on basis series expansion instead of multivariate latent GPs as in \citet{Geletal03} and our current paper. Furthermore, their main focus is on high dimensional variable selection, which is different from the big $n$ problem considered here.

\section{Experiments}
\label{sec:numerical}

This section evaluates the performance of methods based on the divide-and-conquer technique for inference and predictions in Bayesian VCMs using a simulation study and a real data analysis. The simulation settings, including the details of data generation, competing methods and the metrics for comparison are described in the first subsection. The second subsection presents simulation results for different methods for a comprehensive comparison. The third subsection presents an application of the Bayesian spatiotemporal VCM to a large dataset of sea surface temperature and salinity in the North Atlantic Ocean.

\subsection{Setup}
\label{sec:setup}

\noindent\underline{\textbf{Data generation:}}
\vspace{2mm}

To assess performance of distributed methods, we design two simulation studies, referred to as \emph{Simulation 1} and \emph{Simulation 2}, with $n=3000$ and $n=9000$ samples, respectively. The sample size in \emph{Simulation 1} is moderately large to ensure that posterior computation of VCMs using the full data, although exorbitantly slow, are tractable and their results serve as the benchmark. In both simulations, the cardinality of the set of indexes $\Ucal^*$, where the function estimation and prediction are evaluated, is set at $300$. Our simulation studies consider $d=2$, with sample indices  $u_1, \ldots, u_n$, and the indices in the set $\Ucal^*$, $u^*_1, \ldots, u^*_{300}$ are simulated independently from the uniform distribution on $[0, 1]^2$. Both simulations assume
$p=3$ predictors, with all $p$ predictors have varying coefficients, i.e., $q=p=3$. We simulate a bivariate response function at all indices (i.e., $s_i=2$) using the varying coefficient model \eqref{eq:mdl11} as
\begin{align}
  \label{eq:s1}
  y(u) = X(u) \beta_0(u) + \epsilon(u), \quad \beta_0(\cdot) = \alpha_0 + \Gamma_0 \nu(\cdot),  \quad u \in \{u_1, \ldots, u_n, u^*_1, \ldots, u^*_{300}\},
\end{align}
where $X(u)$s are $2 \times 3$ predictor matrices at each index, with each of their entries is independently simulated from $N(0, 1)$. To construct the varying coefficients, entries of the $3\times 3$ matrix $\Gamma_0$ are independently simulated from uniform($0, 3$) and $\alpha_0^\T$ is fixed at $(-2, 2, -2)$. The components $\nu_1(\cdot), \nu_2(\cdot), \nu_3(\cdot)$ of the LMC coefficient vector $\nu(\cdot) = (\nu_1(\cdot), \nu_2(\cdot), \nu_3(\cdot))^\T$, are drawn from independent  GPs with 0 mean and correlation functions $\rho_a(u, u') = e^{-\phi_a \| u - u' \|_2}$, where $\phi_a = a$ for $a = 1, 2, 3$. $\epsilon(u)$s are idiosyncratic errors following i.i.d. $N(0, \tau^2)$. Both simulations set the error variance $\tau^2$ at $0.1$. Each simulation is replicated ten times.
\vspace{3mm}

\noindent\underline{\textbf{Competing methods:}}
\vspace{2mm}

We compare the performance of AMC algorithm with two sets of competitors. The first set of competitors include distributed Bayesian methods which follow the same three step algorithm as AMC, with the main difference appearing in the third step involving the subset posterior combination. As part of our comparison endeavor with such distributed Bayesian methods, we include DPMC, PIE, WASP and CMC algorithms as competitors. For each of these competitors, we first create $k$ subsets of sizes $m=500$ and $m=1000$, using subsampling without replacement in both simulations and vary $k$ as 10, 20 and 30, 60 in \emph{Simulations} 1 and 2, respectively. Second, we use the DA-type algorithm developed in Section \ref{sec:sampling-step} in parallel to obtain posterior samples of $\tau^2$, $\beta(u^*)$, and $y(u^*)$ for $u^* \in \Ucal^*$ from all the subsets. The sampling algorithm uses a sparse GP based on the FITC approximation with $r=400$ inducing points  \citep{QuiRas05,Alvetal12}. The imputation of $\tilde \nu_{j}$ in part (a) of our DA-type algorithm is done using the Lanczos algorithm of \citet{ChoSaa14} for computational tractability. The sampling algorithm in each subset runs for 10,000 iterations, and the Markov chain is thinned by collecting every fifth posterior sample after discarding the first 5,000 posterior samples as burn-in. Finally, we combine subset posterior samples for $\tau^2$, $\beta(u^*)$, and $y(u^*)$ for $u^* \in \Ucal^*$ using CMC, DPMC, WASP, PIE and the AMC algorithms (described in Section \ref{sec:comb-step}). The combination steps of AMC, DPMC, WASP, and PIE satisfy the Assumption \ref{combine_assumption} in Section~\ref{sec:theory}, so we expect similar empirical performance for these four methods. In contrast, no such theoretical guarantee exists for the performance of the CMC algorithm.

It is also instructive to compare performance of AMC posterior with the second set of competitors which include the Bayesian VCMs on the full data. To this end, we compare with the true posterior distribution computed using full data, which sets the performance benchmark for the distributed methods. Since there is no open source implementation available for the Bayesian VCMs with bivariate response vector, we implement it by ourselves following the DA-type algorithm discussed in Section~\ref{da-type-algo}. Additionally, \citet{FinBan20} offer \texttt{spSVC} in the \texttt{spBayes} R package for fitting spatial VCMs, which are special cases of \eqref{eq:mdl1} with $d=2$, and fits to our simulation settings. Unfortunately, the current software support is limited to univariate responses in VCMs; hence we implement \texttt{spSVC} function marginally on each component of the bivariate response vector, and refer to this competitor as \texttt{spSVC}. Ignoring correlation between the two components in the bivariate response will presumably lead to a loss in inferential accuracy of \texttt{spSVC} compared to the other competitors. Both the DA-type algorithm of Section \ref{da-type-algo} and implementation of \texttt{spSVC} are prohibitively slow when $n=9000$, so we present their results only when $n=3000$.
\vspace{3mm}

\noindent\underline{\textbf{Comparison metrics:}}
\vspace{2mm}

The point estimation of the varying coefficients and predictions at $\Ucal^*$ from all methods are compared using mean square error (MSE), and mean square prediction error (MSPE), respectively. Further, the coverage and length of 95\% credible and predictive intervals (CIs, PIs) from the competing methods help assessing uncertainty in function estimation and in prediction, respectively. Let $\beta_0(u^*)=\{\beta_{01}(u^*), \ldots, \beta_{0p}(u^*)\}$, $y(u^*) = \{y_1(u^*), \ldots, y_s(u^*)\}$ be the true values of $\beta(\cdot)$, $y(\cdot)$ at $u^*$, where $p$ and $s$ are their dimensions and $u^* \in \Ucal^* \subset [0,1]^d$. Let $\hat \beta(u^*)$, $\hat y(u^*)$ be the posterior means of
$\beta(u^*)$, $y(u^*)$, respectively. Then, the MSE in estimating $\beta(\cdot)$ and MSPE in predicting $y(\cdot)$ are defined as
\begin{align}
  \label{eq:s2}
  \text{MSE}&= \frac{1}{|\Ucal^*|} \sum_{i=1}^{|\Ucal^*|} \sum_{j=1}^p \{\hat \beta_j(u^*_i) - \beta_{0j}(u^*_i)\}^2, \nonumber \\
  \text{MSPE} &= \frac{1}{|\Ucal^*|} \sum_{i=1}^{|\Ucal^*|} \sum_{j=1}^s \{\hat y_j(u^*_i) - y_j(u^*_i)\}^2.
\end{align}
We evaluate point-wise coverage and length of the 95\%  CIs and PIs obtained from the posterior distributions of $\beta(\cdot)$ and $y(\cdot)$, respectively, for every $u \in \Ucal^*$, for all competitors. As discussed before, the simulation settings assume $s=2$, $p=2$, and $|\Ucal^*|=300$. Let ESS$_{\text{DA}}$ be the effective sample size of any DA-type algorithm that runs for $T_{\text{DA}}$ hours, where DA can signify any of the competitors discussed above. Then, following \citet{Johetal19}, we define the computational efficiency of a DA-type algorithm, including spSVC, AMC, CMC, DPMC, WASP, or the full data posterior (referred to as the true posterior) as
\begin{align}
  \label{eq:s3}
  \text{Computational Efficiency}_{\text{DA}} = \log_2 \text{ESS}_{\text{DA}} / \text{T}_{\text{DA}},
\end{align}
where ESS$_{\text{DA}}$ is computed using the coda R package \citep{Plu03}. We do not compute \eqref{eq:s3} for the PIE combination algorithm since it is not designed to provide MCMC samples for the parameters.

\subsection{Simulated Data Analysis}
\label{sec:simul-data-analys}

Table \ref{tab:coef-cov} and \ref{tab:pred-cov} show the performance of all methods in terms of estimating the true varying coefficient $\beta_0(\cdot)$ and prediction of $y(\cdot)$, respectively. As expected, the empirical performance of combined posterior obtained using AMC, WASP, PIE, and DPMC are very similar in terms of all the comparison metrics. Specifically, for both \emph{Simulation} 1 and 2, they yield similar MSE, MSPE, close to nominal coverage and similar length of 95\% CIs and PIs, which validates Theorem~\ref{thm:main} empirically.  The true posterior being the gold standard, achieves little lower MSE and a bit narrower 95\% CIs and PIs than its distributed competitors; however, the computational efficiency of the true posterior is much smaller than that of AMC, WASP, PIE, and DPMC because it requires much longer to finish an iteration compared to its divide-and-conquer competitors. Increasing the size of each subset $m$ leads to better inference, where as varying the number of subsets $k$ with a fixed value of $m$ does not seem to have much impact on the inference. Among the distributed methods, the CI and PI lengths are slightly larger for PIE, perhaps due to the marginal combination of subset posterior distributions. Although WASP shows marginally narrower 95\% CIs and PIs compared to AMC and DPMC (all maintaining close to the nominal coverage), the subset posterior combination step of AMC and DPMC are much more computationally convenient. On the other hand, the CI and PI lengths of CMC are very small compared to that of the WASP, which results in poor coverage for every $m$ and $n$ and deteriorates as $k$ increases. Except CMC, all other methods show similar performance for inference on the error variance $\tau^2$ (Table \ref{tab:tau-ci}). We also find the performance of \texttt{spSVC} to be excellent in predicting $y(u^*)$s, but becomes extremely poor in inference on $\beta(u^*)$s. The poor performance of \texttt{spSVC} in inference on $\beta(u^*)$s is mainly because the marginal model ignores the dependence between $y_1(\cdot)$ and $y_2(\cdot)$. On the other hand, \texttt{spSVC} shows excellent performance in predicting $y(u^*)$s because the marginal \texttt{spSVC} model still uses three GPs for predicting $y_1(\cdot)$ and $y_2(\cdot)$.

AMC, DPMC, PIE, and WASP satisfy Assumption \ref{combine_assumption} on the combination of subset posterior distributions, whereas CMC does not; therefore, we conclude that methods that satisfy our theoretical assumptions show superior empirical performance. Furthermore, the AMC and DPMC combination algorithms are the simplest among the distributed competitors that offer combination of subset posteriors of all parameters jointly. Hence, they are simple and computationally convenient alternatives to the full data posterior distribution in Bayesian VCMs for massive data.

  \begin{table}[ht]
  \caption{Summary of the results for inference on $\beta(\cdot)$. The CI coverage and their lengths are averaged across 10 simulation replications, dimensions, and $u \in \Ucal^*$. A `-' for the true posterior and spSVC corresponding to $n=9000$ indicates that the results are missing due to intractable posterior computations. On the other hand,  `-' in reporting the computational efficiency for PIE is due to the lack of definition.}
\label{tab:coef-cov}
  \centering
{\tiny    \begin{tabular}{|r|c|c|c|c|c|c|c|c|}
  \hline
  & \multicolumn{8}{c|}{$n=3000$}  \\
  \hline
  & \multicolumn{4}{c|}{Coverage at 95\% Nominal Level} & \multicolumn{4}{c|}{95\% CI Length}  \\
  \hline
  & \multicolumn{2}{c|}{$m=500$} & \multicolumn{2}{|c|}{$m=1000$}  & \multicolumn{2}{|c|}{$m=500$} & \multicolumn{2}{|c|}{$m=1000$}  \\
  \hline
  & $k=10$ & $k=20$ & $k=10$ & $k=20$ & $k=10$ & $k=20$ & $k=10$ & $k=20$  \\
  \hline
      True Posterior &    \multicolumn{4}{c|}{0.96} & \multicolumn{4}{c|}{2.53}  \\
      Marginal \texttt{spSVC} &    \multicolumn{4}{c|}{0.27} & \multicolumn{4}{c|}{3.22}  \\
  \hline
            AMC & 0.97 & 0.97 & 0.96 & 0.97 & 3.18 & 3.19 & 2.88 & 2.89 \\
            PIE & 0.97 & 0.97 & 0.96 & 0.97 & 3.24 & 3.24 & 3.02 & 2.91 \\
            CMC & 0.56 & 0.42 & 0.57 & 0.39 & 1.10 & 0.79 & 1.09 & 0.67 \\
            WASP & 0.96 & 0.96 & 0.96 & 0.96 & 3.06 & 3.06 & 2.78 & 2.78 \\
            DPMC & 0.97 & 0.97 & 0.96 & 0.97 & 3.19 & 3.20 & 2.89 & 2.89 \\
  \hline
  & \multicolumn{4}{c|}{MSE} & \multicolumn{4}{c|}{Computational Efficiency}  \\
  \hline
      True Posterior &    \multicolumn{4}{c|}{0.40} & \multicolumn{4}{c|}{2.10}  \\
      Marginal \texttt{spSVC} &    \multicolumn{4}{c|}{7.24} & \multicolumn{4}{c|}{6.68}  \\
  \hline
            AMC & 0.56 & 0.56 & 0.48 & 0.48 & 9.02 & 8.41 & 8.79 & 8.37 \\
            PIE & 0.57 & 0.57 & 0.52 & 0.48 & - & - & - & - \\
            CMC & 0.56 & 0.56 & 0.51 & 0.48 & 4.70 & 4.56 & 4.27 & 2.60 \\
            WASP & 0.57 & 0.57 & 0.49 & 0.48 & 9.02 & 8.41 & 8.79 & 8.37\\
            DPMC & 0.57 & 0.57 & 0.49 & 0.48 & 9.02 & 8.41 & 8.79 & 8.36\\
  \hline
  & \multicolumn{8}{c|}{$n=9000$}  \\
  \hline
  & \multicolumn{4}{c|}{Coverage at 95\% Nominal Level} & \multicolumn{4}{c|}{95\% CI Length}  \\
  \hline
  & \multicolumn{2}{c|}{$m=500$} & \multicolumn{2}{|c|}{$m=1000$}  & \multicolumn{2}{|c|}{$m=500$} & \multicolumn{2}{|c|}{$m=1000$}  \\
  \hline
  & $k=30$ & $k=60$ & $k=30$ & $k=60$ & $k=30$ & $k=60$ & $k=30$ & $k=60$  \\
  \hline
      True Posterior &    \multicolumn{4}{c|}{-} & \multicolumn{4}{c|}{-}  \\
      Marginal \texttt{spSVC} &    \multicolumn{4}{c|}{-} & \multicolumn{4}{c|}{-}  \\
  \hline
            AMC & 0.98 & 0.98 & 0.98 & 0.98 & 3.04 & 3.05 & 2.75 & 2.75\\
            PIE & 0.98 & 0.98 & 0.98 & 0.98 & 3.07 & 3.08 & 2.77 & 2.77 \\
            CMC & 0.34 & 0.25 & 0.34 & 0.24 & 0.58 & 0.41 & 0.52 & 0.37 \\
            WASP & 0.97 & 0.97 & 0.97 & 0.97 & 2.91 & 2.92 & 2.64 & 2.63 \\
            DPMC & 0.98 & 0.98 & 0.98 & 0.98 & 3.05 & 3.06 & 2.75 & 2.75\\
  \hline
  & \multicolumn{4}{c|}{MSE} & \multicolumn{4}{c|}{Computational Efficiency}  \\
  \hline
      True Posterior &    \multicolumn{4}{c|}{-} & \multicolumn{4}{c|}{-}  \\
  \hline
            AMC & 0.45 & 0.45 & 0.38 & 0.38 & 10.40 & 9.67 & 9.94 & 9.70  \\
            PIE & 0.46 & 0.46 & 0.39 & 0.39 & - & - & - & - \\
            CMC & 0.45 & 0.45 & 0.39 & 0.39 & 5.76 & 4.03 & 5.47 & 4.01\\
            WASP & 0.45 & 0.45 & 0.38 & 0.38 & 10.40 & 9.67 & 9.94 & 9.70\\
            DPMC & 0.45 & 0.45 & 0.38 & 0.38 & 10.40 & 9.67 & 9.94 & 9.70\\
   \hline
\end{tabular}}%
\end{table}

\begin{table}[ht]
  \caption{Summary of the results for $y(\cdot)$ prediction. The PI coverage and their lengths are averaged across 10 simulation replications, dimensions $a=1,2$, and $u \in \Ucal^*$. A `-' for the true posterior and spSVC corresponding to $n=9000$ indicates that the results are missing due to intractable posterior computations. On the other hand,  `-' in reporting the computational efficiency for PIE is due to the lack of definition.}
\label{tab:pred-cov}
\centering
{\tiny
\begin{tabular}{|r|c|c|c|c|c|c|c|c|}
  \hline
  & \multicolumn{8}{c|}{$n=3000$}  \\
  \hline
  & \multicolumn{4}{c|}{Coverage at 95\% Nominal Level} & \multicolumn{4}{c|}{95\% PI Length}  \\
  \hline
  & \multicolumn{2}{c|}{$m=500$} & \multicolumn{2}{|c|}{$m=1000$}  & \multicolumn{2}{|c|}{$m=500$} & \multicolumn{2}{|c|}{$m=1000$}  \\
  \hline
  & $k=10$ & $k=20$ & $k=10$ & $k=20$ & $k=10$ & $k=20$ & $k=10$ & $k=20$  \\
  \hline
  True Posterior &    \multicolumn{4}{c|}{0.96} & \multicolumn{4}{c|}{4.11}  \\
  Marginal \texttt{spSVC} &    \multicolumn{4}{c|}{1.00} & \multicolumn{4}{c|}{1.66}  \\
  \hline
  AMC & 0.97 & 0.97 & 0.96 & 0.96 & 5.05 & 5.05 & 4.61 & 4.62 \\
  PIE & 0.97 & 0.97 & 0.97 & 0.96 & 5.12 & 5.11 & 4.81 & 4.65 \\
  CMC & 0.55 & 0.42 & 0.57 & 0.39 & 1.75 & 1.25 & 1.73 & 1.08 \\
  WASP & 0.96 & 0.96 & 0.95 & 0.95 & 4.87 & 4.86 & 4.46 & 4.46 \\
  DPMC & 0.97 & 0.97 & 0.96 & 0.96 & 5.06 & 5.06 & 4.61 & 4.63 \\
  \hline
  & \multicolumn{4}{c|}{MSPE} & \multicolumn{4}{c|}{Computational Efficiency}  \\
  \hline
  True Posterior &    \multicolumn{4}{c|}{1.30} & \multicolumn{4}{c|}{2.10}  \\
  Marginal \texttt{spSVC} &    \multicolumn{4}{c|}{0.03} & \multicolumn{4}{c|}{6.67}  \\
  \hline
  AMC & 1.79 & 1.76 & 1.53 & 1.54 & 9.02 & 8.41 & 8.79 & 8.37  \\
  PIE & 1.82 & 1.79 & 1.67 & 1.55 & - & - & - & - \\
  CMC & 1.78 & 1.75 & 1.65 & 1.53 & 5.72 & 4.10 & 5.55 & 4.11 \\
  WASP & 1.79 & 1.76 & 1.53 & 1.54 & 9.02 & 8.41 & 8.79 & 8.37\\
  DPMC & 1.79 & 1.76 & 1.53 & 1.54 & 9.02 & 8.41 & 8.79 & 8.37\\
  \hline
  & \multicolumn{8}{c|}{$n=9000$}  \\
  \hline
  & \multicolumn{4}{c|}{Coverage at 95\% Nominal Level} & \multicolumn{4}{c|}{95\% PI Length}  \\
  \hline
  & \multicolumn{2}{c|}{$m=500$} & \multicolumn{2}{|c|}{$m=1000$}  & \multicolumn{2}{|c|}{$m=500$} & \multicolumn{2}{|c|}{$m=1000$}  \\
  \hline
  & $k=30$ & $k=60$ & $k=30$ & $k=60$ & $k=30$ & $k=60$ & $k=30$ & $k=60$  \\
  \hline
  True Posterior &    \multicolumn{4}{c|}{-} & \multicolumn{4}{c|}{-}  \\
  Marginal \texttt{spSVC} &    \multicolumn{4}{c|}{-} & \multicolumn{4}{c|}{-}  \\
  \hline
  AMC & 0.98 & 0.98 & 0.97 & 0.97 & 4.80 & 4.82 & 4.39 & 4.38 \\
  PIE & 0.97 & 0.97 & 0.97 & 0.97 & 4.85 & 4.87 & 4.42 & 4.42 \\
  CMC & 0.33 & 0.24 & 0.32 & 0.23 & 0.92 & 0.65 & 0.83 & 0.58 \\
  WASP & 0.96 & 0.96 & 0.96 & 0.96 & 4.62 & 4.63 & 4.23 & 4.22 \\
  DPMC & 0.98 & 0.98 & 0.97 & 0.97 & 4.82 & 4.84 & 4.40 & 4.39 \\
  \hline
  & \multicolumn{4}{c|}{MSPE} & \multicolumn{4}{c|}{Computational Efficiency}  \\
  \hline
  True Posterior &    \multicolumn{4}{c|}{-} & \multicolumn{4}{c|}{-}  \\
  Marginal \texttt{spSVC} &    \multicolumn{4}{c|}{-} & \multicolumn{4}{c|}{-}  \\
  \hline
  AMC & 1.41 & 1.40 & 1.22 & 1.21 &10.40 & 9.67 & 9.94 & 9.70  \\
  PIE & 1.45 & 1.45 & 1.25 & 1.24 & - & - & - & - \\
  CMC & 1.40 & 1.40 & 1.23 & 1.23 & 5.76 & 4.03 & 5.47 & 4.02 \\
  WASP & 1.41 & 1.40 & 1.22 & 1.21 & 10.40 & 9.67 & 9.94 & 9.70 \\
  DPMC & 1.41 & 1.40 & 1.22 & 1.21 & 10.40 & 9.67 & 9.94 & 9.70\\
   \hline
\end{tabular}
}%
\end{table}

\begin{table}[ht]
\caption{The 95\% credible intervals for inference on $\tau^2$. The lower and upper ends of CIs are averaged across 10 simulation replications. }
  \label{tab:tau-ci}
\centering
{\tiny
\begin{tabular}{|r|c|c|c|c|}
  \hline
  & \multicolumn{4}{c|}{$n=3000$}  \\
  \hline
  & \multicolumn{2}{c|}{$m=500$} & \multicolumn{2}{|c|}{$m=1000$} \\
  \hline
  & $k=10$ & $k=20$ & $k=10$ & $k=20$ \\
  \hline
  True Posterior & \multicolumn{4}{c|}{(0.0891, 0.1067)}  \\
  Marginal \texttt{spSVC} &  \multicolumn{4}{c|}{(0.0975, 0.1141)}  \\
  \hline
  AMC & (0.0707, 0.103) & (0.0695, 0.1016) & (0.0778, 0.1323) & (0.0798, 0.1041) \\
  PIE & (0.0717, 0.1037) & (0.0703, 0.1021) & (0.0821, 0.1316) & (0.0801, 0.1042) \\
  CMC & (0.0799, 0.0909) & (0.0802, 0.0881) & (0.0878, 0.0982) & (0.0884, 0.094) \\
  WASP & (0.0708, 0.1029) & (0.0696, 0.1016) & (0.0799, 0.1268) & (0.0799, 0.1041) \\
  DPMC & (0.0707, 0.1031) & (0.0695, 0.1017) & (0.0778, 0.1347) & (0.0798, 0.1041) \\
  \hline
  & \multicolumn{4}{c|}{$n=9000$}  \\
  \hline
  & \multicolumn{2}{c|}{$m=500$} & \multicolumn{2}{|c|}{$m=1000$} \\
  \hline
  & $k=30$ & $k=60$ & $k=30$ & $k=60$ \\
  \hline
  True Posterior & \multicolumn{4}{c|}{-}  \\
  Marginal \texttt{spSVC} &  \multicolumn{4}{c|}{-}  \\
  \hline
  AMC & (0.0709, 0.103)  & (0.0704, 0.1022) & (0.0799, 0.1038) & (0.0795, 0.1031) \\
  PIE & (0.0718, 0.1037) & (0.0713, 0.1029) & (0.0803, 0.1041) & (0.0799, 0.1034) \\
  CMC & (0.0819, 0.0879) & (0.0821, 0.0864) & (0.0887, 0.0932) & (0.0888, 0.0919) \\
  WASP & (0.071, 0.1029) & (0.0704, 0.1021) & (0.08, 0.1038) & (0.0795, 0.1031) \\
  DPMC & (0.0709, 0.103) & (0.0704, 0.1022) & (0.0799, 0.1037) & (0.0795, 0.1031) \\
   \hline
\end{tabular}}%
\end{table}

\subsection{Real Data Analysis}
\label{sec:real-data-analysis}

We illustrate the performance of the combined posterior distributions obtained using  AMC, DPMC, PIE, WASP or CMC combination technique for the space-time varying coefficient modeling, where the indices are $u=(h,t)$ with $h$ and $t$ denoting the spatial locations and time points of the response and covariates. VCMs are widely used in a variety of spatial applications, mostly without the temporal dimension; see, for example, \citet{wheeler2007assessment,finley2014dynamic,banerjee2006coregionalized}. On the contrary, their applications in large data settings are limited, perhaps due to the demanding computations. This section specifically considers the problem of capturing the spatio-temporal association (with uncertainties) between the sea surface temperature (SST) and sea surface salinity (SSS) in the Atlantic Ocean between $0^{\circ}-70^{\circ}$ north latitudes and $0^{\circ}-80^{\circ}$ west longitudes using the spatiotemporal VCM. This implies that $s_i=1$, $p = 2$, $d = 3$, and the space-time tuples lie in a fixed and bounded domain for the spatiotemporal model based on  \eqref{eq:mdl1}. The data on SST and SSS are obtained from the Hadley center observations under the met office in UK (www.metoffice.gov.uk/hadobs, more description available in \cite{kennedy2011reassessing}). We specifically consider $72000$ space-time observations on SST and SSS over the $12$ months in $2018$ and randomly set aside $|\Ucal^*|=600$ space-time tuples for prediction, which form the set $\Ucal^*$ of size 600. Full scale Bayesian inference of spatio-temporal VCMs
with data at this scale is extremely challenging and has been sparsely dealt with in the literature.

The global association between SST and SSS is well established in the fields of Oceanography and Geophysics \citep{millero1998distribution,key2004global,lee2006global}. In fact, salinity influences the depth to which water masses sink and how far they extend through the ocean. The location and depth of these water masses controls how heat are transported between the tropics and high latitudes. Both SST and SSS are also key in understanding how oceans interact with the atmosphere. Monsoons are driven by exchanges at the air-ocean boundary, affecting almost half of the world's human population each year. Likewise, El Ni\~no has profound effects on humankind and is, to an unknown extent, governed by ocean salinity and temperature. Earlier work with nonlinear regression models to ascertain relationships between SST and SSS \citep{xiong2013relationship, becker1996sea} reveal significant positive association between these two climate indicators. Although
some of these prior studies reveal such associations to be spatially varying \citep{weldeab2006deglacial}, there is still a dearth of model based analysis of spatio-temporally varying associations between SST and SSS.

We compare performance of the combined posterior obtained using AMC, DPMC, WASP, PIE combination schemes (all following theoretically guaranteed optimal performance) along with the other distributed competitor CMC, popularly used in the machine learning literature for distributed inference with massive data. We have also attempted to fit spatially varying coefficient model on the full data using the \texttt{spSVC} function in the \texttt{spBayes} package in \texttt{R}; however, full data posterior computations using \texttt{spSVC} fails due to the large sample size. Gneiting's correlation function \eqref{eq:corr-fun-1} is employed in the spatiotemporal VCM due to its flexibility in modeling space-time correlations \citep{gneiting2002nonseparable}. The values of $k$ and $m$ are set to be 400 and 2500, respectively, and the results for the distributed methods follow from three step strategy described in Section~\ref{sec:algorithm}. Because the true varying coefficients are unknown, we only make assessment of point prediction and predictive uncertainties for all the methods using MSPE and coverage of 95\% PIs for the space-time tuples in $\Ucal^*$, respectively. Computational efficiency of all methods are also reported.

Figure~\ref{Fig_STVC} presents the posterior mean of the spatially varying coefficient corresponding to SSS in January, May and September for AMC, PIE, WASP and DPMC. From the equator to the pole, the annual excess precipitation over evaporation increases, and thus salinity decreases along with SST, with latitude. However, in lower latitude, due to the pronounced salt accumulation as a result of excess heating and oceanic currents, SSS surges, which results in lower $\beta_1(s)$ values. This trend becomes more prominent during the months of summer or fall (columns 2 and 3). In general, SSS decreases in comparison with SST during winter, except for the Brazilian coast, which shows lower coefficient values even in winter due to the strong North Brazil Current \citep{weldeab2006deglacial}. The increase in latitude shows a considerable drop of SSS compared to SST leading to higher $\beta_1(s)$ values. The estimates appear to be consistent over all the four combination approaches (AMC, WASP, PIE and DPMC) following the three step algorithm.

Turning our attention to the predictive inference, Table~\ref{STVC_pred} demonstrates comparable point estimation and predictive uncertainties from AMC, PIE, WASP and DPMC combination schemes. In contrast, the machine learning competitor CMC shows high MSPE and considerably wider credible intervals at all space-time tuples. The computationally efficiency metric for other methods also supercede CMC by a large margin. As a result, the space-time varying coefficient figures corresponding to CMC also appear to be different from the other competitors, and hence it has not been included under Figure~\ref{Fig_STVC}. As a whole, the data analysis reinforces
our findings on the three step divide and conquer approaches with the theoretically guaranteed combination schemes AMC, DPMC, WASP and PIE as simple, computationally efficient, flexible, and fully Bayesian inferential tools for inference in large spatiotemporal data with the VCM model. We emphasize that the scalability of all these approaches depend on the Bayesian VCM model fitted in each subset. Using more computationally efficient variants of GPs for $\nu(\cdot)$ in each subset, a much higher degree of scalability is achievable.

\begin{table}[ht]
  \caption{Summary of the results for prediction of  sea surface temperature at 600 space-time tuples using the spatio-temporal VCM fitted at each subset. The PI coverage and their lengths are based on point-wise 95\% predictive intervals and are averaged across the 600 space-time tuples. Computational efficiency for all methods are also reported. Omission of computational efficiency for PIE is due to the lack of definition.}\label{STVC_pred}
  \centering
  \begin{tabular}{|r|c|c|c|c|}
    \hline
    & Coverage   & MSPE & 95\% PI Length & Computational Efficiency \\
    \hline
    AMC & 0.99 & 2.92 & 6.60 & 9.99 \\
    PIE & 0.99 & 2.93 & 5.61 & - \\
    CMC & 0.90 & 74.95 & 24.71 & 1.06 \\
    WASP & 0.98 & 2.92 & 5.26 & 9.99 \\
    DPMC & 0.99 & 2.92 & 6.53 & 10.09 \\
    \hline
  \end{tabular}
\end{table}
\begin{figure}[ht!]
  \begin{center}
    \subfloat{\includegraphics[width=4.0 cm]{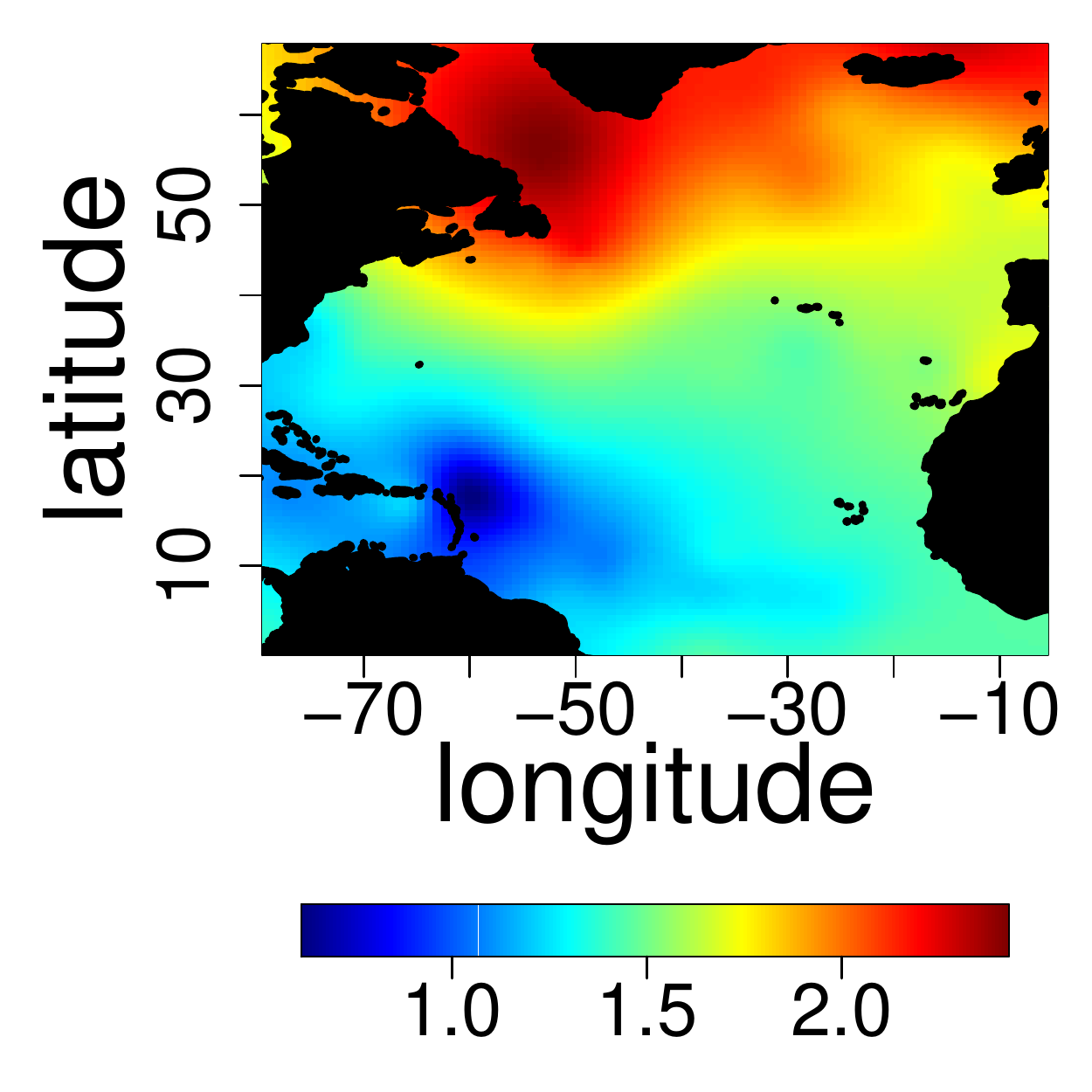}}
    \subfloat{\includegraphics[width=4.0 cm]{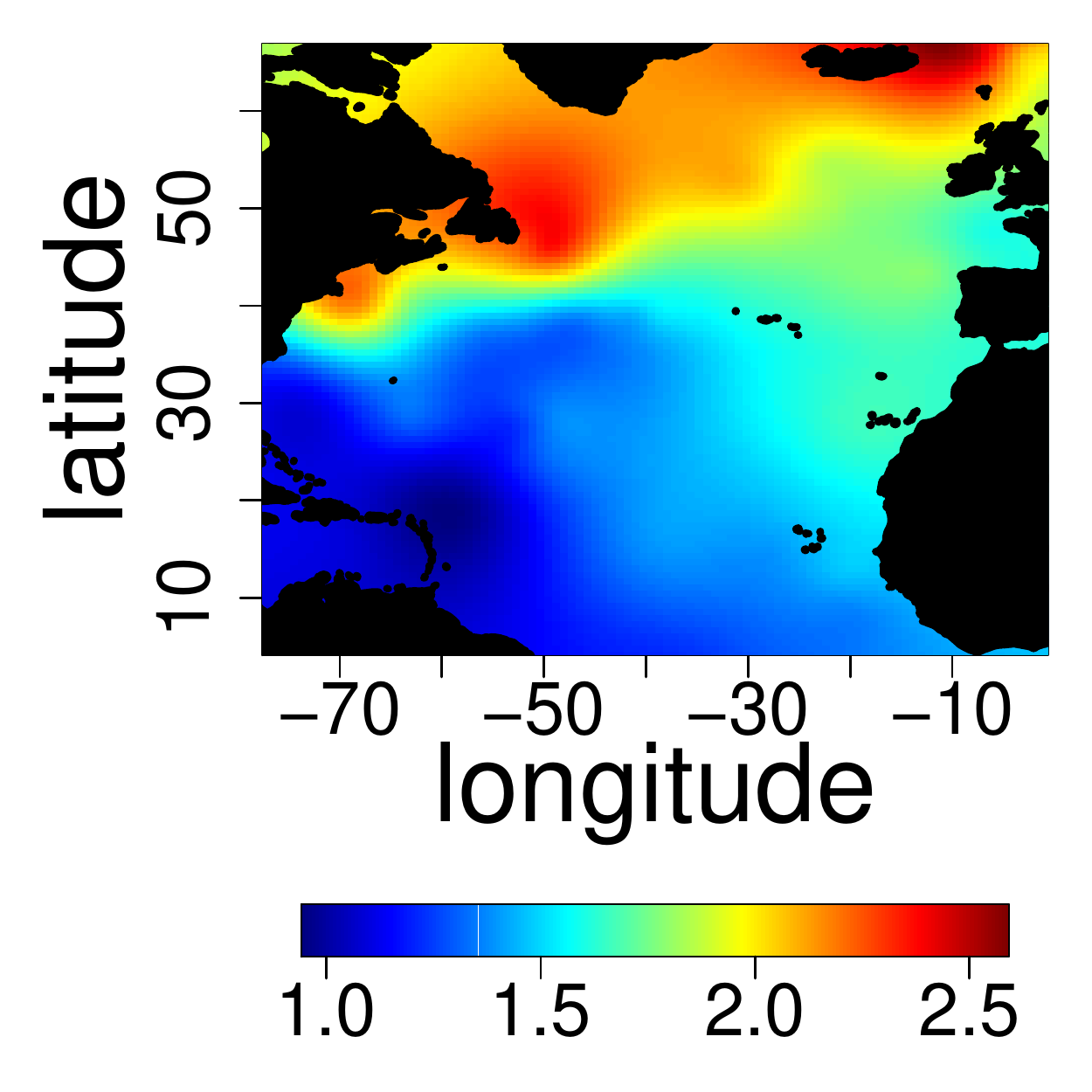}}
    \subfloat{\includegraphics[width=4.0 cm]{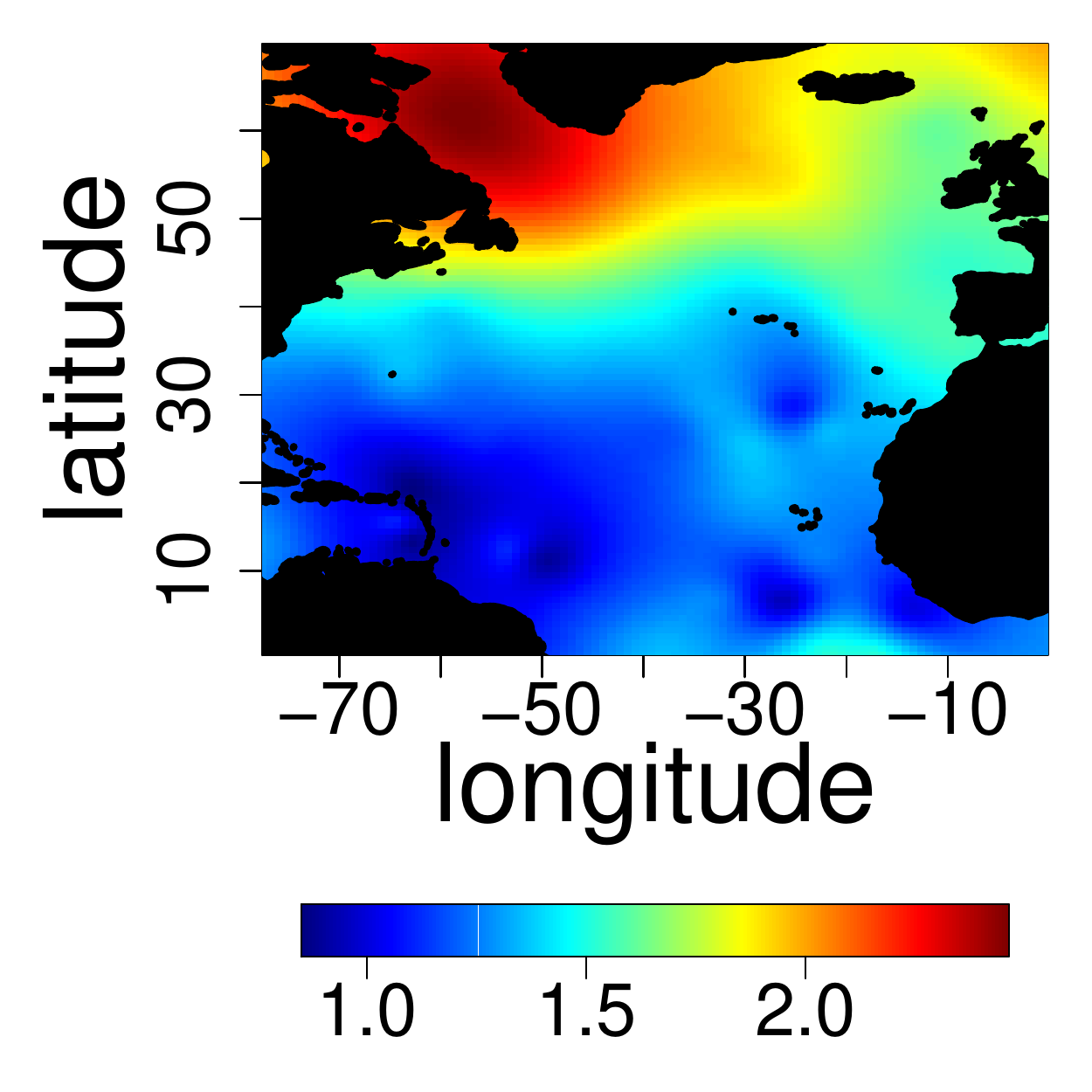}}\\
    \subfloat{\includegraphics[width=4.0 cm]{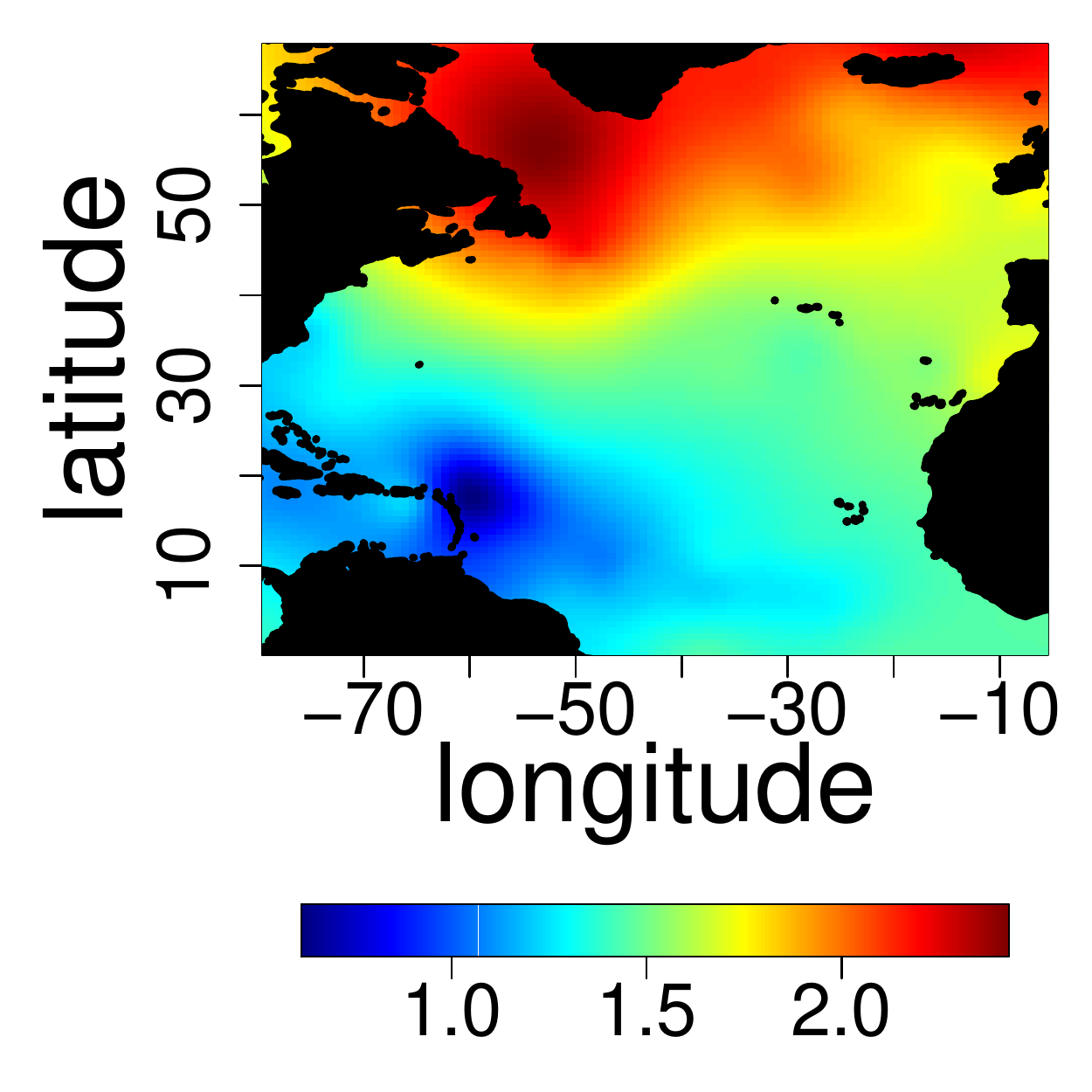}}
    \subfloat{\includegraphics[width=4.0 cm]{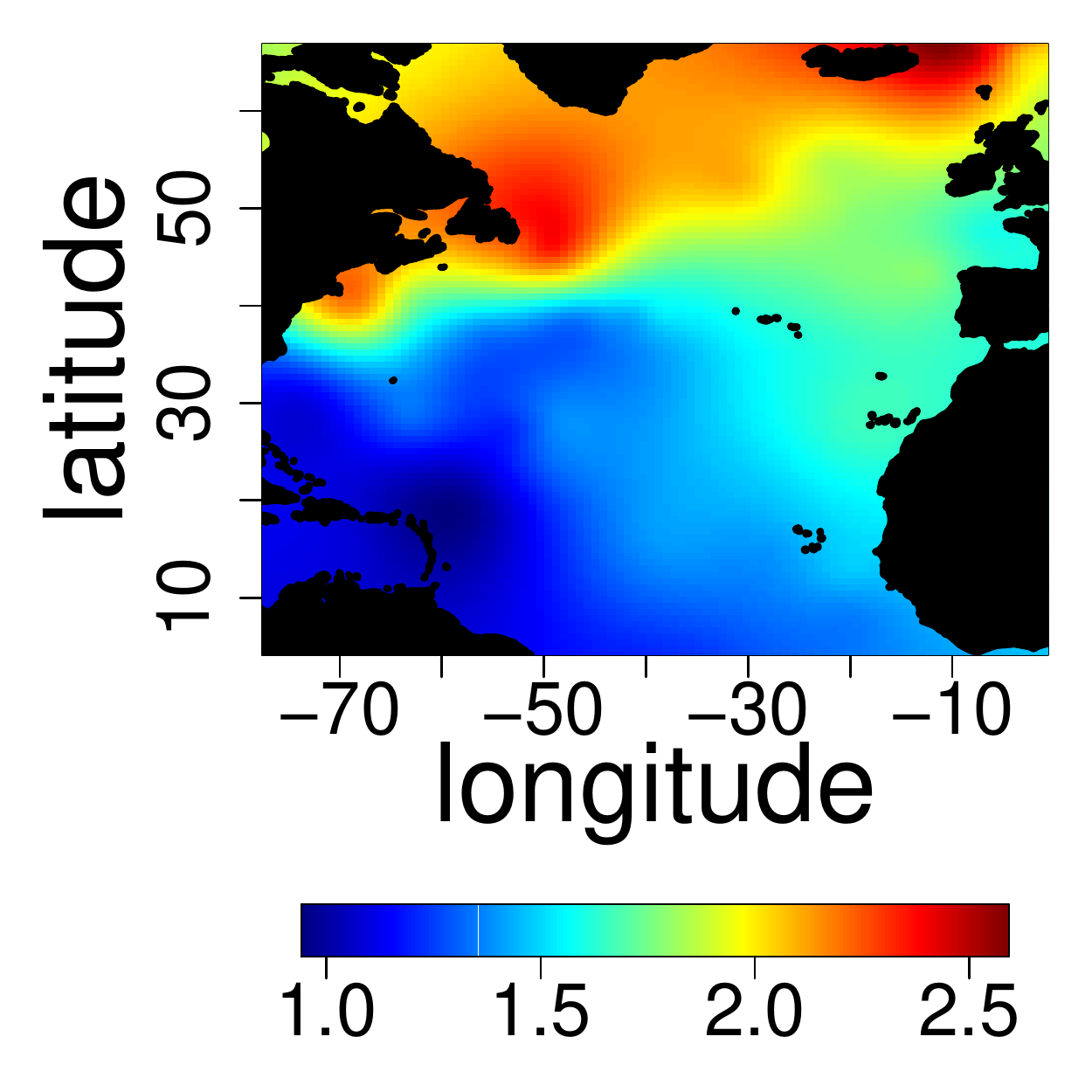}}
    \subfloat{\includegraphics[width=4.0 cm]{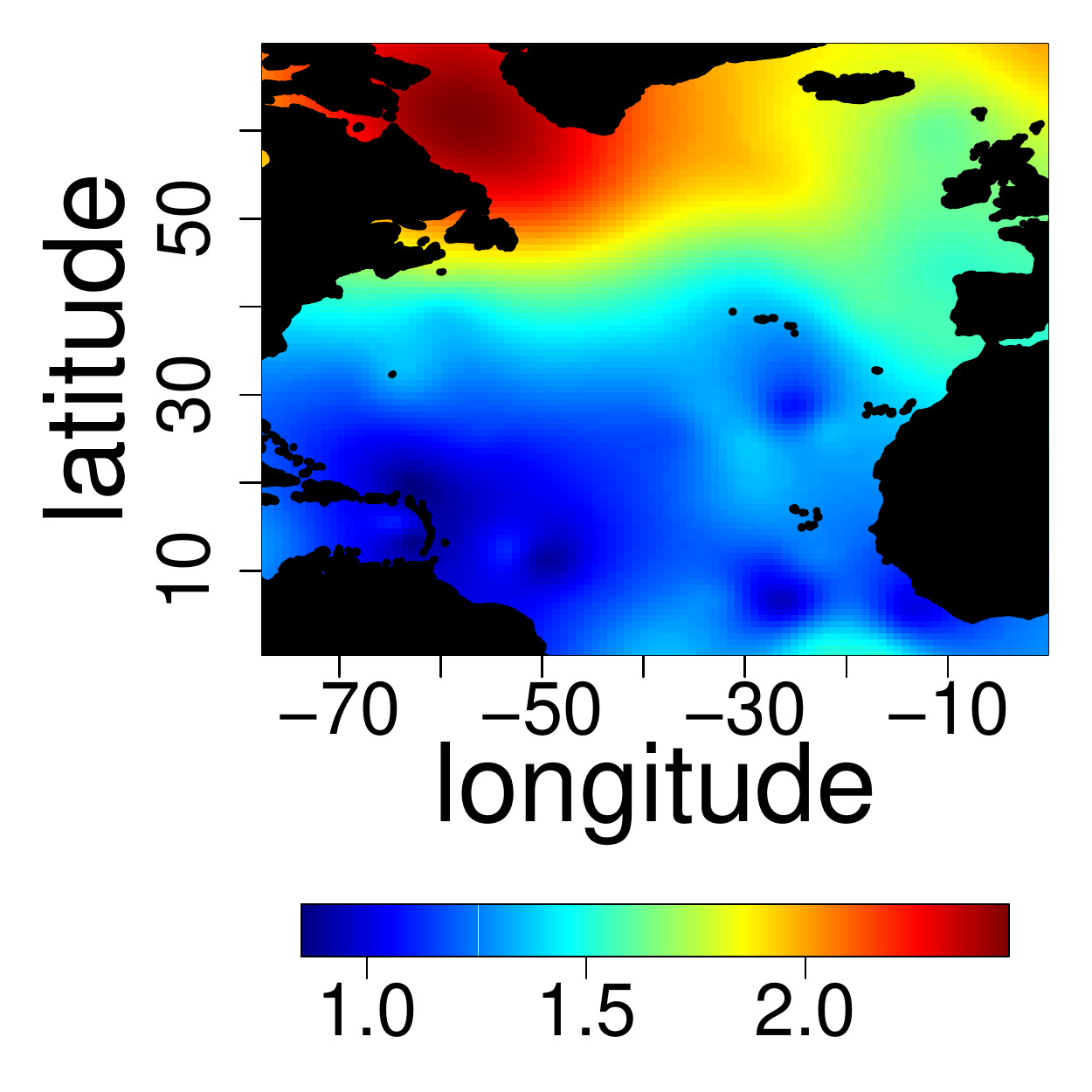}}\\
    \subfloat{\includegraphics[width=4.0 cm]{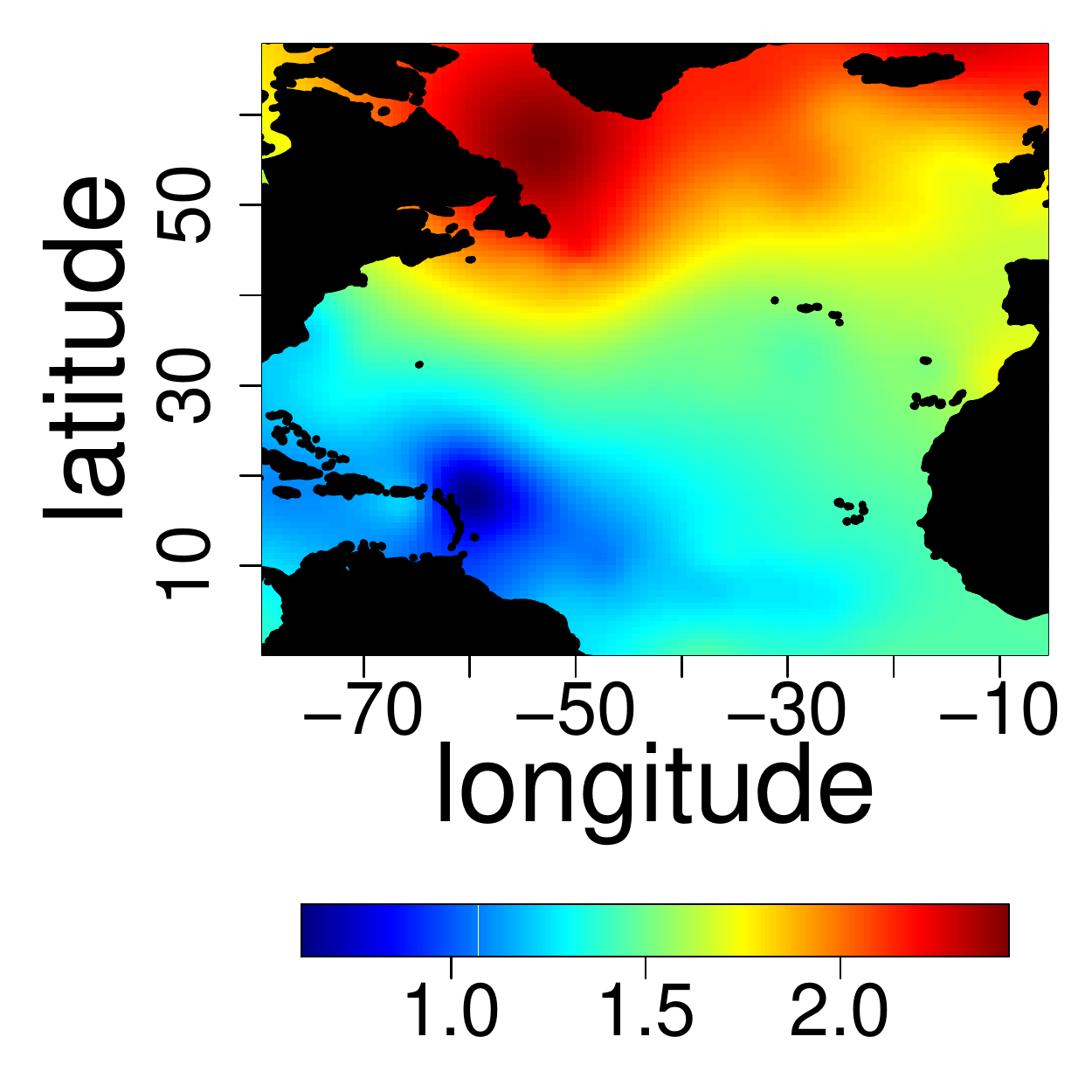}}
    \subfloat{\includegraphics[width=4.0 cm]{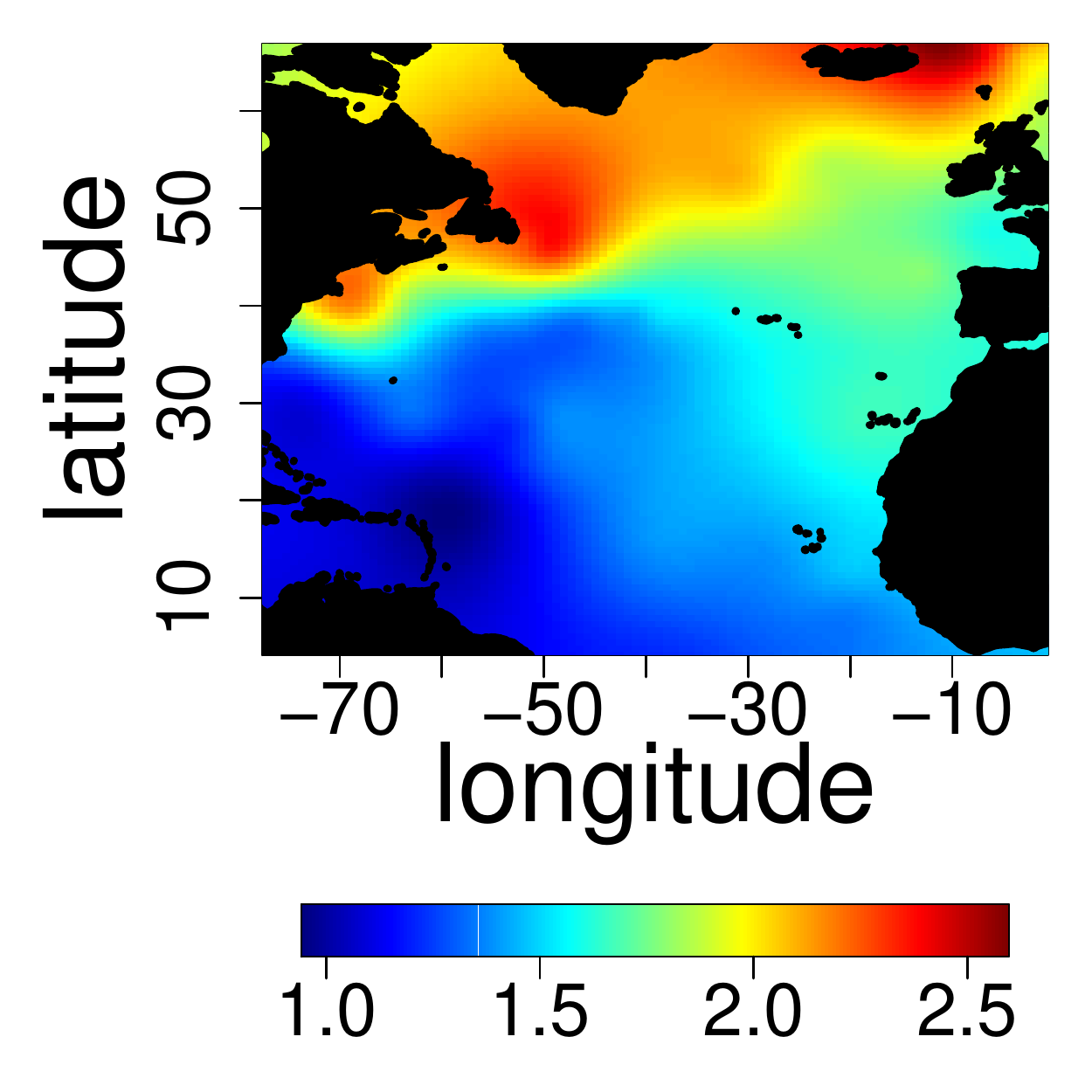}}
    \subfloat{\includegraphics[width=4.0 cm]{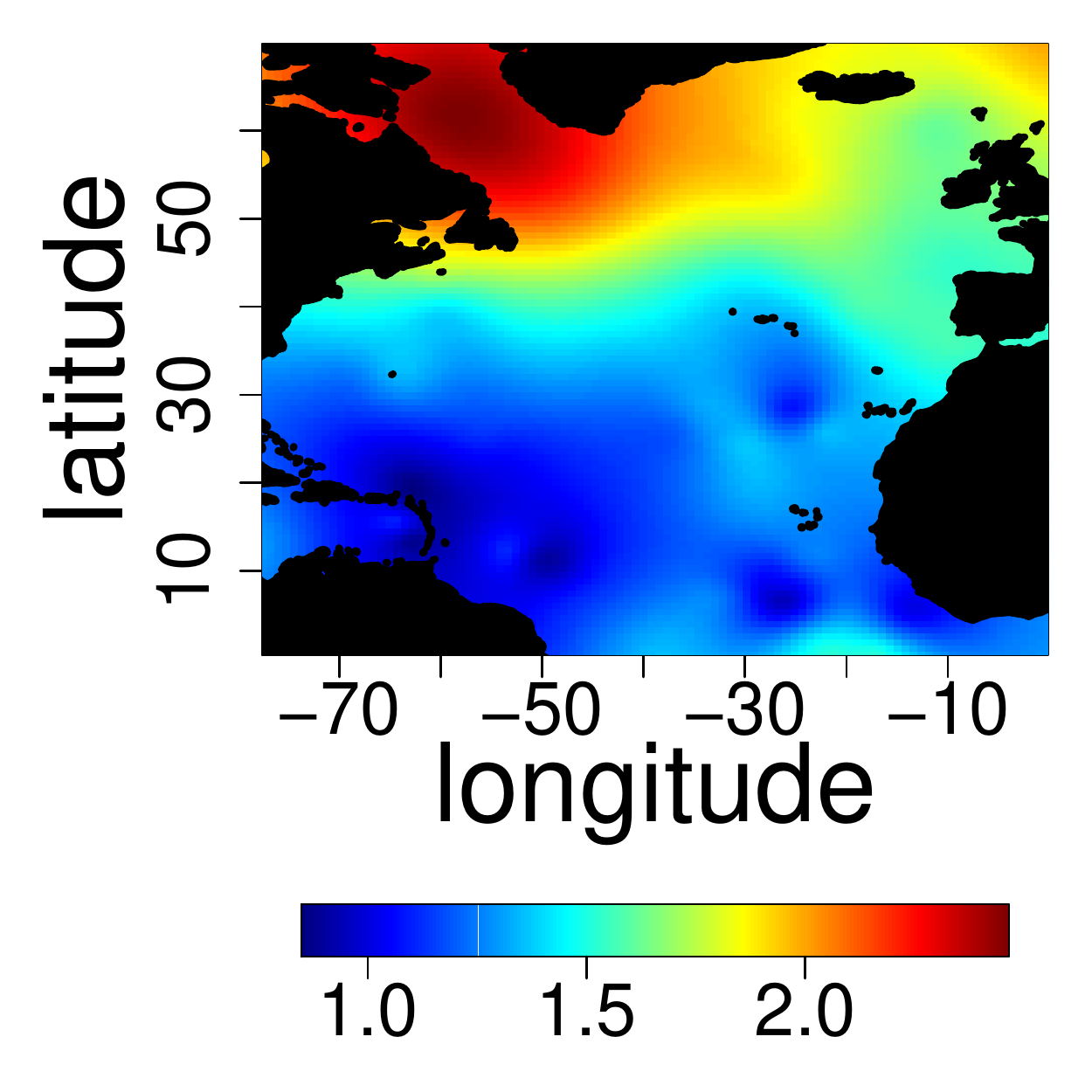}}\\
    \subfloat{\includegraphics[width=4.0 cm]{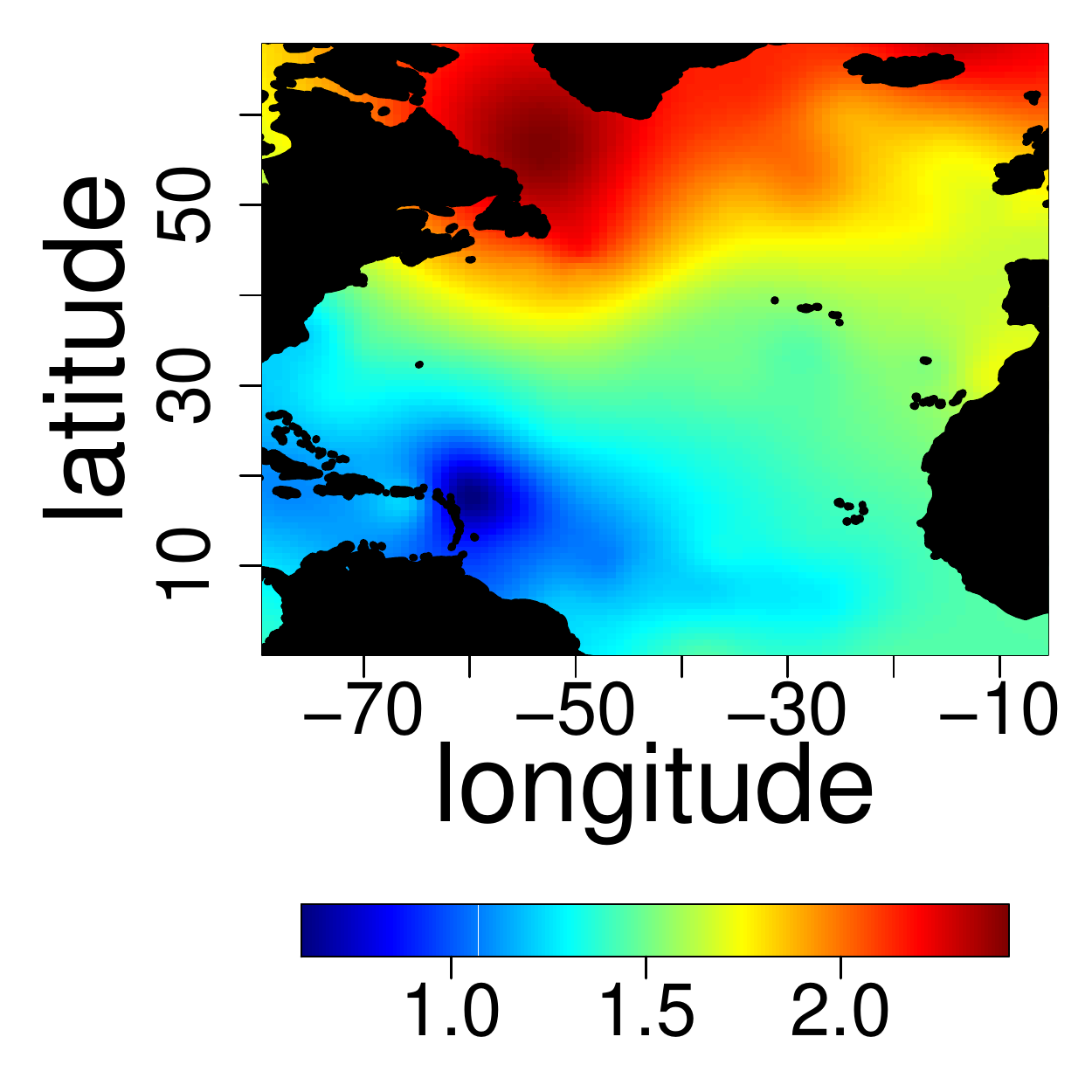}}
    \subfloat{\includegraphics[width=4.0 cm]{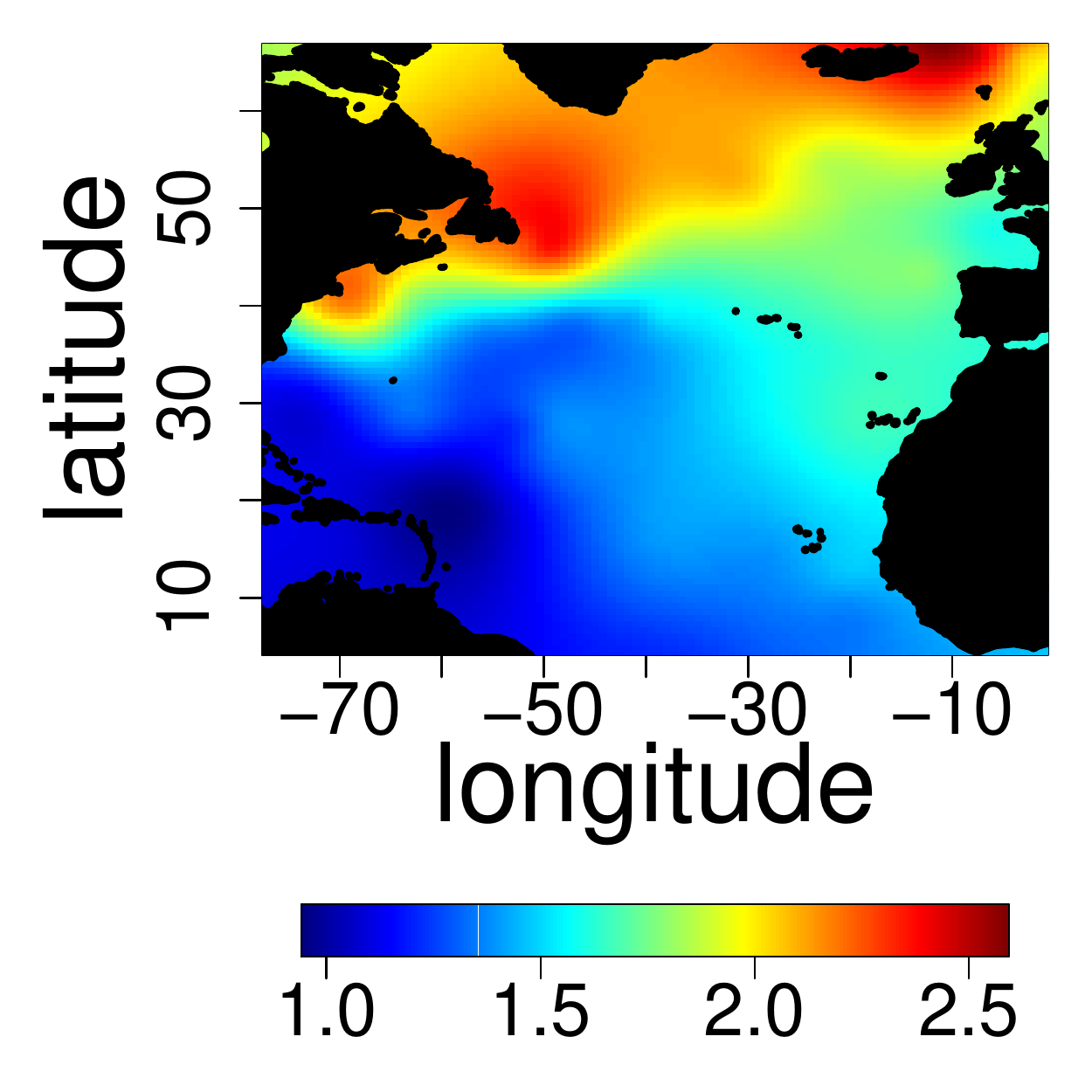}}
    \subfloat{\includegraphics[width=4.0 cm]{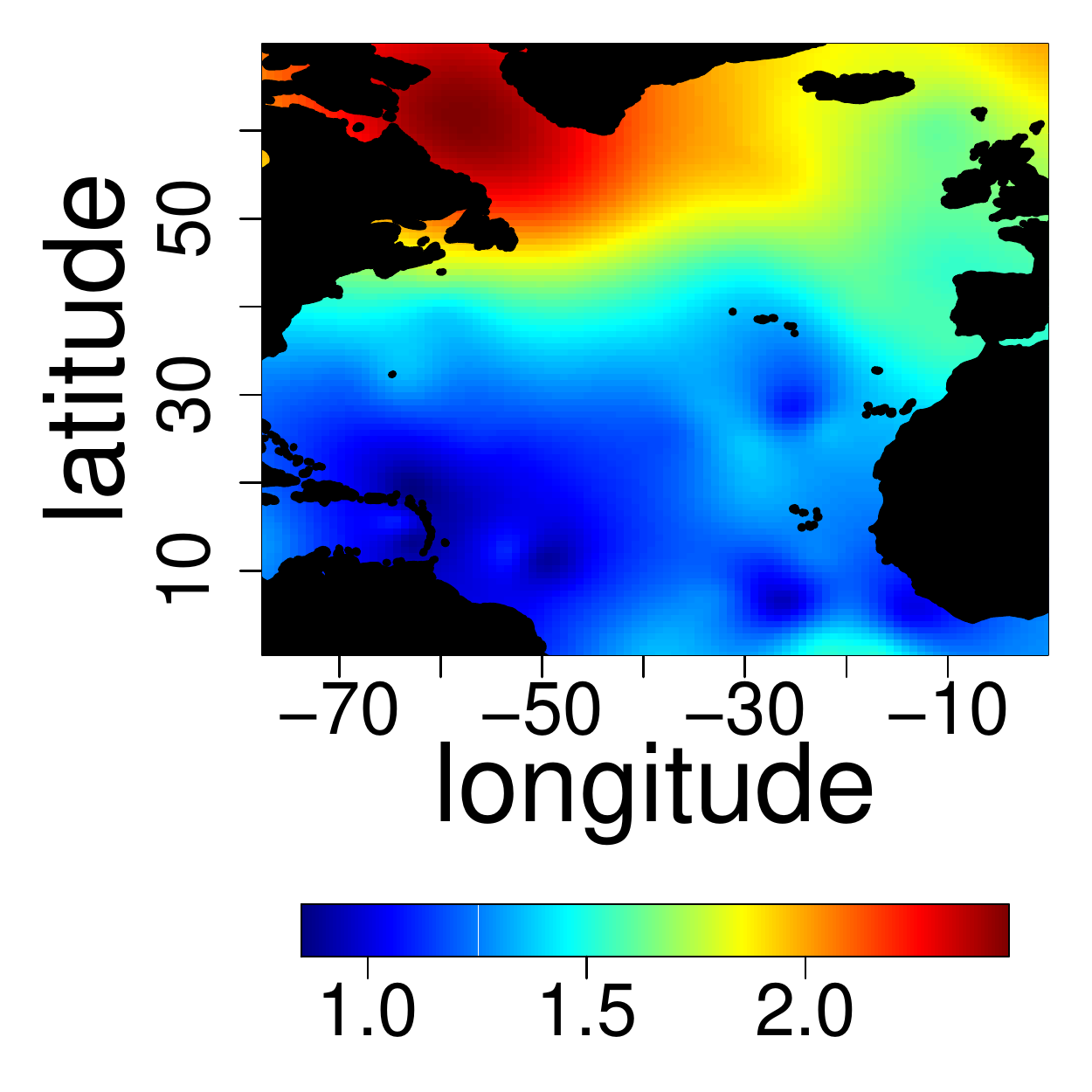}}\\
  \end{center}
  \caption{The four rows represent space-time varying coefficient corresponding to the STVC regression model with SST as the response and SSS as the predictor. Rows 1-4 correspond to the coefficient maps from AMC, WASP, PIE and DPMC respectively. In each row, columns 1, 2, 3 present coefficient maps in January, May and September respectively.}\label{Fig_STVC}
\end{figure}

\section{Discussion}

Bayesian varying coefficient models with multivariate Gaussian process prior on varying coefficients are extremely popular in functional regression models with a wide variety of applications, since they combine the flexibility of a nonparametric model and interpretability of a linear regression model. Unfortunately, full scale Bayesian inference with VCMs are relatively less explored with big data, due to the computation and storage becoming prohibitively burdensome. This article proposes a three step distributed framework that
divides the data into (possibly) overlapping subsets, fits posterior inference with VCM in each subset after appropriately modifying the subset likelihood, and finally aggregates
inference from subsets with a combination algorithm to derive a pseudo posterior distribution that approximates the full posterior. This article is presumably the first approach to design a principled distributed Bayesian algorithm on VCMs with large data. Additionally, we claim threefold contribution in the literature of distributed Bayesian inference on VCMs. First, we identify the modification required in subset posterior likelihoods and justify such modification with rigorous theoretical guarantee. Second, a new subset posterior combination algorithm, referred to as the Aggregated Monte Carlo algorithm, is proposed. Unlike a few other subset posterior combination algorithms proposed in the context of univariate Gaussian process regression models, AMC jointly combines subset posterior of all parameters and yet offers straightforward implementation. Finally, the major contribution of the article becomes theoretically establishing minimax optimal convergence rate for the varying coefficients and regression mean function under the three stage distributed Bayesian framework. The theoretical results explicitly offer the choices of $k$ and $m$ as functions of $n$, smoothness of the fitted Gaussian process and the smoothness of the true varying coefficients to guarantee optimal inference. Furthermore, the theoretical results are proved under mild assumption on the subset posterior combination algorithm, which is found to hold for DPMC, WASP and PIE, along with the proposed AMC algorithm. Simulation studies demonstrate similar point estimation and prediction, along with close to nominal coverage for all the distributed methods that satisfy theoretical assumptions. We fit space-time varying coefficient models to delineate local variability, as well as seasonal variability in the relationship between SST and SSS in the northern Atlantic ocean.

As a first attempt to principled distributed inference with VCMs, we employ FITC approximation of the Gaussian process for each coefficient in every subset inference and are able to seamlessly scale Bayesian VCMs for $\sim 10^5$ observations, even with moderate dimensional multivariate varying coefficients. As an immediate future work, we plan to fit a more computationally convenient approximation to the GPs on varying coefficients \citep{gramacy2015local} to ameliorate scalability. On the theoretical front, our proof techniques do not depend the normality of error terms in the VCM model. Our conjecture is that if we assume the normal error assumption, then by using the techniques in \citet{VarZan11}, it is to possible to improve the rate in Theorem \ref{thm:main} to an adaptive minimax optimal rate $n^{-\vv/(2\vv+d)}$ without the tuning parameter $\lambda_n$, and to relax the function space from RKHS in Assumption \ref{w0_assumption} to the larger space of functions that are less smooth. We leave this direction for future research.

\vspace{10mm}

\noindent {\large \bf Acknowledgements}
\vspace{3mm}

\noindent The research of Rajarshi Guhaniyogi and Sanvesh Srivastava is partially supported by grants from the Office of Naval Research (ONR-BAA N000141812741) and the National Science Foundation (DMS-1854667/1854662). The research of Cheng Li is supported by the Singapore Ministry of Education Academic Research Funds Tier 1 grants R-155-000-201-114 and R-155-000-223-114. The four authors have contributed equally to this work and their names appear in the alphabetical order.

\newpage

\appendix

\section{Draws from the True Posterior Distribution}
\label{da-full-deriv}

Consider the linear mixed effects reformulation of the Bayesian VCM in \eqref{eq:mdl11} and the prior distributions on $(\alpha, \Gamma, \tau^2, \theta_1, \ldots, \theta_q)$ again. Define total number of observations as $s= \sum_{i=1}^n s_i$, $\tilde Z_a$, $\nu_a$, $R_a$ ($a=1, \ldots, q$), and $y$, $X$, $\tilde Z$, $R$, $\nu$, $\epsilon$ as follows:
\begin{align}
  \label{eq-ap-full:1}
  \tilde Z_{a} &= \diag\{Z(u_{1}) \Gamma_a, \ldots,  Z(u_{n}) \Gamma_a\}, \quad \nu_a^\T = \{\nu_{a}(u_{1}), \ldots, \nu_{a}(u_{n})\}, \quad \nu^\T = (\nu_1^\T, \ldots, \nu^\T_q),\nonumber\\
   (R_{a})_{ii'} &= \rho_{a}(u_{i}, u_{i'}), \quad i, i' = 1, \ldots, n, \quad R = \diag(R_1, \ldots, R_q), \quad \epsilon^\T = (\epsilon^\T_1, \ldots, \epsilon^\T_n),\nonumber\\
  y^\T &= \{y(u_1), \ldots, y(u_n)\}, \quad X^\T = [X(u_1)^\T, \ldots, X(u_n)^\T], \quad \tilde Z = [\tilde Z_1, \ldots, \tilde Z_q],
\end{align}
where $\tilde Z_{a}$, $R_a$, $X$, and $Z$ are $s$-by-$n$, $n$-by-$n$, $s$-by-$p$, and $s$-by-$nq$ matrices, respectively, $R$ is a $nq$-by-$nq$ block diagonal matrix, and $\nu$, $\nu_a$, and $y$ are $nq$-by-1, $n$-by-1, and $s$-by-1 vectors, respectively. Using these definitions, the Bayesian VCM in \eqref{eq:mdl11} and the prior distributions are re-written as
\begin{align}
  \label{eq-ap-full:2}
  y &= X \alpha + \sum_{a=1}^q \tilde Z_a \nu_a + \epsilon = \tilde Z \nu + \epsilon, \quad \epsilon \sim N(0, \tau^2 I),  \quad \nu_a \sim N(0, R_a),\nonumber\\
  \nu_a &\sim N(0, R_a), \quad \nu \sim N(0, R), \quad p(\alpha, \Gamma, \tau^2) \propto \tau^{-2}, \quad p(\theta_a) = \text{Uniform}(\underline c_a, \overline c_a),
\end{align}
where the uniform distribution of $\theta_a$ is assumed to be component-wise if $\theta_a$ is a vector. In this case, $\underline c_a$ and $\overline c_a$ are also vectors of the same dimension as $\theta_a$. The prior distribution on the latent variables $\nu$ and parameters $\alpha, \Gamma, \tau^2, \theta_1, \ldots, \theta_q$ are assumed to have the form
\begin{align}
  \label{eq-ap-full:3}
  p(\nu, \alpha, \Gamma, \tau^2, \theta_1, \ldots, \theta_q) = \prod_{a=1}^q \{ p(\nu_a \mid \theta_a) p(\theta_a) \} p(\alpha, \Gamma, \tau^2) \equiv p(\nu, \theta_1, \ldots, \theta_q) p(\alpha, \Gamma, \tau^2).
\end{align}
Our sampling algorithm for drawing $\beta(\cdot)$, $y(\cdot)$, and $\tau^2$ from their respective full data posterior distributions is based on \eqref{eq-ap-full:2} and \eqref{eq-ap-full:3}.

First, we derive the full conditional of $\nu$. Assume that $y, \alpha, \Gamma, \tau^2, \theta_1, \ldots, \theta_q$ are given. Using  \eqref{eq-ap-full:2}, the joint distribution of $(y, \nu)$  is an $(s+nq)$-variate Gaussian distribution with mean $(X \alpha, 0)$, where 0 is an $nq$-by-1 vector, and covariance matrix $\overline C$, where the blocks corresponding to the marginal covariance matrices of $y$, $\nu$, respectively, and their cross covariance matrix are
\begin{align}
  \label{eq-ap-full:4}
  (\overline C)_{yy} = \tilde Z R \tilde Z^\T + \tau^2 I = \sum_{a=1}^q \tilde Z_a R_a \tilde Z_a^\T + \tau^2 I, \quad
  (\overline C)_{\nu \nu} = R, \quad (\overline C)_{y \nu} = \tilde Z R.
\end{align}
This implies that $\nu$ given $y$, $\alpha, \Gamma, \tau^2,\theta_1, \ldots, \theta_q$ follows $N(\mu_{\nu}, \Sigma_{\nu})$, where
\begin{align}
  \label{eq-ap-full:5}
  \mu_{\nu} = R^\T \tilde Z^\T (\overline C)_{yy}^{-1} (y - X \alpha), \quad
  \Sigma_{\nu} = R - R^\T \tilde Z^\T  (\overline C)_{yy}^{-1}\tilde Z R.
\end{align}

Second, we derive the full conditional of $(\alpha, \Gamma, \tau^2)$. Assuming $y, \nu$ are given, define
\begin{align}
  \label{eq-ap-full:6}
  W_i = [X(u_i),\, \nu^\T(u_i) \otimes Z(u_i)], \; i =1 , \ldots, n, \quad
  W^\T = [W^\T_1, \ldots, W^\T_n], \quad b^\T = (\alpha^\T, \gamma^\T),
\end{align}
where $\gamma$ is the column-wise vectorization of $\Gamma$. Rewrite the Bayesian VCM in \eqref{eq-ap-full:2} as
\begin{align}
  \label{eq-ap-full:61}
  y = W b + \epsilon,  \quad \epsilon \sim N(0, \tau^2 I), \quad p(b, \tau^2) \propto \tau^{-2}.
\end{align}
If $\hat b = (W^\T W)^{-1} W^\T y$ is the least squares estimate of $b$ and $\hat y = W \hat b$ is the mean estimate of $y$ based on $\hat b$, then
\begin{align}
  \label{eq-ap-full:7}
  p(\tau^2, b \mid y) &\propto \frac{1}{(\tau^2)^{s/2}} \exp \left\{ - \frac{(y - W b)^\T (y - W b)}{2\tau^2}  \right\} \frac{1} {\tau^2} \nonumber\\
                           &= \frac{1}{(\tau^2)^{s/2 + 1}} \exp \left\{ - \frac{(y - \hat y)^\T (y - \hat y) + (\hat y - W b)^\T (\hat y - W b)}{2\tau^2}  \right\} \nonumber\\
                           &= \frac{1}{(\tau^2)^{s/2 + 1}} \exp \left( - \frac{\| y - \hat y \|_2^2 } {2\tau^2}  \right)
                             \exp \left\{ - \frac{(\hat b - b)^\T (W^\T W) (\hat b -b)}{2\tau^2}  \right\}.
\end{align}
Marginalizing over $b$ in \eqref{eq-ap-full:7} implies that
\begin{align}
  \label{eq-ap-full:8}
  p(\tau^2 \mid y) &\propto \frac{1}{(\tau^2)^{(s - p - q^2)/2 + 1}} \exp \left( - \frac{\| y - \hat y \|_2^2 } {2\tau^2}  \right), \quad \tau^2 \mid y, \nu  \sim \frac{\| y - \hat y \|_2^2 } {\chi^2_{s - p - q^2}},
\end{align}
and \eqref{eq-ap-full:7} assuming $\tau^2$ is given implies that
\begin{align}
  \label{eq-ap-full:9}
  p(b \mid y, \tau^2) \propto
  \exp \left\{ - \frac{(\hat b - b)^\T (W^\T W) (\hat b -b)}{2\tau^2}  \right\} , \quad
  b \mid y, \nu, \tau^2 \sim N \left\{ \hat b,  \tau^2 (W^\T W)^{-1} \right\};
\end{align}
therefore,  \eqref{eq-ap-full:8} and  \eqref{eq-ap-full:9} imply that the distribution of $(b, \tau^2)$ given $y$ lies in the Normal-Inverse-Gamma family. The $\alpha$ draw corresponds to the first $p$ elements of $b$ and the remaining elements of $b$ are ``un-vectorized'' into the $q$-by-$q$ matrix $\Gamma$.

Finally, we derive the ESS algorithm for drawing $\theta_1, \ldots, \theta_q$ given $\nu$. Using \eqref{eq-ap-full:3}, we have that
\begin{align}
  \label{eq-ap-full:10}
  \log p(\theta_1, \ldots, \theta_q \mid \nu) \propto - \frac{1}{2} \sum_{a=1}^q \log |R_a(\theta_a)| - \frac{1}{2} \sum_{a=1}^q \nu_a^\T R_a^{-1}(\theta_a) \nu_a + \sum_{a=1}^q \log 1_{\underline  c_a \leq \theta_a \leq \overline c_a},
\end{align}
but ESS cannot be applied directly for sampling $\theta_1, \ldots, \theta_q$ given $\nu$ due the range restrictions on $\theta_a$s imposed by the uniform prior distributions. We address this problem by first transforming $\theta_1, \ldots, \theta_q$ to $\thetabar_1, \ldots, \thetabar_q$ as
\begin{align}
  \label{eq-ap-full:11}
  \overline \theta_a = \log \frac{\theta_a - \underline c_a} {\overline c_a - \theta_a} , \quad \theta_a = \underline c_a + \frac{\overline c_a - \underline c_a}{1 + e^{-\thetabar_a}},
\end{align}
where each $\thetabar_a \in (-\infty, \infty)$ and the mapping of $\theta_a$ to $\thetabar_a$ is done component-wise if $\theta_a$ is a vector. The form is $\log p(\theta_1, \ldots, \theta_q \mid \nu)$ in  \eqref{eq-ap-full:10} is modified using the Jacobian of the transform as
\begin{align}
  \label{eq-ap-full:13}
  \log p(\thetabar_1, \ldots, \thetabar_q \mid \nu ) \propto &- \frac{1}{2} \sum_{a=1}^q \log |R_a(\thetabar_a)| - \frac{1}{2}  \sum_{a=1}^q \nu_a^\T R_a^{-1}(\thetabar_a) \nu_a + \sum_{a=1}^q \log (\overline c_a - \underline c_a) \nonumber\\
  &+ \sum_{a=1}^q \log \thetabar_a -2  \sum_{a=1}^q \log (1 + e^{\thetabar_a }).
\end{align}
The $\thetabar_a$s in \eqref{eq-ap-full:13} are supported on $(-\infty, \infty)$, so we apply the ESS algorithm using the proposal $N(0, 4 I)$ for $(\thetabar_1, \ldots, \thetabar_q)$, where $I$ is an identiy whose dimension is determined by that of $(\thetabar_1, \ldots, \thetabar_q)$; see Algorithm 1 in \citet{Nisetal14}.

In summary, the sampling algorithm for drawing from the posterior distribution of $(\alpha, \Gamma, \tau^2, \theta_1, \ldots, \theta_q)$ starts from an initial value of parameters $(\alpha^{(0)}, \Gamma^{(0)}, \tau^{2(0)}, \theta^{(0)}_1, \ldots, \theta^{(0)}_q)$  and cycles through the following four steps for $t=0, 1, \ldots, \infty$:
\begin{enumerate}
\item[(a)] draw $\nu^{(t+1)}$ given $y, \alpha^{(t)}, \Gamma^{(t)}, \tau^{2(t)}, \theta^{(t)}_1, \ldots, \theta^{(t)}_q$ from $N(\mu_{\nu}, \Sigma_{\nu})$, where $\mu_{\nu}, \Sigma_{\nu}$  are defined in \eqref{eq-ap-full:5};
\item[(b)] draw $\tau^{2(t+1)}$ given $y, \nu^{(t+1)}$ using \eqref{eq-ap-full:8};
\item[(c)] draw $(\alpha^{(t+1)}, \gamma^{(t+1)})$ given $y, \nu^{(t+1)}, \tau^{2(t+1)}$ using \eqref{eq-ap-full:9} and the vectorization of $\gamma^{(t+1)}$ is reversed to obtain $\Gamma^{(t+1)}$; and
\item[(d)] draw $\theta_1^{(t+1)}, \ldots, \theta_q^{(t+1)}$ given $\nu^{(t+1)}$ using ESS, the likelihood in \eqref{eq-ap-full:13}, and the relation between $\thetabar_a$ and $\theta_a$ in \eqref{eq-ap-full:11}.
\end{enumerate}

In practice, the interest also lies in drawing $\beta(u^*)$  and $y(u^*)$ for $u^* \in \Ucal^*$, where $\Ucal^*$ is a known subset of $[0,1]^d$ with $l$ elements. This accomplished by the addition of two extra steps after steps (a)--(d). Let $\nu_a^* = \{\nu_a(u^*_1), \ldots, \nu_a(u^*_l)\}$ ($a=1, \ldots, q$) and $\nu(u^*) = \{\nu_1(u^*), \ldots, \nu_q(u^*)\}$, $u^* \in \Ucal^*$. Then, the GP prior on $\nu_a(\cdot)$ ($a=1, \ldots, q$) implies that the $\nu^*_a$ given $\nu_a$ and $\theta_a$ is drawn from $N(\mu_a^*, \Sigma_a^*)$, where
\begin{align}
  \label{eq-ap-full:15}
  &\mu_a^* = R_{a*}^\T R_a^{-1} \nu_a , \quad \Sigma_a^* = R_{a**} - R_{a*}^\T R_a^{-1} R_{a*}, a = 1, \ldots, q,\nonumber\\
  &(R_{a*})_{ii'} = \rho_a(u_i, u^*_{i'}), \quad (R_{a**})_{i'i''} = \rho_a(u^*_i, u^*_{i''}),  \quad i = 1, \ldots, n, \; i', i'' = 1, \ldots, l.
\end{align}
Given $(\alpha, \Gamma, \tau^2, \theta_1, \ldots, \theta_q)$ and $u^* \in \Ucal^*$, the draws of $\beta(u^*)$ is obtained as
\begin{align}
  \label{eq-ap-full:16}
  \beta(u^*) = [\{\beta(u^*) \}^\T_{\text{va}}, \{\beta(u^*) \}^\T_{\text{nv}}]^\T , \quad \{\beta(u^*) \}_{\text{va}} = \alpha_{\text{va}} + \Gamma \nu(u^*), \quad  \{\beta(u^*) \}_{\text{nv}} = \alpha_{\text{nv}},
\end{align}
where nv,v correspond to the non-varying and varying coefficients indices, and the draw of $y(u^*)$ is obtained as
\begin{align}
  \label{eq-ap-full:17}
  y(u^*) \sim N (\mu_{y^*}, \tau^2 I), \quad \mu_{y^*} = X(u^*) \beta(u^*);
\end{align}
therefore, at the $(t+1)$ iteration of the full data sampling algorithm, the following two steps are added if $\Ucal^*$ and $X(u^*)$ for every $u^* \in \Ucal^*$ are known:
\begin{enumerate}
\item [(e)] draw $\nu_a^{*(t+1)}$ given $\nu^{(t+1)}_a$, $\theta_a^{(t+1)}$ using \eqref{eq-ap-full:15} ($a=1, \ldots, q$) and obtain $\beta(u^*)^{(t+1)}$ given $\alpha^{(t+1)}$, $\Gamma^{(t+1)}$, $\nu_a^{*(t+1)}$ using \eqref{eq-ap-full:16} for every $u^* \in \Ucal^*$; and
\item [(f)] draw $y(u^*)^{(t+1)}$ given $X(u^*), \beta(u^*)^{(t+1)}, \tau^{2(t+1)}$ using \eqref{eq-ap-full:17} for every $u^* \in \Ucal^*$.
\end{enumerate}
We also take advantage of the low-rank structure in the covariance matrices if inducing points are used defining the covariance functions of GP \citep{QuiRas05,AlvLaw11}.

\section{Draws from the $j$th Subset Posterior Distribution}
\label{da-sub-deriv}

Consider the reformulation of \eqref{eq-ap-full:2} on subset $j$. Define total number of observations on subset $j$ as $\tilde s_j= \sum_{i=1}^m s_{ji}$, $\tilde Z_{ja}$, $\tilde \nu_{ja}$, $R_{ja}$, $y_j$, $X_j$, $\tilde Z_j$, $R_j$, $\tilde \nu_j$, $\epsilon_j$ are the subset $j$ counterparts of $\tilde Z_{a}$, $\nu_{a}$, $R_{a}$, $y$, $X$, $\tilde Z$, $R$, $\nu$, $\epsilon$, where the dimensions of subset $j$ variables are obtained by replacing $s$ by $\tilde s_j$ and $n$ by $m$ in their full data counterparts. Using these definitions, the Bayesian VCM in on subset $j$ and the prior distributions are re-written as
\begin{align}
  \label{eq-ap-sub:21}
  y_j &= X_j \alpha_j + \tilde Z_{j1} \tilde \nu_{j1} + \cdots + \tilde Z_{jq} \tilde \nu_{jq} + \epsilon_j = \tilde Z_j \tilde  \nu_j + \epsilon_j, \quad \epsilon_j \sim N(0, \tau^2_j I),  \quad \tilde  \nu_{ja} \sim N(0, R_{ja}),\nonumber\\
  \tilde \nu_{ja} &\sim N(0, R_{ja}), \quad \tilde \nu \sim N(0, R_j), \quad p(\alpha_j, \Gamma_j, \tau^2_j) \propto \tau_j^{-2}, \quad p(\theta_{aj}) = \text{Uniform}(\underline c_{a}, \overline c_a).
\end{align}
where the prior distributions on $\theta_{ja}$s remain unchanged. The prior distribution of the subset $j$ latent variables $\tilde \nu_j$ and parameters $\alpha_j, \Gamma_j, \tau^2_j, \theta_{j1}, \ldots, \theta_{jq}$ also have the same form as in \eqref{eq-ap-full:3}
\begin{align}
  \label{eq-ap-sub:31}
  p(\tilde \nu_j, \alpha_j, \Gamma_j, \tau_j^2, \theta_{j1}, \ldots, \theta_{jq}) &= \prod_{a=1}^q \{ p(\tilde \nu_{ja} \mid \theta_{ja}) p(\theta_{ja}) \} p(\alpha_j, \Gamma_j, \tau_j^2) \nonumber\\
  &\equiv p(\tilde \nu_j, \theta_{j1}, \ldots, \theta_{jq}) p(\alpha_j, \Gamma_j, \tau_j^2).
\end{align}

The sampling algorithm for drawing $\beta(\cdot)$, $y(\cdot)$, and $\tau^2$ from their respective subset $j$ posterior distributions is based on a modified form of \eqref{eq-ap-sub:21} and \eqref{eq-ap-sub:31}. The conditional density of $(\alpha_j, \Gamma_j, \tau^2_j, \theta_{j1}, \ldots, \theta_{jq})$ given $y_j$ after stochastic approximation is defined following \eqref{eq:fullsub-j} as
\begin{align}
  \label{eq-ap-sub:2}
  \pi(\alpha_j, \Gamma_j, \tau^2_j, \theta_{j1}, \ldots, \theta_{jq} \mid y_j) &\propto \{p(y_j \mid \tilde \nu_j , \alpha_j, \Gamma_j, \tau^2_j) p(\tilde \nu_j, \theta_{j1}, \ldots, \theta_{jq}) \}^{\delta_n} p(\alpha_j, \Gamma_j, \tau^2_j),
\end{align}
where $p(\tilde \nu_j, \theta_{j1}, \ldots, \theta_{jq})$ is the subset $j$ version of $p(\nu, \theta_{1}, \ldots, \theta_{q})$  in \eqref{eq-ap-full:3}. The model in \eqref{eq-ap-sub:21} implies that the joint distribution of $(y_j, \tilde \nu_j)$  given $\alpha_j, \Gamma_j, \tau^2_j, \theta_{j1}, \ldots, \theta_{jq}$ is an $(\tilde s_j+mq)$-variate Gaussian distribution with mean $(X_j \alpha_j, 0)$, where 0 is an $mq$-by-1 vector, and covariance matrix $\overline C_j$, where the blocks corresponding to the marginal covariance matrices of $y_j$, $\tilde \nu_j$, respectively, and their cross covariance matrix are
\begin{align}
  \label{eq-ap-sub:4}
  (\overline C_j)_{y_j y_j} = \tilde Z_j R_j \tilde Z_j^\T +  \tau^2_j I = \sum_{a=1}^q \tilde Z_{ja} R_{ja} \tilde Z_{ja}^\T + \tau^2_j I, \quad
  (\overline C_j)_{\tilde \nu_j \tilde \nu_j} = R_j, \quad (\overline C_j)_{y_j \tilde \nu_j} = \tilde Z_j R_j.
\end{align}
This implies that $\tilde \nu_j$ given $y_j$, $\alpha_j, \Gamma_j, \tau_j^2,\theta_{j1}, \ldots, \theta_{jq}$ follows $N(\mu_{\tilde \nu_j}, \Sigma_{\tilde \nu_j})$, where
\begin{align}
  \label{eq-ap-sub:5}
  \mu_{\tilde \nu_j} = R_j^\T \tilde Z_j^\T (\overline C_j)_{y_j y_j}^{-1} (y_j - X_j \alpha_j), \quad
  \Sigma_{\tilde \nu_j} = R_j - R_j^\T \tilde Z_j^\T  (\overline C_j)_{y_j y_j}^{-1}\tilde Z_j R_j.
\end{align}

The conditional density of $(\alpha_j, \Gamma_j, \tau_j^2)$ given $y_j$ is obtained using \eqref{eq-ap-sub:2}. The subset $j$ counterpart of $W$ and $b$ in \eqref{eq-ap-full:6} are $W_j$ and $b^\T_j = (\alpha^\T_j, \gamma^\T_j)$, where $\gamma_j$ is the column-wise vectorization of $\Gamma_j$. The conditional density in \eqref{eq-ap-sub:2} together with \eqref{eq-ap-full:61} implies that
\begin{align}
  \label{eq-ap-sub:7}
  \pi(\tau^2_j, b_j \mid y_j) &\propto \frac{1}{(\tau_j^2)^{\tilde s_j \delta_n /2}} \exp \left\{ - \frac{(y_j - W_j b_j)^\T (y_j - W_j b_j)}{2\tau_j^2 / \delta_n}  \right\} \frac{1} {\tau^2_j} \nonumber\\
                                   &= \frac{1}{(\tau_j^2)^{\tilde s_j \delta_n /2 + 1}} \exp \left\{ - \frac{(y_j - \hat y_j)^\T (y_j - \hat y_j) + (\hat y_j - W_j b_j)^\T (\hat y_j - W_j b_j)}{2\tau^2_j/\delta_n}  \right\} \nonumber\\
                                   &= \frac{1}{(\tau^2_j)^{\tilde s_j \delta_n /2 + 1}} \exp \left( - \frac{\delta_n \| y_j - \hat y_j \|_2^2 } {2\tau_j^2}  \right)
                                     \exp \left\{ - \frac{(\hat b_j - b_j)^\T \delta_n (W_j^\T W_j) (\hat b_j -b_j)}{2\tau_j^2}  \right\},
\end{align}
where $\hat b_j = (W^\T_j W_j)^{-1} W_j^\T y_j$ is the least squares estimate of $b_j$ and $\hat y_j = W_j \hat b_j$ is the mean estimate of $y_j$ based on $\hat b_j$. Marginalizing over $b_j$ in \eqref{eq-ap-sub:7} implies that
\begin{align}
  \label{eq-ap-sub:8}
  \pi(\tau_j^2 \mid y_j) &\propto \frac{1}{(\tau^2_j)^{(\tilde s_j \delta_n - p - q^2)/2 + 1}} \exp \left( - \frac{\delta_n \| y_j - \hat y_j \|_2^2 } {2\tau_j^2}  \right), \quad \tau^2_j \mid y_j  \sim \frac{\delta_n \| y_j - \hat y_j \|_2^2 } {\chi^2_{\tilde s_j \delta_n - p - q^2}},
\end{align}
and \eqref{eq-ap-sub:7} assuming $\tau_j^2$ is given implies that
\begin{align}
  \label{eq-ap-sub:9}
  \pi(b_j \mid y_j, \tau_j^2) &\propto
  \exp \left\{ - \frac{(\hat b_j - b_j)^\T \delta_n (W_j^\T W_j) (\hat b_j -b_j)}{2\tau_j^2}  \right\} , \nonumber\\
  b_j \mid y_j, \tau^2_j &\sim N \left\{ \hat b_j,  \frac{\tau^2_j} {\delta_n} (W_j^\T W_j)^{-1}  \right\};
\end{align}
therefore,  \eqref{eq-ap-sub:8} and  \eqref{eq-ap-sub:9} imply that the distribution of $(b_j, \tau_j^2)$ given $y_j$ lies in the Normal-Inverse-Gamma family. The $\alpha_j$ draw corresponds to the first $p$ elements of $b_j$ and the remaining elements of $b$ are ``un-vectorized'' into the $q$-by-$q$ matrix $\Gamma_j$.

Finally, the form of the full conditionals of $\theta_{j1}, \ldots, \theta_{jq}$ and $\beta(u^*), y(u^*)$ for $u^* \in \Ucal^*$ remain the same as in steps (d)--(f) of the sampling algorithm for drawing from the full data posterior distribution, except the log likelihood for $\theta_{j1}, \ldots, \theta_{jq}$ is multiplied by a factor $\delta_n$ in step (d). Let $\beta_j(u^*), y_j(u^*), \nu_j(u^*) = \{\nu_{j1}(u^*), \ldots, \nu_{jq}(u^*)\}$, and $R_{ja*}, R_{ja**}, \nu_{ja}^*$ ($a=1, \ldots, q$) be the subset $j$ counterparts of $\beta(u^*), y(u^*), \nu(u^*)$ and $R_{a*}, R_{a**}, \nu_{a}^*$. The full conditional of $\nu_j$ in \eqref{eq-ap-sub:2} implies that the conditional likelihood of $\theta_{j1}, \ldots, \theta_{jq}$ given $\tilde \nu_j$ remains unchanged after stochastic approximation. Similarly, the GP prior on $\nu_a(\cdot)$ remains unchanged except that $\theta_{ja}$ replaces $\theta_a$ ($a=1, \ldots, q$). This implies that the steps for drawing from the  full conditionals of $\theta_{j1}, \ldots, \theta_{jq}$ and $\beta(u^*), y(u^*)$ for $u^* \in \Ucal^*$  have the same form as in the steps (d)--(f).

In summary, the sampling algorithm for drawing from the $j$th subset posterior distribution of $(\alpha, \Gamma, \tau^2, \theta_1, \ldots, \theta_q)$ and $\beta(u^*), y(u^*)$ for $u^* \in \Ucal^*$ starts from an initial value of parameters $(\alpha^{(0)}_j, \Gamma^{(0)}_j, \tau^{2(0)}_j, \theta^{(0)}_{j1}, \ldots, \theta^{(0)}_{jq})$  and cycles through the following six steps for $t=0, 1, \ldots, \infty$:
\begin{enumerate}
\item[(a*)] draw $\tilde \nu_j^{(t+1)}$ given $y_j, \alpha_j^{(t)}, \Gamma_j^{(t)}, \tau_j^{2(t)}, \theta^{(t)}_{jq}, \ldots, \theta^{(t)}_{jq}$ from $N(\mu_{\tilde \nu_j}, \Sigma_{\tilde \nu_j})$, where $\mu_{\tilde \nu_j}, \Sigma_{\tilde \nu_j}$  are defined in \eqref{eq-ap-sub:5};
\item[(b*)] draw $\tau_j^{2(t+1)}$ given $y_j, \tilde \nu^{(t+1)}_j$ using \eqref{eq-ap-sub:8};
\item[(c*)] draw $(\alpha_j^{(t+1)}, \gamma_j^{(t+1)})$ given $y_j, \tilde \nu_j^{(t+1)}, \tau_j^{2(t+1)}$ using \eqref{eq-ap-sub:9} and the vectorization of $\gamma_j^{(t+1)}$ is reversed to obtain $\Gamma_j^{(t+1)}$;
\item[(d*)] draw $\theta_{j1}^{(t+1)}, \ldots, \theta_{jq}^{(t+1)}$ given $\tilde \nu_j^{(t+1)}$ using ESS as in step (d) of the full data posterior sampling algorithm after multiplying the log likelihood by $\delta_n$;
\item [(e*)] draw $\nu_{ja}^{*(t+1)}$ and obtain $\beta_j(u^*)^{(t+1)}$ given $\tilde \nu_{ja}^{(t+1)}$, $\alpha^{(t+1)}$, $\Gamma^{(t+1)}$, $\theta_{ja}^{(t+1)}$ as in step (e) of the full data posterior sampling algorithm; and
\item [(f*)] draw $y_j(u^*)^{(t+1)}$ given $X(u^*), \beta_j(u^*)^{(t+1)}, \tau_j^{2(t+1)}$ as in step (f) of the full data posterior sampling algorithm.
\end{enumerate}

\section{Proof of Theorem \ref{thm:main}}
\label{app:proof}

We use $\vertiii{\cdot}$ to denote the matrix operator norm, and $\tr(\cdot)$ to denote the trace of a matrix or a covariance function. For two semi-positive definite matrices $A_1$ and $A_2$, $A_1\preceq A_2$ means $A_2-A_1$ is semi-positive definite.
For abbreviation, we write the Bayes $L_2$-risk in estimating $\beta_0(\cdot)$ using the combined posterior from Theorem \ref{thm:main} as
\begin{align} \label{L2_risk}
\LL(\overline \Pi)& \equiv \EE_{u^*}\EE_{y,u} \EE_{\overline \beta\mid y,u} \left\|\overline \beta(u^*) - \beta_0(u^*) \right\|_2^2,
\end{align}
where $\overline \beta(\cdot)$ be a $q$-variate random function drawn from the AMC posterior of $\beta(\cdot)$. Let $\EE_{u^*}$, $\EE_{y_j,u_j}$, $\EE_{y,u}$, $\EE_{\beta_j \mid y_j,u_j}$ (similarly $\EE_{\nu_j \mid y_j,u_j}$  and $\EE_{\nu_{ja} \mid y_j,u_j}$ for $a=1,\ldots,q$), $\EE_{\overline \beta \mid y,u}$ (similarly $\EE_{\overline \nu \mid y,u}$ and $\EE_{\overline \nu_a \mid y,u}$ for $a=1,\ldots,q$) respectively be the expectations with respect to the distribution of $u^*$, the true data generating distribution of the subset data $(y_j,u_j)$, the true data generating distribution of the full data $(y,u)$, the subset posterior distribution of $\beta$ (similarly $\nu$ and $\nu_a$) given $y_j,u_j$ after stochastic approximation, and the AMC posterior distribution of $\overline \beta$ (similarly $\overline \nu$ and $\overline \nu_a$) given the full data. Notations for the variances are similarly defined.

We begin by decomposing this Bayes $L_2$-risk. Given the relation $\beta(\cdot)=\Gamma_0\nu(\cdot)$, it follows that
\begin{align} \label{L2_risk1}
\LL(\overline \Pi)& \leq \vertiii{\Gamma_0}^2 \EE_{u^*}\EE_{y,u} \EE_{\overline \beta\mid y,u} \left\|\overline \nu(u^*) - \nu_0(u^*) \right\|_2^2,
\end{align}
where $\vertiii{\Gamma_0}<\infty$ given Assumption \ref{prior_assumption}. Therefore, it suffices to study the Bayes $L_2$ risk of $\nu(\cdot)$.

Based on Assumption \ref{combine_assumption}, we have the following decomposition: For any $u^*\in [0,1]^d$,
\begin{align} \label{combine_subset_bias_var1}
& \EE_{\overline \nu \mid y,u} \left\|\overline \nu(u^*) - \nu_0(u^*)\right\|_2^2 \nonumber \\
={}&  \EE_{\overline \nu \mid y,u} \left \|\overline \nu(u^*) - \EE_{\overline \nu \mid y,u}\{\overline \nu(u^*)\} + \EE_{\overline \nu\mid y,u}\{\overline \nu(u^*)\} - \nu_0(u^*) \right \|_2^2 \nonumber \\
={}& \EE_{\overline \nu\mid y,u} \left\| \overline \nu(u^*) - \EE_{\overline \nu\mid y,u}\{\overline \nu(u^*)\} \right \|_2^2 + \left\| \EE_{\overline \nu\mid y,u}\{\overline \nu(u^*)\} - \nu_0(u^*) \right\|_2^2 \nonumber \\
={}& \sum_{a=1}^q {\var}_{\overline \nu_a \mid y,u} \left \{\overline \nu_a (u^*) \right \} + \left\| \EE_{\overline \nu\mid y,u}\{\overline \nu(u^*)\} - \nu_0(u^*) \right\|_2^2 \nonumber \\
\leq{}& \frac{\overline c}{k} \sum_{a=1}^q \sum_{j=1}^k {\var}_{\nu_{ja} \mid y_j,u_j}\{\nu_{ja}(u^*)\} \nonumber \\
&\quad+ \sum_{a=1}^q \Bigg\{ \frac{1}{k}\sum_{j=1}^k \left[\EE_{\nu_{ja}\mid y_j,u_j}\{\nu_{ja}(u^*)\} -\nu_{0a}(u^*) \right] + O_p\left(n^{-1/2}\right) \Bigg\}^2 \nonumber \\
\leq{}& \frac{\overline c}{k} \sum_{a=1}^q \sum_{j=1}^k {\var}_{\nu_{ja} \mid y_j,u_j}\{\nu_{ja}(u^*)\} \nonumber \\
&\quad+ 2\sum_{a=1}^q \Bigg\{ \frac{1}{k}\sum_{j=1}^k \left[\EE_{\nu_{ja}\mid y_j,u_j}\{\nu_{ja}(u^*)\} -\nu_{0a}(u^*) \right]\Bigg\}^2 + O_p\left(n^{-1}\right),
\end{align}
where the last inequality follows from $(a+b)^2\leq 2(a^2+b^2)$. Using this inequality again, we have that
\begin{align} \label{combine_subset_bias_var2_0}
& \EE_{u^*}\EE_{y,u} \EE_{\overline \nu \mid y,u} \left\|\overline \nu(u^*) - \nu_0(u^*)\right\|_2^2 \nonumber \\
\leq{}& \frac{\overline c}{k} \sum_{a=1}^q \sum_{j=1}^k \EE_{u^*}\EE_{y_j,u_j} {\var}_{\nu_{ja} \mid y_j,u_j}\{\nu_{ja}(u^*)\}  \nonumber \\
&\quad + 2 \sum_{a=1}^q \EE_{u^*}\EE_{y,u} \sum_{a=1}^q \Bigg\{ \frac{1}{k}\sum_{j=1}^k \left[\EE_{\nu_{ja}\mid y_j,u_j}\{\nu_{ja}(u^*)\} -\nu_{0a}(u^*) \right]\Bigg\}^2 + O_p\left(n^{-1}\right) \nonumber \\
={}&  \frac{\overline c}{k} \sum_{a=1}^q \sum_{j=1}^k \EE_{u^*}\EE_{y_j,u_j} {\var}_{\nu_{ja} \mid y_j,u_j}\{\nu_{ja}(u^*)\}  \nonumber \\
&+ 2 \sum_{a=1}^q \EE_{u^*}\EE_{y,u} \Bigg\{ \frac{1}{k}\sum_{j=1}^k \left[\EE_{\nu_{ja}\mid y_j,u_j}\{\nu_{ja}(u^*)\} - \EE_{\nu_{ja}, y_j\mid u_j}\{\nu_{ja}(u^*)\} \right] \nonumber \\
& + \frac{1}{k}\sum_{j=1}^k \left[\EE_{\nu_{ja}, y_j\mid u_j}\{\nu_{ja}(u^*)\} - \nu_{0a}(u^*) \right] \Bigg\}^2 + O_p\left(n^{-1}\right) \nonumber \\
\leq{}& \frac{\overline c}{k} \sum_{a=1}^q \sum_{j=1}^k \EE_{u^*}\EE_{y_j,u_j} {\var}_{\nu_{ja} \mid y_j,u_j}\{\nu_{ja}(u^*)\} \nonumber \\
&+ 4\sum_{a=1}^q \EE_{u^*}\EE_{y,u} \Bigg(\frac{1}{k}\sum_{j=1}^k \left[\EE_{\nu_{ja}\mid y_j,u_j}\{\nu_{ja}(u^*)\} - \EE_{\nu_{ja}, y_j\mid u_j}\{\nu_{ja}(u^*)\} \right] \Bigg)^2 \nonumber \\
& + 4\sum_{a=1}^q \EE_{u^*}\EE_{u}  \Bigg(\frac{1}{k}\sum_{j=1}^k \left[\EE_{\nu_{ja}, y_j\mid u_j}\{\nu_{ja}(u^*)\} - \nu_{0a}(u^*) \right] \Bigg)^2 + O_p\left(n^{-1}\right).
\end{align}

Furthermore, using the sampling independence between subsets of $u_j$ as in Assumption \ref{sampling}, the second term in \eqref{combine_subset_bias_var2_0} can be simplified as
\begin{align} \label{eq:cross2}
& 4\sum_{a=1}^q \EE_{u^*}\EE_{y,u} \Bigg(\frac{1}{k}\sum_{j=1}^k \left[\EE_{\nu_{ja}\mid y_j,u_j}\{\nu_{ja}(u^*)\} - \EE_{\nu_{ja}, y_j\mid u_j}\{\nu_{ja}(u^*)\} \right] \Bigg)^2 \nonumber \\
={}& \frac{4}{k^2}\sum_{a=1}^q  \sum_{j=1}^k \EE_{u^*} \EE_{y_j,u_j} \left[\EE_{\nu_{ja}\mid y_j,u_j}\{\nu_{ja}(u^*)\} - \EE_{\nu_{ja}, y_j\mid u_j}\{\nu_{ja}(u^*)\} \right]^2 \nonumber  \\
&~~ +  \frac{4}{k^2}\sum_{a=1}^q  \sum_{j\neq j'} \EE_{u^*} \EE_{u} \EE_{u_j|u} \EE_{u_{j'}|u}  \Bigg(\left[\EE_{y_j\mid u_j}\EE_{\nu_{ja}\mid y_j,u_j}\{\nu_{ja}(u^*)\} - \EE_{\nu_{ja}, y_j\mid u_j}\{\nu_{ja}(u^*)\} \right] \nonumber  \\
&\quad \times \left[\EE_{y_{j'}\mid u_{j'}}\EE_{\nu_{j'a}\mid y_{j'},u_{j'}}\{\nu_{j'a}(u^*)\} - \EE_{\nu_{j'a}, y_{j'}\mid u_{j'}}\{\nu_{j'a}(u^*)\} \right] \Bigg)\nonumber  \\
={}& \frac{4}{k^2}\sum_{a=1}^q  \sum_{j=1}^k \EE_{u^*} \EE_{y_j,u_j} \left[\EE_{\nu_{ja}\mid y_j,u_j}\{\nu_{ja}(u^*)\} - \EE_{\nu_{ja}, y_j\mid u_j}\{\nu_{ja}(u^*)\} \right]^2 \nonumber \\
={}& \frac{4}{k^2}\sum_{a=1}^q  \sum_{j=1}^k \EE_{u^*}\EE_{u_j} \left( {\var}_{y_j\mid u_j} \left[\EE_{\nu_{ja} \mid y_j,u_j}\{\nu_{ja}(u^*)\} \right]\right),
\end{align}
where the second term after the first equal sign is zero due to the independence between $u_j$ and $u_{j'}$ given the full data $u$ according to Assumption \ref{sampling}. The third term in \eqref{combine_subset_bias_var2_0} can be bounded from above by the Cauchy-Schwarz inequality:
\begin{align} \label{eq:cross3}
& 4\sum_{a=1}^q \EE_{u^*}\EE_{u}  \Bigg(\frac{1}{k}\sum_{j=1}^k \left[\EE_{\nu_{ja}, y_j\mid u_j}\{\nu_{ja}(u^*)\} - \nu_{0a}(u^*) \right] \Bigg)^2 \nonumber \\
&\leq \frac{4}{k} \sum_{a=1}^q \sum_{j=1}^k \EE_{u^*} \EE_{u_j} \left[\EE_{\nu_{ja}, y_j\mid u_j}\{\nu_{ja}(u^*)\} - \nu_{0a}(u^*) \right]^2.
\end{align}
We can plug in \eqref{eq:cross2}, and \eqref{eq:cross3} to \eqref{combine_subset_bias_var2_0} to further obtain that
\begin{align} \label{combine_subset_bias_var2}
& \EE_{u^*}\EE_{y,u} \EE_{\overline \nu \mid y,u} \left\|\overline \nu(u^*) - \nu_0(u^*)\right\|_2^2 \nonumber \\
\leq {}& \frac{4}{k} \sum_{a=1}^q \sum_{j=1}^k \EE_{u^*} \EE_{u_j} \left[\EE_{\nu_{ja}, y_j\mid u_j}\{\nu_{ja}(u^*)\} - \nu_{0a}(u^*) \right]^2 \nonumber \\
&+ \frac{4}{k^2}\sum_{a=1}^q  \sum_{j=1}^k \EE_{u^*}\EE_{u_j} \left( {\var}_{y_j\mid u_j} \left[\EE_{\nu_{ja} \mid y_j,u_j}\{\nu_{ja}(u^*)\} \right]\right) \nonumber \\
&+ \frac{\overline c}{k} \sum_{a=1}^q \sum_{j=1}^k \EE_{u^*}\EE_{y_j,u_j}{\var}_{\nu_{ja} \mid y_j,u_j}\{\nu_{ja}(u^*)\} + O_p\left(n^{-1}\right) \nonumber \\
\leq {}&\frac{4}{k} \sum_{j=1}^k \EE_{u^*} \EE_{u_j} \left\| \EE_{\nu_{j}, y_j\mid u_j}\{\nu_{j}(u^*)\} - \nu_{0}(u^*) \right\|_2 ^2 \nonumber \\
&+ \frac{4}{k^2} \sum_{j=1}^k \EE_{u^*}\EE_{u_j} \tr\left( {\var}_{y_j\mid u_j} \left[\EE_{\nu_{j} \mid y_j,u_j}\{\nu_{j}(u^*)\} \right]\right) \nonumber \\
&+ \frac{\overline c}{k} \sum_{j=1}^k \EE_{u^*}\EE_{y_j,u_j}\tr\left({\var}_{\nu_{j} \mid y_j,u_j}\{\nu_{j}(u^*)\}\right) + O_p\left(n^{-1}\right).
\end{align}
Now we provide the detailed expressions for three terms on the right-hand side of \eqref{combine_subset_bias_var2}. We define the following notations. Let $\Gamma_{0a}$ be the $a$th column of $\Gamma_0$, $a=1,\ldots,q$. In the derivation below, $\rho_{a}(\cdot,\cdot)$ is the correlation function with its parameter fixed at $\theta_{0a}$, $a=1,\ldots,q$. For $j=1, \ldots, k$ and $a=1, \ldots q$,
\begin{align} \label{eq:be1}
R_{ja} (u^*) &= \{\rho_{ja}(u_{j1}, u^*), \ldots, \rho_{a}(u_{jm}, u^*)\}^\T  \in \RR^{m} , \nonumber\\
R_j(u^*) &= \diag\{R^\T_{j1} (u^*), \ldots, R^\T_{jq} (u^*) \} \in \RR^{q \times qm} ,  \nonumber\\
R(u,u') &= \diag\{\rho_1(u,u'),\ldots,\rho_q(u,u')\} \in \RR^{q \times q}, \quad u,u'\in [0,1]^d, \nonumber \\
\tilde Z(\cdot) &= Z(\cdot) \Gamma_0 \in \RR^{1\times q},\quad \tilde Z_a(\cdot) = Z(\cdot)\Gamma_{0a} \in \RR, \nonumber \\
\tilde Z_{ja} &= \diag \left\{\tilde Z_a(u_{j1}), \ldots, \tilde Z_a(u_{jm}) \right\} \in \RR^{m \times m},  \quad \tilde Z_j = \left\{\tilde Z_{j1}, \ldots, \tilde Z_{jq}\right\} \in \RR^{m \times qm}, \nonumber\\
(R_{ja})_{ii'} &= \rho_a(u_{ji}, u_{ji'}), \; i,i' \in \{1, \ldots, m\}, \quad \tilde R_{jj} = \diag(R_{j1}, \ldots, R_{jq}) \in \RR^{qm \times qm}, \nonumber\\
\tilde \nu_{ja} &= \{\nu_{a}(u_{j1}), \ldots, \nu_{a}(u_{jm})\}^\T  \in \RR^{m}, \quad \tilde \nu_j^\T = (\tilde \nu_{1j}^\T, \ldots, \tilde \nu_{qj}^\T) \in \RR^{qm}, \nonumber \\
\nu_{0ja} &= \{\nu_{0a}(u_{j1}), \ldots, \nu_{0a}(u_{jm})\}^\T \in \RR^{m}.
\end{align}

Based on the subset model \eqref{subset_model} and the standard GP derivation, we have that for $a=1,\ldots,q$,
\begin{align}
& \EE_{\nu_{j}, y_j\mid u_j}\{\nu_{j}(u^*)\} - \nu_{0}(u^*) = R_{j}(u^*) \tilde Z_j^\T \left(\tilde Z_j \tilde R_{jj} \tilde Z_j^\T + \frac{\tau_0^2\lambda_n}{k} I_{m} \right)^{-1}\tilde Z_j \tilde \nu_j - \nu_0(u^*), \label{subset_bias1}  \\
& {\var}_{y_j\mid u_j} \left[\EE_{\nu_{j} \mid y_j,u_j}\{\nu_{j}(u^*)\}\right] = {\var}_{y_j\mid u_j} \left[R_{j}(u^*) \tilde Z_j^\T \left(\tilde Z_j \tilde R_{jj} \tilde Z_j^\T + \frac{\tau_0^2\lambda_n}{k} I_{m} \right)^{-1} y_j  \right] \nonumber \\
&\qquad = \tau_0^2 R_{j}(u^*) \tilde Z_j^\T \left( \tilde Z_j \tilde R_{jj} \tilde Z_j^\T + \frac{\tau_0^2\lambda_n}{k} I_{m} \right)^{-2} \tilde Z_j R_{j}(u^*)^\T, \label{subset_var1} \\
& {\var}_{\nu_{j} \mid y_j,u_j}\{\nu_{j}(u^*)\}  = \lambda_n^{-1} \bigg\{ R(u^*,u^*) - R_{j}(u^*) \tilde Z_j^\T \left(  \tilde Z_j \tilde R_{jj} \tilde Z_j^\T + \frac{\tau_0^2\lambda_n}{k} I_{m} \right)^{-1} \tilde Z_j R_{j}(u^*)^\T \bigg\}. \label{subset_var2}
\end{align}
In the following lemmas, we give upper bounds for \eqref{subset_bias1}, \eqref{subset_var1}, and \eqref{subset_var2}, respectively. We first introduce some additional notations. In Assumption \ref{z_assumption}, we assume that the largest eigenvalue of $\Omega$ is bounded above by a constant $\overline c_{\Omega}$ and its smallest eigenvalue is bounded below by a constant $\underline c_{\Omega}$, for $0<\underline c_{\Omega}< \overline c_{\Omega}<\infty$ and for all $a=1,\ldots,q$. By the definition of $\tilde Z(\cdot)=Z(\cdot)\Gamma_0$ and $\tilde Z_a(\cdot)=Z(\cdot) \Gamma_{0a}$ ($a=1,\ldots,q$) and Assumption \ref{z_assumption}, we have that $|\tilde Z_a(u)|\leq C_Z \|\Gamma_{0a}\|_2 \leq C_Z \vertiii{\Gamma_0}$ for all $a=1,\ldots,q$ and all $u\in [0,1]^d$. We define the positive constant $\tilde C_Z \equiv C_Z \vertiii{\Gamma_0}$ for the derivations below.
We define the $\RR^q$-valued functional space $L_2^q(\dd u)$ for $\mathrm{f}=(f_1,\ldots,f_q)^\T$ with $f_1,\ldots,f_q\in L_2(\dd u)$. The $q$-variate RKHS $\HH$ is defined to the space of functions with the finite $\HH$-norm $\|\mathrm{f}\|_{\HH}^2=\sum_{a=1}^q \|f_a\|_{\HH_a}^2$ if $\mathrm{f}=(f_1,\ldots,f_q)^\T\in L_2^q(\dd u)$ and $f_a\in \HH_a$ for $a=1,\ldots,q$. This $q$-variate RKHS satisfies the reproducing property that if $\mathrm{f}\in \HH$, then for any $u\in [0,1]^d$ and any $\mathrm{c}\in \RR^q$, $\langle \mathrm{f}, R(u,\cdot)\mathrm{c} \rangle_{\HH} = \mathrm{f}(u)^\T \mathrm{c}$; see Section 3.2 of \citet{Alvetal12}. For abbreviation, we let $\Lambda(c,h) = \sum_{i=1}^{h} \mu_{i*}/(\mu_{i*} + c) $ for any $c>0$, any positive integer $h$, and $\mu_{i*}$ as defined in Assumption \ref{eigen_assumption}. For $a=1,\ldots,q$ and any positive integer $h$, let $\tr(\rho_a)=\sum_{i=1}^{\infty}\mu_{ai}$, $\tr(\rho_{a,h})=\sum_{i=h+1}^{\infty}\mu_{ai}$, $\Tr(\rho)=\sum_{a=1}^q \tr(\rho_a)$, and ${\Tr}(\rho,h)=\sum_{a=1}^q \tr(\rho_{a,h})$.

\begin{lemma} \label{lem:bias_bound}
Suppose that Assumptions \ref{sampling}--\ref{combine_assumption} hold. Then for every $j=1,\ldots,k$,
\begin{align*}
& \EE_{u^*} \EE_{u_j} \left\| \EE_{\nu_{j}, y_j\mid u_j}\{\nu_{j}(u^*)\} - \nu_{0}(u^*) \right\|_2^2 \\
\leq{}& 8\frac{\tau_0^2\lambda_n}{\underline c_{\Omega}km} \|\nu_0\|_{\HH}^2 + 8 \frac{ km}{\tau_0^2 \lambda_n} C_{\varphi}^4 \tilde C_Z^4 \overline c_{\Omega} \underline c_{\Omega}^{-2} \|\nu_0\|_{\HH}^2 \Tr(\rho) \Tr(\rho,\bar h) +  \mu_{(\bar h+1)*}\|\nu_0\|_{\HH}^2  \nonumber \\
&\quad + 2q\bar h \|\nu_0\|_{\HH}^2 \Tr(\rho) \exp\left\{-\frac{m}{8(B^2+B) } \right\}, \nonumber
\end{align*}
where $B=C_{\varphi}^2 \tilde C_Z^2 \underline c_{\Omega}^{-1} q \Lambda\left(\tau_0^2\lambda_n/(\underline c_{\Omega}km),\bar h\right)+ 1$.
\end{lemma}

\begin{proof}[Proof of Lemma \ref{lem:bias_bound}]
Consider the subset posterior distribution of $\nu_j(\cdot)$ on the subset $j$. The $j$th subset posterior distribution of $\nu_j(u^*)$ has the mean
\begin{align} \label{eq:t1e1}
\hat\nu_j(u^*) &= R_j(u^*) \tilde Z_j^\T \left(\tilde Z_j \tilde R_{jj} \tilde Z_j^\T + \frac{\tau_0^2\lambda_n}{k} I_m \right)^{-1}  y_j.
\end{align}
We view $\hat\nu_j(u^*)=\{\hat\nu_{j1}(u^*),\ldots,\hat\nu_{jq}(u^*)\}^\T$ as a $q$-variate function of $u^*$. Define the $[0,1]^d \mapsto \RR^q$ operators $\Delta_j$ ($j = 1, \ldots, k$) and $\Delta$ as
\begin{align} \label{eq:Deltaj}
\Delta_j(\cdot) &= R_j(\cdot) \tilde Z_j^\T \left(\tilde Z_j \tilde R_{jj} \tilde Z_j^\T + \frac{\tau_0^2\lambda_n}{k} I_m \right)^{-1}  y_j - \nu_0(\cdot) \equiv \hat\nu_j (\cdot) - \nu_0(\cdot), \nonumber\\
\Delta(\cdot)
&= \frac{1}{k} \sum_{j=1}^k \left\{ \hat\nu_j (\cdot) - \nu_0(\cdot) \right\}
= \frac{1}{k} \sum_{j=1}^k \Delta_j(\cdot).
\end{align}
For $a=1,\ldots,q$, $\Delta_{ja}$ and $\Delta_{a}$ denote the $a$th components of $\Delta_j$ and $\Delta$, respectively. We can recognize that for $j = 1,\ldots,k$, $\hat\nu_j(\cdot)$ in \eqref{eq:Deltaj} is the solution of $q$-variate function to the optimization problem (see Section 3.2 in \citealt{Alvetal12})
\begin{align} \label{eq:wjopt}
{\text{argmin}}_{\nu \in \HH } \left[ \sum_{i=1}^m \left\{ y(u_{ji}) - \tilde Z(u_{ji}) \nu(u_{ji})\right\}^2 + \frac{\tau_0^2\lambda_n}{k} \| \nu \|_{\HH}^2 \right].
\end{align}
Now fix $a\in \{1,\ldots,q\}$. Taking the Frech\'et derivative of \eqref{eq:wjopt} with respect to $\nu_a$ and plugging in $\hat\nu_j$ gives
\begin{align} \label{eq:ba21}
& \sum_{i=1}^m  \left\{ \tilde Z (u_{ji}) \hat\nu_j(u_{ji}) - y(u_{ji}) \right\} \tilde Z_a(u_{ji}) \rho_a(u_{ji},\cdot) + \frac{\tau_0^2 \lambda_n}{k} \hat \nu_{ja}(\cdot) =0 .
\end{align}
Stacking this up across $a=1,\ldots,q$ gives
\begin{align} \label{eq:ba22}
& \sum_{i=1}^m  \left\{ \tilde Z (u_{ji}) \hat\nu_j(u_{ji}) - y(u_{ji}) \right\} R(u_{ji},\cdot) \tilde Z(u_{ji})^\T + \frac{\tau_0^2 \lambda_n}{k} \hat \nu_{j}(\cdot) =0 .
\end{align}
Since the true model assumes that
$$y(u_{ji}) = \tilde Z(u_{ji}) \nu_0(u_{ji}) + \epsilon(u_{ji}) = \langle \nu_0, R(u_{ji},\cdot) \tilde Z(u_{ji})^\T  \rangle_{\HH} + \epsilon(u_{ji}),$$
so
\begin{align} 
\tilde Z(u_{ji}) \hat \nu(u_{ji}) - y(u_{ji}) &= \tilde Z(u_{ji}) \{\hat \nu_j(u_{ji}) - \nu_0(u_{ji})) - \epsilon(u_{ji}) \nonumber \\
&= \langle \Delta_j(\cdot), R(u_{ji},\cdot) \tilde Z(u_{ji})^\T  \rangle_{\HH} - \epsilon(u_{ji}), \nonumber
\end{align}
Hence, we take expectations $\EE_{\nu_j, y_j \mid u_j}$ on both sides of \eqref{eq:ba21} and obtain that
\begin{align} \label{eq:ba6}
0={}& \sum_{i=1}^m  \EE_{\nu_j, y_j \mid u_j} \left\{ \tilde Z (u_{ji}) \hat\nu_j(u_{ji}) - y(u_{ji}) \right\} R(u_{ji},\cdot) \tilde Z(u_{ji})^\T + \frac{\tau_0^2 \lambda_n}{k} \EE_{\nu_j, y_j \mid u_j} \left\{\hat \nu_{j}(\cdot)\right\} \nonumber \\
={}& \sum_{i=1}^m \langle \EE_{\nu_j, y_j \mid u_j} \left\{\Delta_j(\cdot)\right\}, R(u_{ji},\cdot) \tilde Z(u_{ji})^\T \rangle_{\HH} R(u_{ji},\cdot) \tilde Z(u_{ji})^\T \nonumber \\
& - \EE_{\nu_j, y_j \mid u_j}\{\epsilon(u_{ji})\} R(u_{ji},\cdot) \tilde Z(u_{ji})^\T  + \frac{\tau_0^2 \lambda_n}{k} \EE_{ \nu_j, y_j \mid u_j} \left\{\hat \nu_{j}(\cdot)\right\} \nonumber \\
={}& \sum_{i=1}^m \left\langle \EE_{\nu_j, y_j \mid u_j} \left\{\Delta_j(\cdot)\right\}, R(u_{ji},\cdot) \tilde Z(u_{ji})^\T \right\rangle_{\HH} R(u_{ji},\cdot) \tilde Z(u_{ji})^\T  + \frac{\tau_0^2 \lambda_n}{k} \EE_{\nu_j, y_j \mid u_j} \left\{\hat \nu_{j}(\cdot)\right\}.
\end{align}
Using \eqref{eq:Deltaj}, $\EE_{\nu_j, y_j \mid u_j} \left\{\hat\nu_j(\cdot)\right\}   = \EE_{\nu_j, y_j \mid u_j} \left\{\Delta_j(\cdot)\right\} + \nu_0$, and dividing by $m$ in \eqref{eq:ba6}, we obtain that
\begin{align} \label{eq:16}
& \frac{1}{m}\sum_{i=1}^m  \left\langle \EE_{\nu_j, y_j \mid u_j} \left\{\Delta_j(\cdot)\right\}, R(u_{ji},\cdot) \tilde Z(u_{ji})^\T \right\rangle_{\HH} R(u_{ji},\cdot) \tilde Z(u_{ji})^\T  + \frac{\tau_0^2 \lambda_n}{km} \EE_{\nu_j, y_j \mid u_j} \left\{\Delta_j(\cdot)\right\} \nonumber \\
={}& - \frac{\tau_0^2 \lambda_n}{km} \nu_0(\cdot).
\end{align}
If we define the $j$th subset covariance operator $\hat \Sigma_j = m^{-1} \sum_{j=1}^m  R(u_{ji},\cdot) \tilde Z(u_{ji})^\T \otimes R(u_{ji},\cdot) \tilde Z(u_{ji})^\T$, then \eqref{eq:16} reduces to
\begin{align} \label{eq:17}
& \left( \hat \Sigma_j  +  \frac{\tau_0^2 \lambda_n}{km} I_{m} \right) \EE_{\nu_j, y_j \mid u_j} \left\{\Delta_j(\cdot)\right\} = - \frac{\tau_0^2\lambda_n}{km} \nu_0(\cdot) \nonumber \\
&\implies \|\EE_{\nu_j, y_j \mid u_j}(\Delta_j) \|_{\HH} \leq \|  \nu_0 \|_{\HH}, \quad j = 1, \ldots, k,
\end{align}
where the last inequality follows because $\hat \Sigma_j$ is a positive semi-definite matrix.

The rest of the proof finds an upper bound for
$$\left\| \EE_{\nu_j,y_j \mid u_j}(\Delta_j) \right\|_2^2 = \sum_{a=1}^q \left\| \EE_{\nu_{ja},y_j \mid u_j}(\Delta_{ja}) \right\|_2^2.$$
The main idea is to reduce this problem to a finite dimensional one indexed by a chosen $\bar h=\lceil \frac{2\vv}{2\vv-d}\rceil \in \NN$ as specified in Assumption \ref{z_assumption}. For each $j =1, \ldots, k$ and $a=1,\ldots,q$, let $\delta_{ja} = \left(\delta_{ja1}, \ldots, \delta_{ja\bar h}, \delta_{ja(\bar h+1)}, \ldots, \delta_{j a\infty}\right)^\T$, such that
\begin{align} \label{eq:ba1}
&\EE_{\nu_{ja},y_j \mid u_j} \left\{\Delta_{ja}(\cdot)\right\} = \sum_{h=1}^{\infty} \delta_{jah} \phi_{ah}(\cdot), \quad \delta_{jah} = \left\langle \EE_{\nu_{ja},y_j \mid u_j} \left\{\Delta_{ja}(\cdot)\right\}, \phi_h(\cdot) \right\rangle_{L_2(\dd u)},\nonumber \\
& \quad \left\| \EE_{\nu_{ja},y_j \mid u_j} (\Delta_j) \right\|_2^2 = \sum_{h=1}^{\infty} \delta^2_{jah} .
\end{align}
Define the vectors
$$\delta_{ja}^{\downarrow}=\left(\delta_{ja1}, \ldots, \delta_{ja\bar h}\right)^\T,\quad \delta_{ja}^{\uparrow}=\left(\delta_{ja(\bar h+1)}, \ldots, \delta_{ja\infty}\right)^\T,$$
so $\| \EE_{\nu_{ja},y_j \mid u_j} (\Delta_j) \|_2^2 = \| \delta_{ja}^{\downarrow} \|_2^2 + \| \delta_{ja}^{\uparrow} \|_2^2$ and we upper bound $\| \EE_{\nu_{ja},y_j \mid u_j} (\Delta_{ja}) \|_2^2$ by separately upper bounding $\| \delta_{ja}^{\downarrow} \|_2^2$ and $\| \delta_{ja}^{\uparrow} \|_2^2 $. Using the expansion $\rho_a(u,u') = \sum_{h=1}^{\infty} \mu_h \varphi_{ah}(u) \varphi_{ah}(u')$ for any $u,u' \in [0,1]^d$, we have the following upper bound for $\| \delta_{ja}^{\downarrow} \|_2^2$ and $\| \delta_{ja}^{\uparrow} \|_2^2$:
\begin{align}
\left\| \delta_{ja}^{\downarrow} \right\|_2^2 &= \frac{\mu_{a1}}{\mu_{a1}} \sum_{h=1}^{\bar h} \delta_{jah}^2 \overset{(i)}{\leq} \mu_{a1} \sum_{h=1}^{\bar h} \frac{\delta_{jah}^2}{\mu_{ah}} \nonumber \\
&\leq \mu_{a1} \sum_{h=1}^{\infty} \frac{\delta_{jah}^2}{\mu_{ah}} \overset{(ii)}{=} \mu_{a1} \left\| \EE_{\nu_{ja},y_j \mid u_j} (\Delta_{ja}) \right\|_{\HH_a}^2  \overset{(iii)}{\leq} \tr(\rho_a) \|\nu_{0a}\|_{\HH_a}^2, \label{eq:ba2} \\
\left\| \delta_{ja}^{\uparrow} \right\|_2^2  &= \frac{\mu_{a(\bar h+1)}}{\mu_{a(\bar h+1)}} \sum_{h=\bar h+1}^{\infty} \delta_{jah}^2
\leq \mu_{a(\bar h+1)} \sum_{h=\bar h+1}^{\infty} \frac{\delta_{jah}^2} {\mu_{ah}} \nonumber \\
& \overset{(iv)}{\leq} \mu_{a(\bar h+1)} \left\| \EE_{\nu_{ja},y_j \mid u_j} (\Delta_{ja}) \right\|_{\HH_a}^2 \label{eq:ba3}
\end{align}
where $(i)$ follows from the decreasing eigenvalues $\mu_{a1}\geq \mu_{a2}\geq \ldots$, $(ii)$ and $(iv)$ follow because $\| \EE_{\nu_{ja},y_j \mid u_j} (\Delta_{ja}) \|_{\HH_a}^2 = \sum_{h=1}^{\infty} \delta_{jah}^2/\mu_{ah}$, $(iii)$ follow from $\mu_{a1}\leq \tr(\rho_a)$.

Let
\begin{align*}
&\delta_j =(\delta_{j1}^\T,\ldots,\delta_{jq}^\T)^\T , \quad
\delta_j^{\downarrow} = (\delta_{j1}^{\downarrow \T},\ldots,\delta_{jq}^{\downarrow \T})^\T ,\quad
\delta_j^{\uparrow} = (\delta_{j1}^{\uparrow \T},\ldots,\delta_{jq}^{\uparrow \T})^\T .
\end{align*}
We now derive a more refined upper bound than \eqref{eq:ba2} for $\| \delta_{j}^{\downarrow} \|_2^2 = \sum_{a=1}^q \| \delta_{ja}^{\downarrow} \|_2^2$, such that the upper bound converges to zero as $m\to\infty$. For a given positive integer $\bar h$, each $a=1,\ldots,q$, and each $j=1,\ldots,k$, let $M_a = \diag(\mu_{a1}, \ldots, \mu_{a\bar h})\in \RR^{\bar h \times \bar h}$, $M=\diag\{M_1,\ldots,M_q\} \in \RR^{q\bar h \times q\bar h}$. Let $\Phi_j \in \RR^{qm \times q\bar h}$ and $\Phi_{ja}\in \RR^{m\times \bar h}$ be matrices such that
\begin{align} \label{eq:ba4}
&(\Phi_{ja})_{ih} =  \varphi_{ah}(u_{ji}), ~~i = 1, \ldots, m, ~~ h = 1, \ldots, \bar h, ~~ j = 1, \ldots, k, \nonumber\\
&\Phi_j =\diag\{\Phi_{j1},\ldots,\Phi_{jq}\} \in \RR^{qm\times q\bar h}.
\end{align}
Let $\nu_{0a}(\cdot) = \sum_{h=1}^{\infty} \zeta_{0ah} \varphi_{ah}(\cdot)$ with $\zeta_{0ah}=\langle \nu_{0a},\varphi_{ah}\rangle_{L_2(\dd u)}$ for $h=1,2,\ldots$ and $a=1,\ldots,q$. Also define the tail error vector $v_{ja} = (v_{ja1},\ldots, v_{jam})^\T \in \RR^m$ and $v_j=\left(v_{j1}^\T,\ldots,v_{jq}^\T\right)^\T \in \RR^{qm}$ ($j=1, \ldots, k$) such that
\begin{align*}
v_{jai} = \sum_{h = \bar h+1}^{\infty} \delta_{jah} \varphi_{ah}(u_{ji}), \quad i = 1, \ldots, m.
\end{align*}
For any $g \in \{1, \ldots, \bar h\}$ and any $a\in \{1,\ldots,q\}$, we take the $\HH_a$-inner product between $\varphi_{ag}$ and the $a$th component of \eqref{eq:16} and obtain that for $j = 1, \ldots, k$,
\begin{align} \label{eq:18}
& \frac{1}{m}\sum_{i=1}^m  \left\langle \EE_{\nu_j, y_j \mid u_j} \left\{\Delta_{j}(\cdot)\right\}, R(u_{ji},\cdot) \tilde Z(u_{ji})^\T \right\rangle_{\HH} \left\langle \rho_a(u_{ji},\cdot) \tilde Z_a(u_{ji}), \varphi_{ag}(\cdot) \right\rangle_{\HH_a}  \nonumber \\
&+ \frac{\tau_0^2 \lambda_n}{km} \left\langle \EE_{\nu_j, y_j \mid u_j} \left\{\Delta_{ja}(\cdot)\right\}, \varphi_{ag}(\cdot) \right\rangle_{\HH_a} = - \frac{\tau_0^2\lambda_n}{km} \langle \nu_{0a}, \varphi_{ag}  \rangle_{\HH_a},
\end{align}
which implies that
\begin{align} \label{eq:19}
& \frac{1}{m}\sum_{i=1}^m \left[\sum_{b=1}^q \tilde Z_b(u_{ji})\EE_{\nu_j, y_j \mid u_j} \left\{\Delta_{jb}(u_{ji})\right\}\right] \tilde Z_a(u_{ji}) \varphi_{ag}(u_{ji}) + \frac{\tau_0^2\lambda_n}{km}  \frac{\delta_{jag}}{\mu_{ag}} = - \frac{\tau_0^2\lambda_n}{km}  \frac{\zeta_{0ag}}{\mu_{ag}}.
\end{align}
The first term on the left-hand side of \eqref{eq:19} can be rewritten as
\begin{align} \label{eq:20}
& \frac{1}{m}\sum_{i=1}^m \left[\sum_{b=1}^q \tilde Z_b(u_{ji})\EE_{\nu_j, y_j \mid u_j} \left\{\Delta_{jb}(u_{ji})\right\}\right] \tilde Z_a(u_{ji}) \varphi_{ag}(u_{ji}) \nonumber \\
={}& \frac{1}{m}\sum_{i=1}^m \left[\sum_{b=1}^q \tilde Z_b(u_{ji}) \sum_{h=1}^{\bar h} \delta_{jbh}\varphi_{bh}(u_{ji}) \right] \tilde Z_a(u_{ji}) \varphi_{ag}(u_{ji}) \nonumber \\
&+\frac{1}{m}\sum_{i=1}^m \left[\sum_{b=1}^q \tilde Z_b(u_{ji}) \sum_{h=\bar h+1}^{\infty} \delta_{jbh}\varphi_{bh}(u_{ji}) \right] \tilde Z_a(u_{ji}) \varphi_{ag}(u_{ji}) \nonumber \\
={}& \frac{1}{m}\sum_{i=1}^m  \sum_{b=1}^q \left(\tilde Z_{ja} \Phi_{ja}\right)_{ig} \left(\tilde Z_{jb} \Phi_{jb} \delta_{jb}^{\downarrow}\right)_{i} + \frac{1}{m}\sum_{i=1}^m  \left(\tilde Z_{ja} \Phi_{ja}\right)_{ig} \left(\tilde Z_{j} v_{j}\right)_{i} \nonumber \\
={}& \frac{1}{m} \left(\Phi_{ja}^\T \tilde Z_{ja}^\T \sum_{b=1}^q\tilde Z_{jb} \Phi_{jb} \delta_{jb}^{\downarrow} \right)_g + \frac{1}{m} \left(\Phi_{ja}^\T \tilde Z_{ja}^\T \tilde Z_{j} v_{j}\right)_g \nonumber \\
={}& \frac{1}{m} \left(\Phi_{ja}^\T \tilde Z_{ja}^\T \tilde Z_{j} \Phi_{j} \delta_{j}^{\downarrow} \right)_g + \frac{1}{m} \left(\Phi_{ja}^\T \tilde Z_{ja}^\T \tilde Z_{j} v_{j}\right)_g,
\end{align}
which implies that in \eqref{eq:19}, for $g=1,\ldots,\bar h$ and $a=1,\ldots,q$,
\begin{align*}
& \frac{1}{m} \left(\Phi_{ja}^\T \tilde Z_{ja}^\T \tilde Z_{j} \Phi_{j} \delta_{j}^{\downarrow} \right)_g + \frac{1}{m} \left(\Phi_{ja}^\T \tilde Z_{ja}^\T \tilde Z_{j} v_{j}\right)_g + \frac{\tau_0^2\lambda_n}{km}  \frac{\delta_{jag}}{\mu_{ag}} =-\frac{\tau_0^2\lambda_n}{km}  \frac{\zeta_{0ag}}{\mu_{ag}} .
\end{align*}
Let $\zeta_{0a}^{\downarrow}=(\zeta_{0a1},\ldots,\zeta_{0a\bar h})^\T$ and $\zeta_0^{\downarrow} = (\zeta_{01}^{\downarrow \T},\ldots,\zeta_{0q}^{\downarrow\T})^\T \in \RR^{q\bar h}$.  Stacking the last display over $g=1,\ldots,\bar h$ gives
\begin{align} 
& \frac{1}{m} \Phi_{ja}^\T \tilde Z_{ja}^\T \tilde Z_{j} \Phi_{j} \delta_{j}^{\downarrow} + \frac{1}{m} \Phi_{ja}^\T \tilde Z_{ja}^\T \tilde Z_{j} v_{j} + \frac{\tau_0^2\lambda_n}{km} M_a^{-1} \delta_{ja}^{\downarrow} = -\frac{\tau_0^2\lambda_n}{km}  M_a^{-1}\zeta_{0a}^{\downarrow} . \nonumber
\end{align}
Then stacking this over $a=1,\ldots,q$ gives
\begin{align} 
& \frac{1}{m} \Phi_{j}^\T \tilde Z_{j}^\T \tilde Z_{j} \Phi_{j} \delta_{j}^{\downarrow} + \frac{1}{m} \Phi_{j}^\T \tilde Z_{j}^\T \tilde Z_{j} v_{j} + \frac{\tau_0^2\lambda_n}{km} M^{-1} \delta_{j}^{\downarrow} = -\frac{\tau_0^2\lambda_n}{km}  M^{-1}\zeta_{0}^{\downarrow}, \nonumber
\end{align}
which implies that
\begin{align} \label{eq:21}
&\left( \frac{1}{m}\Phi_{j}^\T \tilde Z_{j}^\T \tilde Z_{j} \Phi_{j} + \frac{\tau_0^2\lambda_n}{km} M^{-1}  \right)  \delta_j^{\downarrow} = -\frac{\tau_0^2\lambda_n}{km}  M^{-1}\zeta_{0}^{\downarrow}  - \frac{1}{m} \Phi_{j}^\T \tilde Z_{j}^\T \tilde Z_{j} v_{j}.
\end{align}
The proof is completed by showing that the right hand side expression in \eqref{eq:21} gives an upper bound for $\| \delta_j^{\downarrow} \|_2^2$. Using the $\Omega$ matrix defined in Assumption \ref{z_assumption}, we define the matrix $Q = \left(I_{q\bar h} + \frac{\tau_0^2\lambda_n}{km} \Omega^{-1} M^{-1} \right)^{1/2} \in \RR^{q\bar h\times q\bar h}$. Then
\begin{align*}
&\frac{1}{m}\Phi_{j}^\T \tilde Z_{j}^\T \tilde Z_{j} \Phi_{j}  + \frac{\tau_0^2\lambda_n}{km} M^{-1} = \Omega  +   \frac{\tau_0^2\lambda_n}{km} M^{-1} +   \frac{1}{m}\Phi_{j}^\T \tilde Z_{j}^\T \tilde Z_{j} \Phi_{j}  - \Omega \\
&= \Omega Q \left\{ I_{q\bar h} + Q^{-1}\Omega^{-1} \left( \frac{1}{m}\Phi_{j}^\T \tilde Z_{j}^\T \tilde Z_{j} \Phi_{j} - \Omega \right) Q^{-1}\right\} Q.
\end{align*}
and using this in \eqref{eq:21} gives
\begin{align} \label{eq:23}
&\left\{ I_{q\bar h} + Q^{-1}\Omega^{-1} \left(\frac{1}{m}\Phi_{j}^\T \tilde Z_{j}^\T \tilde Z_{j} \Phi_{j} - \Omega\right) Q^{-1}\right\} Q  \delta_j^{\downarrow} \nonumber\\
={}& - \frac{\tau_0^2\lambda_n}{km} Q^{-1}\Omega^{-1} M^{-1} \zeta_0^{\downarrow}  -  \frac{1}{m} Q^{-1}\Omega^{-1} \Phi_{j}^\T \tilde Z_{j}^\T \tilde Z_{j} v_{j}.
\end{align}
Now we define the event
\begin{align}\label{eset1}
& \Ecal_{j1} = \left\{ \vertiii{ Q^{-1}\Omega^{-1} \left( \frac{1}{m}\Phi_{j}^\T \tilde Z_{j}^\T \tilde Z_{j} \Phi_{j} - \Omega \right) Q^{-1} } \leq 1/2\right\},
\end{align}
with the randomness in $\{u_{j1},\ldots,u_{jm}\}$. We have that
\begin{align}\label{eq:IQ1}
& I_{q\bar h} + Q^{-1}\Omega^{-1} \left( \frac{1}{m}\Phi_{j}^\T \tilde Z_{j}^\T \tilde Z_{j} \Phi_{j} - \Omega \right) Q^{-1} \succeq \frac{1}{2} I_{q\bar h}
\end{align}
whenever $\Ecal_{j1}$ occurs. It is also clear that $Q\succeq I_{q\bar h}$. Therefore, when $\Ecal_{j1}$ occurs, \eqref{eq:23} implies that
\begin{align} \label{eq:28}
& \left\| \delta_j^{\downarrow} \right\|_2^2 \leq 4 \left\| \frac{\tau_0^2\lambda_n}{km} Q^{-1}\Omega^{-1} M^{-1} \zeta_0^{\downarrow} +  \frac{1}{m} Q^{-1}\Omega^{-1} \Phi_{j}^\T \tilde Z_{j}^\T \tilde Z_{j} v_{j} \right\|_2^2 \nonumber \\
& \leq 8 \left\| \frac{\tau_0^2\lambda_n}{km} Q^{-1}\Omega^{-1} M^{-1} \zeta_0^{\downarrow} \right\|_2^2 + 8 \left\| \frac{1}{m} Q^{-1}\Omega^{-1} \Phi_{j}^\T \tilde Z_{j}^\T \tilde Z_{j} v_{j} \right\|_2^2,
\end{align}
where the last inequality follows because $(a + b)^2 \leq 2 a^2 + 2 b^2$ for any $a, b \in \RR$.

For the first term on the right hand side of \eqref{eq:28}, we have that
\begin{align} \label{eq:28a}
& \quad \left\| \frac{\tau_0^2\lambda_n}{km} Q^{-1}\Omega^{-1} M^{-1} \zeta_0^{\downarrow} \right\|_2^2 \leq \left(\frac{\tau_0^2\lambda_n}{km}\right)^2 \zeta_0^{\downarrow \T} \left\{M\Omega \left(I_{q\bar h}+\frac{\tau_0^2\lambda_n}{km} \Omega^{-1} M^{-1} \right) \Omega M \right\}^{-1} \zeta_0^{\downarrow} \nonumber \\
&=\left(\frac{\tau_0^2\lambda_n}{km}\right)^2 \zeta_0^{\downarrow \T} \left(M\Omega^2 M + \frac{\tau_0^2\lambda_n}{km} \Omega M\right)^{-1} \zeta_0^{\downarrow} \nonumber \\
&\leq \left(\frac{\tau_0^2\lambda_n}{km}\right)^2 \zeta_0^{\downarrow \T} \left(\frac{\tau_0^2\lambda_n}{km} \Omega M\right)^{-1} \zeta_0^{\downarrow}  = \frac{\tau_0^2\lambda_n}{\underline c_{\Omega}km} \sum_{a=1}^q \sum_{h=1}^{\bar h} \frac{\zeta_{0h}^2}{\mu_{ah}} \leq \frac{\tau_0^2\lambda_n}{\underline c_{\Omega}km} \|\nu_0\|_{\HH}^2.
\end{align}
For the second term on right hand side of \eqref{eq:28}, it is equal to
\begin{align} \label{eq:28b1}
& \frac{1}{m} Q^{-1}\Omega^{-1} \Phi_{j}^\T \tilde Z_{j}^\T \tilde Z_{j} v_{j}
= \left(M + \frac{\tau_0^2 \lambda_n}{km}\Omega^{-1} \right)^{-1/2} \cdot \frac{1}{m} M^{1/2} \Omega^{-1} \Phi_{j}^\T \tilde Z_{j}^\T \tilde Z_{j} v_j  .
\end{align}
The first term in \eqref{eq:28b1} has bounded matrix operator norm by Assumption \ref{z_assumption}:
\begin{align} \label{eq:28b2}
\vertiii{ \left(M + \frac{\tau_0^2 \lambda_n}{km}\Omega^{-1} \right)^{-1/2} } & = \max_{1\leq a\leq q,1\leq h\leq \bar h} \frac{1}{\sqrt{\mu_{ah} + \tfrac{\tau_0^2 \lambda_n}{\overline c_{\Omega} km}}} \leq \sqrt{\frac{\overline c_{\Omega} km}{\tau_0^2 \lambda_n}}.
\end{align}
For the second term in \eqref{eq:28b1}, we repeatedly apply $\left|\varphi_{ai}\right|\leq C_{\varphi}$ and $\left|\tilde Z_{a}(u_{ji})\right|\leq \tilde C_Z$ in Assumptions \ref{eigen_assumption} and \ref{z_assumption} to obtain that
\begin{align} \label{eq:28b3}
&\quad  \left\|\frac{1}{m} M^{1/2} \Omega^{-1} \Phi_{j}^\T \tilde Z_{j}^\T \tilde Z_{j} v_j \right\|_2^2 \leq \frac{\underline c_{\Omega}^{-2}}{m^2} \left(\tilde Z_{j} v_j\right)^\T \tilde Z_{j} \Phi_{j} M \Phi_{j}^\T \tilde Z_{j}^\T \left(\tilde Z_{j} v_j\right) \nonumber \\
&= \frac{\underline c_{\Omega}^{-2}}{m^2} \sum_{a=1}^q \left(\tilde Z_{j} v_j\right)^\T \left(\tilde Z_{ja} \Phi_{ja} M_a \Phi_{ja}^\T \tilde Z_{ja} \right) \left(\tilde Z_{j} v_j\right) \nonumber \\
& \leq  \frac{\underline c_{\Omega}^{-2}}{m^2}\sum_{a=1}^q \sum_{g=1}^{\bar h} \mu_{ag} \left\{\sum_{i=1}^m \tilde Z_{a}(u_{ji}) \varphi_{ag}(u_{ji}) \left(\tilde Z_{j} v_j\right)_{i} \right\}^2  \nonumber \\
&\leq \frac{C_{\varphi}^2 \tilde C_Z^2 \underline c_{\Omega}^{-2}}{m^2} \sum_{a=1}^q \sum_{g=1}^{\bar h} \mu_{ag} \left\{\sum_{i=1}^m \left|\left(\tilde Z_{j} v_j\right)_{i}\right| \right\}^2 \nonumber \\
&= \frac{C_{\varphi}^2 \tilde C_Z^2 \underline c_{\Omega}^{-2}}{m^2} \sum_{a=1}^q \sum_{g=1}^{\bar h} \mu_{ag} \left\{\sum_{i=1}^m \sum_{b=1}^q \left| \tilde Z_{b}(u_{ji}) v_{jbi}\right| \right\}^2 \nonumber \\
&\leq \frac{C_{\varphi}^2 \tilde C_Z^4 \underline c_{\Omega}^{-2}}{m^2} \sum_{a=1}^q \sum_{g=1}^{\bar h} \mu_{ag} \left\{\sum_{i=1}^m \sum_{b=1}^q \sum_{h=\bar h+1}^{\infty}\left|\delta_{jbh} \varphi_{bh}(u_{ji})\right| \right\}^2 \nonumber \\
&\leq C_{\varphi}^4 \tilde C_Z^4 \underline c_{\Omega}^{-2} \left(\sum_{a=1}^q \sum_{g=1}^{\bar h} \mu_{ag} \right) \left(\sum_{b=1}^q \sum_{h=\bar h+1}^{\infty} |\delta_{jbh}| \right )^2 \nonumber \\
&\stackrel{(i)}{\leq} C_{\varphi}^4 \tilde C_Z^4 \underline c_{\Omega}^{-2} \left(\sum_{a=1}^q \tr(\rho_a)\right) \cdot \left(\sum_{b=1}^q\sum_{h=\bar h+1}^{\infty}  \mu_{bh} \right) \left(\sum_{b=1}^q \sum_{h=\bar h+1}^{\infty} \frac{\delta_{jbh}^2}{\mu_{bh}} \right ) \nonumber \\
&= C_{\varphi}^4 \tilde C_Z^4 \underline c_{\Omega}^{-2} \left(\sum_{a=1}^q \tr(\rho_a)\right) \left(\sum_{a=1}^q \tr(\rho_{a,\bar h})\right) \left(\sum_{b=1}^q \left\|\EE_{\nu_{jb},y_j\mid u_j} (\Delta_{jb}) \right\|_{\HH_b}^2 \right)  \nonumber \\
&= C_{\varphi}^4 \tilde C_Z^4 \underline c_{\Omega}^{-2} \Tr(\rho) \Tr(\rho,\bar h) \left\|\EE_{\nu_{j},y_j\mid u_j} (\Delta_{j}) \right\|_{\HH}^2 \nonumber \\
&\stackrel{(ii)}{\leq} C_{\varphi}^4 \tilde C_Z^4 \underline c_{\Omega}^{-2} \|\nu_0\|_{\HH}^2 \Tr(\rho) \Tr(\rho,\bar h) ,
\end{align}
where $(i)$ is from the Cauchy-Schwarz inequality and the definition of $\tr(\rho_a)$, and $(ii)$ is from the relation \eqref{eq:17}. Combining \eqref{eq:28b1}, \eqref{eq:28b2}, and \eqref{eq:28b3} leads to
\begin{align} \label{eq:28b4}
\left\|\frac{1}{m} Q^{-1}\Omega^{-1} \Phi_{j}^\T \tilde Z_{j}^\T \tilde Z_{j} v_{j}  \right\|_2^2 &\leq \frac{\overline c_{\Omega} km}{\tau_0^2 \lambda_n} C_{\varphi}^4 \tilde C_Z^4 \underline c_{\Omega}^{-2} \|\nu_0\|_{\HH}^2 \Tr(\rho) \Tr(\rho,\bar h) .
\end{align}

Finally we combine these results. Note that \eqref{eq:ba2} implies that
\begin{align} \label{eq:ba2.1}
\|\delta_j^{\downarrow}\|_2^2 & = \sum_{a=1}^q \|\delta_{ja}^{\downarrow}\|_2^2 \leq \sum_{a=1}^q \tr(\rho_a) \left\| \nu_{0a} \right\|_{\HH_a}^2 \leq \Tr(\rho) \left\| \nu_{0} \right\|_{\HH}^2.
\end{align}
\eqref{eq:ba3} and \eqref{eq:17} together imply that
\begin{align} \label{eq:ba3.1}
\|\delta_j^{\uparrow}\|_2^2 & = \sum_{a=1}^q \|\delta_{ja}^{\uparrow}\|_2^2 \leq \sum_{a=1}^q \mu_{a(\bar h+1)} \left\| {\EE}_{\nu_{ja},y_j|u_j}(\Delta_{ja}) \right\|_{\HH_a}^2 \leq \mu_{(\bar h+1)*} \left\| \nu_{0} \right\|_{\HH}^2.
\end{align}
Based on \eqref{subset_bias1} and the definition of $\Delta_j$, we combine \eqref{eq:ba2.1}, \eqref{eq:ba3}, \eqref{eq:28}, \eqref{eq:28a}, \eqref{eq:28b4}, and Lemma \ref{lem:zhanglem10} to obtain that
\begin{align}
& \quad \EE_{u^*} \EE_{u_j} \left\| \EE_{\nu_j, y_j\mid u_j}\{\nu_j(u^*)\} - \nu_0(u^*) \right\|_2 ^2 = \EE_{u^*} \EE_{u_j} \left\|\EE_{\nu_j, y_j\mid u_j}(\Delta_j) \right\|_2^2 \nonumber \\
&\leq \EE_{u_j} \left(\| \delta_j^{\downarrow} \|_2^2 + \| \delta_j^{\uparrow} \|_2^2 \right) \nonumber \\
&= \EE_{u_j} \left\{\| \delta_j^{\downarrow} \|_2^2 1(\Ecal_{j1}) +  \| \delta_j^{\downarrow} \|_2^2 1(\Ecal_{j1}^c) + \| \delta_j^{\uparrow} \|_2^2 \right\} \nonumber \\
&\leq \EE_{u_j} \left\{\| \delta_j^{\downarrow} \|_2^2 1(\Ecal_{j1})\right\}  + \Tr(\rho) \|\nu_0\|_{\HH}^2 \PP_{u_j}(\Ecal_{j1}^c) + \EE_{u_j} \left( \| \delta_j^{\uparrow} \|_2^2 \right) \nonumber \\
&\leq 8\frac{\tau_0^2\lambda_n}{\underline c_{\Omega}km} \|\nu_0\|_{\HH}^2 + 8 \frac{ km}{\tau_0^2 \lambda_n} C_{\varphi}^4 \tilde C_Z^4 \overline c_{\Omega} \underline c_{\Omega}^{-2} \|\nu_0\|_{\HH}^2 \Tr(\rho) \Tr(\rho,\bar h)  \nonumber \\
&\quad + 2q\bar h \|\nu_0\|_{\HH}^2 \Tr(\rho) \exp\left\{-\frac{m}{8(B^2+B) } \right\} +  \mu_{(\bar h+1)*}\|\nu_0\|_{\HH}^2 , \nonumber
\end{align}
where $B=C_{\varphi}^2 \tilde C_Z^2 \underline c_{\Omega}^{-1} q \Lambda\left(\tau_0^2\lambda_n/(\underline c_{\Omega}km),\bar h\right)+ 1$. This proves the conclusion.
\end{proof}

\vspace{5mm}

\begin{lemma} \label{lem:var1_bound}
Suppose that Assumptions \ref{sampling}--\ref{combine_assumption} hold. Then for every $j=1,\ldots,k$,
\begin{align*}
& \EE_{u^*}\EE_{u_j} \tr\left( {\var}_{y_j\mid u_j} \left[\EE_{\nu_j\mid y_j,u_j}\left\{\nu_j(u^*)\right\} \right] \right)   \\
\leq {}& 12\frac{\tau_0^2\lambda_n}{\underline c_{\Omega}km} \|\nu_0\|_{\HH}^2 + \frac{24km}{\tau_0^2 \lambda_n} C_{\varphi}^4 \tilde C_Z^4 \overline c_{\Omega} \underline c_{\Omega}^{-2} \Tr(\rho) \Tr(\rho,\bar h) \left(\frac{km}{\lambda_n} + 2\|\nu_0\|_{\HH}^2\right) \nonumber \\
&~~ + 12 \frac{C_{\varphi}^2 \tilde C_Z^2 \underline c_{\Omega}^{-2} \tau_0^2 q}{m} \Lambda(\tau_0^2\lambda_n/(\overline c_{\Omega} km),\bar h) + 2\mu_{(\bar h+1)*} \left( \frac{km}{\lambda_n} + 2 \|\nu_0\|_{\HH}^2 \right) \nonumber\\
&~~ + 4q\bar h \Tr(\rho) \left( \frac{km}{\lambda_n} + 2 \|\nu_0\|_{\HH}^2 \right) \exp\left\{-\frac{m}{8(B^2+B)} \right\},
\end{align*}
where $B=C_{\varphi}^2 \tilde C_Z^2 \underline c_{\Omega}^{-1} q \Lambda\left(\tau_0^2\lambda_n/(\underline c_{\Omega}km),\bar h\right)+ 1$.
\end{lemma}

\begin{proof}[Proof of Lemma \ref{lem:var1_bound}]
We use the same notations as in the proof of Lemma \ref{lem:bias_bound}. We further expand the functions $\Delta_{ja}(\cdot)$ defined in \eqref{eq:Deltaj} for $a=1,\ldots,q$ and $j=1,\ldots,k$ as
\begin{align*}
& \Delta_{ja}(\cdot) = \sum_{i=1}^{\infty} \tilde \delta_{jai} \varphi_{ai}(\cdot), \quad  \Delta^{\downarrow}_{ja}(\cdot) = \sum_{i=1}^{\bar h} \tilde\delta_{jai} \varphi_{ai}(\cdot), \quad
  \Delta^{\uparrow}_{ja} (\cdot) = \sum_{i=\bar h+1}^{\infty} \tilde \delta_{jai} \varphi_{ai}(\cdot), \\
& \tilde \delta_{ja}^{\downarrow} = \left(\tilde \delta_{ja1}, \ldots, \tilde \delta_{ja\bar h}\right)^\T, \quad \tilde\delta_{ja}^{\uparrow} = \left(\tilde\delta_{ja(\bar h+1)}, \ldots, \tilde\delta_{j a\infty}\right)^\T, \\
& \tilde \delta_j^{\downarrow} = \left(\tilde \delta_{j1}^{\downarrow \T}, \ldots, \tilde \delta_{jq}^{\downarrow \T}\right)^\T, \quad \tilde \delta_j^{\uparrow} = \left(\tilde \delta_{j1}^{\uparrow \T}, \ldots, \tilde \delta_{jq}^{\uparrow \T}\right)^\T.
\end{align*}
From \eqref{subset_var1}, we can see that
\begin{align} \label{eq:var1deltaj}
& \quad \EE_{u^*} \tr\left({\var}_{y_j\mid u_j} \left[\EE_{\nu_j\mid y_j,u_j}\{\nu_j(u^*)\} \right]\right)\nonumber \\
& = \EE_{u^*} \tr\left[ {\var}_{y_j\mid u_j} \left\{R_{j}(u^*) \tilde Z_j^\T \left(\tilde Z_j \tilde R_{jj} \tilde Z_j^\T + \frac{\tau_0^2\lambda_n}{k} I_{m} \right)^{-1} y_j \right\} \right] \nonumber \\
&\leq \EE_{u^*}\EE_{y_j\mid u_j} \left\| R_{j}(u^*) \tilde Z_j^\T \left(\tilde Z_j \tilde R_{jj} \tilde Z_j^\T + \frac{\tau_0^2\lambda_n}{k} I_{m} \right)^{-1} y_j - \nu_0(u^*) \right\|_2^2 \nonumber \\
&= \EE_{u^*} \EE_{y_j\mid u_j} \|\Delta_j(u^*)\|_2^2 =  \EE_{u^*}\EE_{y_j\mid u_j} \|\Delta_j^{\downarrow}(u^*)\|_2^2 +\EE_{u^*}\EE_{y_j\mid u_j} \|\Delta_j^{\uparrow}(u^*)\|_2^2 \nonumber \\
&= \EE_{y_j\mid u_j} \|\tilde\delta_j^{\downarrow}\|_2^2 +\EE_{y_j\mid u_j} \|\tilde\delta_j^{\uparrow}\|_2^2.
\end{align}
Therefore, we will find an upper bound for $\EE_{u^*} \EE_{y_j\mid u_j} \|\Delta_j(u^*)\|_2^2$ in the following. We start with finding a rough upper bound for $\EE_{y_j\mid u_j} \|\Delta_j\|_{\HH}^2$. Using the definition of $\hat \nu_j$ in \eqref{eq:Deltaj} and the optimizer property in \eqref{eq:wjopt}, we have that
\begin{align} \label{eq:vv12}
\| \hat \nu_j \|_{\HH}^2 &\overset{(i)}{\leq}  \sum_{i=1}^m \frac{\left\{y(u_{ji}) - \tilde Z(u_{ji}) \hat \nu_j(u_{ji})  \right\}^2}{\tau_0^2\lambda_n /k} + \| \hat \nu_j \|_{\HH}^2\nonumber \\
& \overset{(ii)}{\leq}
\sum_{i=1}^m \frac{\left\{y(u_{ji})- \tilde Z(u_{ji}) \nu_{0}(u_{ji}) \right\}^2} {\tau_0^2\lambda_n /k} +  \| \nu_0 \|_{\HH}^2 \nonumber \\
& \overset{(iii)}{\leq}
\sum_{i=1}^m \frac{\left\{ \epsilon(u_{ji}) \right\}^2} {\tau_0^2\lambda_n /k} +  \| \nu_0 \|_{\HH}^2,
\end{align}
where $(i)$ follows because the term inside the summation is non-negative, $(ii)$ follows because $\hat \nu_j$ minimizes the objective, and $(iii)$ follows from our model assumption. Since the error variance is $\tau_0^2$, \eqref{eq:vv12} implies that
\begin{align} \label{eq:vv3}
\EE_{y_j\mid u_j} \|\Delta_j\|_{\HH}^2 &\leq 2\EE_{y_j\mid u_j} \|\hat \nu_j \|_{\HH}^2 + 2\EE_{y_j\mid u_j} \|\nu_0\|_{\HH}^2 \nonumber \\
&\leq 2\EE_{y_j\mid u_j} \left[\sum_{i=1}^m \frac{\left\{ \epsilon(u_{ji}) \right\}^2} {\tau_0^2\lambda_n /k} \right] + 4  \|\nu_0\|_{\HH}^2 \nonumber \\
&\leq \frac{2km}{\lambda_n} + 4 \|\nu_0\|_{\HH}^2.
\end{align}
Using this bound for $\EE_{y_j\mid u_j} \|\Delta_j\|_{\HH}^2$, we can find an upper bound for $\EE_{y_j\mid u_j} \|\delta_j^{\uparrow}\|_2^2$:
\begin{align} \label{eq:v31}
& \quad \EE_{y_j\mid u_j} \| \tilde \delta_j^{\uparrow} \|_2^2 =  \sum_{a=1}^q \sum_{i=\bar h+1}^{\infty} \EE_{y_j\mid u_j} (\tilde \delta_{jai}^2) \nonumber \\
&= \sum_{a=1}^q \mu_{a(\bar h+1)}\sum_{i=\bar h+1}^{\infty} \frac{\EE_{y_j\mid u_j} (\tilde \delta_{jai}^2)}{\mu_{a(\bar h+1)}} \leq \sum_{a=1}^q \mu_{a(\bar h+1)}\sum_{i=\bar h+1}^{\infty} \frac{\EE_{y_j\mid u_j} (\tilde \delta_{jai}^2)}{\mu_{ai}}  \nonumber \\
&= \sum_{a=1}^q \mu_{a(\bar h+1)} \EE_{y_j\mid u_j} \| \Delta_{ja}^{\uparrow} \|^2_{\HH_a} \leq \sum_{a=1}^q \mu_{a(\bar h+1)} \EE_{y_j\mid u_j} \| \Delta_{ja} \|^2_{\HH_a} \nonumber \\
&\leq \mu_{(\bar h+1)*} \left( \frac{2km}{\lambda_n} + 4  \|\nu_0\|_{\HH}^2 \right),
\end{align}
and also an upper bound for $\EE_{y_j\mid u_j} \|\tilde\delta_j\|_2^2$:
\begin{align} \label{eq:v32}
& \quad \EE_{y_j\mid u_j} \|\tilde\delta_j\|_2^2  = \sum_{a=1}^q \sum_{i=1}^{\infty} \EE_{y_j\mid u_j} (\tilde \delta_{jai}^2) \nonumber \\
&\leq \sum_{a=1}^q\mu_{a1}\sum_{i=1}^{\infty} \frac{\EE_{y_j\mid u_j} (\tilde \delta_{ji}^2)}{\mu_{ai}} = \sum_{a=1}^q\mu_{a1} \EE_{y_j\mid u_j} \| \Delta_j^{\uparrow} \|^2_{\HH_a} \leq \mu_{1*}  \EE_{y_j\mid u_j} \| \Delta_j \|^2_{\HH} \nonumber \\
&\leq \Tr(\rho) \left( \frac{2km}{\lambda_n} + 4  \|\nu_0\|_{\HH}^2 \right),
\end{align}

Now we find an upper bound for $\EE_{y_j\mid u_j} \| \tilde \delta_j^{\downarrow} \|_2^2$. Define the error vectors
\begin{align*}
& \tilde v_{jai} = \sum_{h=\bar h+1}^\infty \tilde \delta_{jai} \varphi_{ah}(u_{ji}), \\
& \tilde v_{ja} = (\tilde v_{ja1},\ldots,\tilde v_{jam})^\T \in \RR^m, ~~ i=1,\ldots,m,~~ a=1,\ldots,q,\\
& \tilde v_j = \left(\tilde v_{j1}^\T, \ldots, \tilde v_{jq}^\T \right)^\T \in \RR^{qm}.
\end{align*}
Now we use an argument similar to the derivation of \eqref{eq:ba22}, \eqref{eq:ba6}, \eqref{eq:16}, \eqref{eq:18}, \eqref{eq:19}, \eqref{eq:20}, and \eqref{eq:21}. Instead of taking the ${\EE}_{\nu_j,y_j|u_j}$ as in \eqref{eq:16}, we do not take this expectation and keep the error term $\epsilon(u_{ji})$ all the way along the derivation. We can obtain the following relation similar to \eqref{eq:21}:
\begin{align} \label{eq:v5}
&\left( \frac{1}{m}\Phi_{j}^\T \tilde Z_{j}^\T \tilde Z_{j} \Phi_{j} + \frac{\tau_0^2\lambda_n}{km} M^{-1}  \right)  \tilde \delta_j^{\downarrow} = -\frac{\tau_0^2\lambda_n}{km}  M^{-1}\zeta_{0}^{\downarrow}  - \frac{1}{m} \Phi_{j}^\T \tilde Z_{j}^\T \tilde Z_{j} \tilde v_{j} + \frac{1}{m} \Phi_{j}^\T \tilde Z_{j}^\T \epsilon_j.
\end{align}
We use the same $Q$ matrix as defined in the proof of Lemma \ref{lem:bias_bound}. Then \eqref{eq:v5} can be rewritten as
\begin{align} \label{eq:v15}
&\left\{ I_{q\bar h} + Q^{-1}\Omega^{-1} \left(\frac{1}{m}\Phi_{j}^\T \tilde Z_{j}^\T \tilde Z_{j} \Phi_{j} - \Omega\right) Q^{-1}\right\} Q  \tilde \delta_j^{\downarrow} \nonumber\\
={}& - \frac{\tau_0^2\lambda_n}{km} Q^{-1}\Omega^{-1} M^{-1} \zeta_0^{\downarrow}  -  \frac{1}{m} Q^{-1}\Omega^{-1} \Phi_{j}^\T \tilde Z_{j}^\T \tilde Z_{j} \tilde v_{j} + \frac{1}{m} Q^{-1}\Omega^{-1} \Phi_{j}^\T \tilde Z_{j}^\T \epsilon_j.
\end{align}
On the event $\Ecal_{j1}$ defined as in \eqref{eset1}, using \eqref{eq:IQ1} and the fact $Q\succeq I_{q\bar h}$ \eqref{eq:v15} imply that
\begin{align} \label{eq:v1:3terms}
& \EE_{y_j\mid u_j} \| \tilde \delta_j^{\downarrow} \|_2^2 \leq  \EE_{y_j\mid u_j} \| Q  \tilde \delta_j^{\downarrow} \|_2^2 \nonumber \\
&\leq 4 \EE_{y_j\mid u_j}  \left\| - \frac{\tau_0^2\lambda_n}{km} Q^{-1}\Omega^{-1} M^{-1} \zeta_0^{\downarrow}  -  \frac{1}{m} Q^{-1}\Omega^{-1} \Phi_{j}^\T \tilde Z_{j}^\T \tilde Z_{j} \tilde v_{j} + \frac{1}{m} Q^{-1}\Omega^{-1} \Phi_{j}^\T \tilde Z_{j}^\T \epsilon_j  \right\|_2^2 \nonumber \\
&\leq 12  \left\| \frac{\tau_0^2\lambda_n}{km} Q^{-1}\Omega^{-1} M^{-1} \zeta_0^{\downarrow}\right\|_2^2 + 12\EE_{y_j\mid u_j}  \left\| \frac{1}{m} Q^{-1}\Omega^{-1} \Phi_{j}^\T \tilde Z_{j}^\T \tilde Z_{j} \tilde v_{j}  \right\|_2^2  \nonumber \\
&\quad +  12\EE_{y_j\mid u_j}  \left\| \frac{1}{m} Q^{-1}\Omega^{-1} \Phi_{j}^\T \tilde Z_{j}^\T \epsilon_j  \right\|_2^2,
\end{align}
where the last inequality follows because $(a + b + c)^2 \leq 3 a^2 + 3 b^2 + 3 c^2$ for any $a, b, c \in \RR$. We bound the three terms on the right hand side of \eqref{eq:v1:3terms}. The first term can be bounded as in \eqref{eq:28a}. The second term can be bounded similar to the proof of Lemma \ref{lem:bias_bound}: By Assumption \ref{eigen_assumption}, we have that
\begin{align} \label{eq:v1:28b1}
& \frac{1}{m} Q^{-1}\Omega^{-1} \Phi_{j}^\T \tilde Z_{j}^\T \tilde Z_{j} \tilde v_{j}
= \left(M + \frac{\tau_0^2 \lambda_n}{km}\Omega^{-1} \right)^{-1/2} \cdot \frac{1}{m} M^{1/2} \Omega^{-1} \Phi_{j}^\T \tilde Z_{j}^\T \tilde Z_{j} \tilde v_j.
\end{align}
The first term in \eqref{eq:v1:28b1} has bounded matrix operator norm in \eqref{eq:28b2}. For the second term in \eqref{eq:28b1}, we repeatedly apply Assumptions \ref{eigen_assumption} and \ref{z_assumption} to obtain that
\begin{align} \label{eq:v1:28b3}
&\quad \EE_{y_j\mid u_j}  \left\|\frac{1}{m} M^{1/2} \Omega^{-1} \Phi_{j}^\T \tilde Z_{j}^\T \tilde Z_{j} \tilde v_j\right\|_2^2 \nonumber \\
&= \frac{\underline c_{\Omega}^{-2}}{m^2} \EE_{y_j\mid u_j}  \sum_{a=1}^q \left(\tilde Z_{j} v_j\right)^\T \left(\tilde Z_{ja} \Phi_{ja} M_a \Phi_{ja}^\T \tilde Z_{ja} \right) \left(\tilde Z_{j} \tilde v_j\right) \nonumber \\
& \leq  \frac{\underline c_{\Omega}^{-2}}{m^2} \EE_{y_j\mid u_j} \sum_{a=1}^q \sum_{g=1}^{\bar h} \mu_{ag} \left\{\sum_{i=1}^m \tilde Z_{a}(u_{ji}) \varphi_{ag}(u_{ji}) \left(\tilde Z_{j} \tilde v_j\right)_{i} \right\}^2  \nonumber \\
&\leq \frac{C_{\varphi}^2 \tilde C_Z^2 \underline c_{\Omega}^{-2}}{m^2} \EE_{y_j\mid u_j}\sum_{a=1}^q \sum_{g=1}^{\bar h} \mu_{ag} \left\{\sum_{i=1}^m \left|\left(\tilde Z_{j} \tilde v_j\right)_{i}\right| \right\}^2 \nonumber \\
&= \frac{C_{\varphi}^2 \tilde C_Z^2 \underline c_{\Omega}^{-2}}{m^2} \EE_{y_j\mid u_j}\sum_{a=1}^q \sum_{g=1}^{\bar h} \mu_{ag} \left\{\sum_{i=1}^m \sum_{b=1}^q \left| \tilde Z_{b}(u_{ji}) \tilde v_{jbi}\right| \right\}^2 \nonumber \\
&\leq \frac{C_{\varphi}^2 \tilde C_Z^4 \underline c_{\Omega}^{-2}}{m^2} \EE_{y_j\mid u_j} \sum_{a=1}^q \sum_{g=1}^{\bar h} \mu_{ag} \left\{\sum_{i=1}^m \sum_{b=1}^q \sum_{h=\bar h+1}^{\infty}\left|\tilde \delta_{jbh} \varphi_{bh}(u_{ji})\right| \right\}^2 \nonumber \\
&\leq C_{\varphi}^4 \tilde C_Z^4 \underline c_{\Omega}^{-2} \left(\sum_{a=1}^q \sum_{g=1}^{\bar h} \mu_{ag} \right) \EE_{y_j\mid u_j} \left(\sum_{b=1}^q \sum_{h=\bar h+1}^{\infty} \left|\tilde \delta_{jbh}\right| \right )^2 \nonumber \\
&\stackrel{(i)}{\leq} C_{\varphi}^4 \tilde C_Z^4 \underline c_{\Omega}^{-2} \left(\sum_{a=1}^q \tr(\rho_a)\right) \cdot \left(\sum_{b=1}^q\sum_{h=\bar h+1}^{\infty}  \mu_{bh} \right) \EE_{y_j\mid u_j}\left(\sum_{b=1}^q \sum_{h=\bar h+1}^{\infty} \frac{\tilde \delta_{jbh}^2}{\mu_{bh}} \right ) \nonumber \\
&= C_{\varphi}^4 \tilde C_Z^4 \underline c_{\Omega}^{-2} \Tr(\rho) \left(\sum_{a=1}^q \tr(\rho_{a,\bar h})\right) \EE_{y_j\mid u_j} \left(\sum_{b=1}^q \left\|\Delta_{jb} \right\|_{\HH_b}^2 \right)  \nonumber \\
&= C_{\varphi}^4 \tilde C_Z^4 \underline c_{\Omega}^{-2} \Tr(\rho) \Tr(\rho,\bar h) \EE_{y_j\mid u_j} \left\|\Delta_{j} \right\|_{\HH}^2 \nonumber \\
&\stackrel{(ii)}{\leq} C_{\varphi}^4 \tilde C_Z^4 \underline c_{\Omega}^{-2} \Tr(\rho) \Tr(\rho,\bar h) \left(\frac{2km}{\lambda_n} + 4 \|\nu_0\|_{\HH}^2\right) ,
\end{align}
where $(i)$ is from the Cauchy-Schwarz inequality and the definition of $\tr(C)$, and $(ii)$ is from the relation \eqref{eq:vv3}. Combining \eqref{eq:v1:28b1}, \eqref{eq:28b2}, and \eqref{eq:v1:28b3} leads to
\begin{align} \label{eq:v1:28b4}
\EE_{y_j\mid u_j}  \left\| \frac{1}{m} Q^{-1}\Omega^{-1} \Phi_{j}^\T \tilde Z_{j}^\T \tilde Z_{j} \tilde v_{j}  \right\|_2^2 &\leq \frac{2km}{\tau_0^2 \lambda_n} C_{\varphi}^4 \tilde C_Z^4 \overline c_{\Omega} \underline c_{\Omega}^{-2} \Tr(\rho) \Tr(\rho,\bar h) \left(\frac{km}{\lambda_n} + 2\|\nu_0\|_{\HH}^2\right).
\end{align}
For the third term in \eqref{eq:v1:3terms}, by Assumptions \ref{eigen_assumption} and \ref{z_assumption}, we have that
\begin{align} \label{eq:v1:28c}
& \quad \EE_{y_j\mid u_j}  \left\| \frac{1}{m} Q^{-1}\Omega^{-1} \Phi_{j}^\T \tilde Z_{j}^\T \epsilon_j \right\|_2^2 \nonumber \\
&\leq \frac{\underline c_{\Omega}^{-2}}{m^2} \EE_{y_j\mid u_j} \left\{ \epsilon_j^\T \tilde Z_{j} \Phi_{j} \left(I_{q\bar h} + \frac{\tau_0^2 \lambda_n}{km} \Omega^{-1}M^{-1}\right)^{-1} \Phi_{j}^\T \tilde Z_{j}^\T \epsilon_j \right\} \nonumber \\
&\leq \frac{\underline c_{\Omega}^{-2}}{m^2} \EE_{y_j\mid u_j} \sum_{a=1}^q \sum_{h=1}^{\bar h} \frac{1}{1+\tfrac{\tau_0^2\lambda_n}{\overline c_{\Omega} km\mu_{ah}}} \left\{ \sum_{i=1}^m \tilde Z_a(u_{ji})\varphi_{ah}(u_{ji}) \epsilon(u_{ji})\right\}^2 \nonumber \\
&\stackrel{(i)}{=} \frac{\underline c_{\Omega}^{-2}}{m^2} \sum_{a=1}^q \sum_{h=1}^{\bar h} \frac{1}{1+\tfrac{\tau_0^2\lambda_n}{\overline c_{\Omega} km\mu_{ah}}} \EE_{y_j\mid u_j} \left\{ \sum_{i=1}^m \tilde Z_a(u_{ji})^2 \varphi_{ah}(u_{ji})^2 \epsilon(u_{ji})^2\right\} \nonumber\\
&\leq \frac{C_{\varphi}^2 \tilde C_Z^2 \underline c_{\Omega}^{-2}}{m^2} \sum_{a=1}^q \sum_{h=1}^{\bar h} \frac{1}{1+\tfrac{\tau_0^2\lambda_n}{\overline c_{\Omega} km\mu_{ah}}} \EE_{y_j\mid u_j} \left\{ \sum_{i=1}^m \epsilon(u_{ji})^2\right\} \nonumber\\
&\leq \frac{C_{\varphi}^2 \tilde C_Z^2 \underline c_{\Omega}^{-2} \tau_0^2 q}{m} \Lambda(\tau_0^2\lambda_n/(\overline c_{\Omega} km),\bar h).
\end{align}
where $(i)$ follows from the independence between $\{\epsilon_{j1},\ldots,\epsilon_{jm}\}$.
Therefore, we can obtain that
\begin{align}
& \quad \EE_{u^*}\EE_{u_j} \tr\left( {\var}_{y_j\mid u_j} \left[\EE_{\nu_j\mid y_j,u_j}\left\{\nu_j(u^*)\right\} \right] \right) \nonumber \\
&\stackrel{(i)}{\leq} \EE_{y_j\mid u_j} \|\tilde\delta_j^{\downarrow}\|_2^2 +\EE_{y_j\mid u_j} \|\tilde\delta_j^{\uparrow}\|_2^2 \nonumber \\
&= \EE_{y_j\mid u_j} \left\{\|\tilde\delta_j^{\downarrow}\|_2^2 1(\Ecal_{j1})\right\} + \EE_{y_j\mid u_j} \left\{\|\tilde\delta_j^{\downarrow}\|_2^2 1(\Ecal_{j1}^c)\right\}  + \EE_{y_j\mid u_j} \|\tilde\delta_j^{\uparrow}\|_2^2 \nonumber \\
&\stackrel{(ii)}{\leq} 12  \left\| \frac{\tau_0^2\lambda_n}{km} Q^{-1}\Omega^{-1} M^{-1} \zeta_0^{\downarrow}\right\|_2^2 + 12\EE_{y_j\mid u_j}  \left\| \frac{1}{m} Q^{-1}\Omega^{-1} \Phi_{j}^\T \tilde Z_{j}^\T \tilde Z_{j} \tilde v_{j}  \right\|_2^2  \nonumber \\
&\quad +  12\EE_{y_j\mid u_j}  \left\| \frac{1}{m} Q^{-1}\Omega^{-1} \Phi_{j}^\T \tilde Z_{j}^\T \epsilon_j  \right\|_2^2 +
\mu_{(\bar h+1)*} \left( \frac{2km}{\lambda_n} + 4  \|\nu_0\|_{\HH}^2 \right) \nonumber\\
&\quad + \Tr(\rho) \left( \frac{2km}{\lambda_n} + 4 \|\nu_0\|_{\HH}^2 \right) \PP_{u_j}(\Ecal_{j1}^c) \nonumber \\
&\stackrel{(iii)}{\leq} 12\frac{\tau_0^2\lambda_n}{\underline c_{\Omega}km} \|\nu_0\|_{\HH}^2 + \frac{24km}{\tau_0^2 \lambda_n} C_{\varphi}^4 \tilde C_Z^4 \overline c_{\Omega} \underline c_{\Omega}^{-2} \Tr(\rho) \Tr(\rho,\bar h) \left(\frac{km}{\lambda_n} + 2\|\nu_0\|_{\HH}^2\right) \nonumber \\
&\quad + 12 \frac{C_{\varphi}^2 \tilde C_Z^2 \underline c_{\Omega}^{-2} \tau_0^2 q}{m} \Lambda(\tau_0^2\lambda_n/(\overline c_{\Omega} km),\bar h) + 2\mu_{(\bar h+1)*} \left( \frac{km}{\lambda_n} + 2 \|\nu_0\|_{\HH}^2 \right) \nonumber\\
&\quad + 4q\bar h \Tr(\rho) \left( \frac{km}{\lambda_n} + 2 \|\nu_0\|_{\HH}^2 \right) \exp\left\{-\frac{m}{8(B^2+B)} \right\}, \nonumber
\end{align}
where $B=C_{\varphi}^2 \tilde C_Z^2 \underline c_{\Omega}^{-1} q \Lambda\left(\tau_0^2\lambda_n/(\underline c_{\Omega}km),\bar h\right)+ 1$, $(i)$ is from \eqref{eq:var1deltaj}, $(ii)$ is from \eqref{eq:v31}, \eqref{eq:v32}, and \eqref{eq:v1:3terms}, and $(iii)$ is from \eqref{eq:28a}, \eqref{eq:v1:28b4}, \eqref{eq:v1:28c}, and Lemma \ref{lem:zhanglem10}. This completes the proof of Lemma \ref{lem:var1_bound}.
\end{proof}

\vspace{5mm}

\begin{lemma} \label{lem:var2_bound}
Suppose that Assumptions \ref{sampling}--\ref{combine_assumption} hold. Then for every $j=1,\ldots,k$,
\begin{align*}
& \EE_{u^*}\EE_{y_j,u_j} \tr \left[{\var}_{\nu_j\mid y_j,u_j}\{\nu_j(u^*)\}\right]  \\
\leq{} & \frac{5\tau_0^2\lambda_n}{2\underline c_{\Omega} km} \Lambda\left(\tau_0^2\lambda_n/(\underline c_{\Omega} km),\bar h\right) + \frac{4\tilde C_Z^2 km}{\tau_0^2\lambda_n^2} \Tr(\rho,\bar h) \Tr(\rho)  + \lambda_n^{-1} \Tr(\rho,\bar h)  \\
&+  2\lambda_n^{-1}q\bar h\Tr(\rho) \exp\left\{-\frac{m}{8(B^2+B)}\right\},
\end{align*}
where $B=C_{\varphi}^2 \tilde C_Z^2 \underline c_{\Omega}^{-1} q \Lambda\left(\tau_0^2\lambda_n/(\underline c_{\Omega}km),\bar h\right)+ 1$.
\end{lemma}

\begin{proof}[Proof of Lemma \ref{lem:var2_bound}]
For each $a=1,\ldots,q$, we have the eigen-decomposition $\rho_a(u, u') = \sum_{i=1}^{\infty} \mu_{ai} \varphi_{ai}(u) \varphi_{ai}(u')$ for $u, u' \in [0,1]^d$.
This together with the expression of ${\var}_{\nu_j\mid y_j,u_j}\{\nu_j(u^*)\}$ in \eqref{subset_var2} and the orthonormal property of $\left\{\varphi_{ai}\right\}_{i=1}^{\infty}$ imply that
\begin{align} \label{eq:v2a}
&\quad \EE_{u^*}\EE_{y_j,u_j} \tr\left[{\var}_{\nu_j\mid y_j,u_j}\{\nu_j(u^*)\}\right] \nonumber \\
&=  \lambda_n^{-1} \EE_{u^*}\EE_{y_j,u_j} \sum_{a=1}^q \bigg\{ \rho_a(u^*,u^*) - R_{ja}(u^*)^\T \tilde Z_{ja}^\T \left(  \tilde Z_j \tilde R_{jj} \tilde Z_j^\T + \frac{\tau_0^2\lambda_n}{k} I_{m} \right)^{-1} \tilde Z_{ja} R_{ja}(u^*) \bigg\} \nonumber \\
&= \lambda_n^{-1} \sum_{a=1}^q \sum_{h=1}^{\infty} \mu_{ah} \EE_{u^*} \left\{\varphi_{ah}(u^*)^2 \right\}  \nonumber \\
& \quad - \lambda_n^{-1} \EE_{u_j}\sum_{a=1}^q \sum_{i=1}^{m} \sum_{i'=1}^{m} \sum_{h=1}^{\infty} \sum_{h'=1}^{\infty}  \mu_{ah} \mu_{ah'} \left\{ \tilde Z_{ja}^\T \left(  \tilde Z_j \tilde R_{jj} \tilde Z_j^\T + \frac{\tau_0^2\lambda_n}{k} I_{m} \right)^{-1} \tilde Z_{ja}\right\}_{i' i''} \nonumber\\
& \quad \times \left[ \varphi_{ah}(u_{ji}) \varphi_{ah'}(u_{ji'}) \EE_{u^*} \left\{ \varphi_{ah}(u^*) \varphi_{ah'}(u^*) \right\}\right]\nonumber\\
&= \lambda_n^{-1} \sum_{a=1}^q \sum_{h=1}^{\infty} \mu_{ah} - \lambda_n^{-1}\EE_{u_j} \sum_{a=1}^q \sum_{i=1}^{m} \sum_{i'=1}^{m} \sum_{h=1}^{\infty}  \mu^2_{ah} \varphi_{ah}(u_{ji}) \varphi_{ah}(u_{ji'}) \nonumber \\
&\quad \times \left\{ \tilde Z_{ja}^\T \left(  \tilde Z_j \tilde R_{jj} \tilde Z_j^\T + \frac{\tau_0^2\lambda_n}{k} I_{m} \right)^{-1} \tilde Z_{ja} \right\}_{i i'}  \nonumber\\
&= \lambda_n^{-1} \sum_{a=1}^q \sum_{h=1}^{\bar h} \mu_{ah} - \lambda_n^{-1}\EE_{u_j} \sum_{a=1}^q \sum_{i=1}^{m} \sum_{i'=1}^{m} \sum_{h=1}^{\bar h}  \mu^2_{ah} \varphi_{ah}(u_{ji}) \varphi_{ah}(u_{ji'}) \nonumber \\
&\quad \times \left\{ \tilde Z_{ja}^\T \left(  \tilde Z_j \tilde R_{jj} \tilde Z_j^\T + \frac{\tau_0^2\lambda_n}{k} I_{m} \right)^{-1} \tilde Z_{ja} \right\}_{i i'}  \nonumber\\
&\quad + \lambda_n^{-1} \sum_{a=1}^q \sum_{h=\bar h+1}^{\infty} \mu_{ah} - \lambda_n^{-1}\EE_{u_j} \sum_{a=1}^q \sum_{i=1}^{m} \sum_{i'=1}^{m} \sum_{h=\bar h+1}^{\infty}  \mu^2_{ah} \varphi_{ah}(u_{ji}) \varphi_{ah}(u_{ji'}) \nonumber \\
&\quad \times \left\{ \tilde Z_{ja}^\T \left(  \tilde Z_j \tilde R_{jj} \tilde Z_j^\T + \frac{\tau_0^2\lambda_n}{k} I_{m} \right)^{-1} \tilde Z_{ja} \right\}_{i i'}  \nonumber\\
&\overset{(i)}{\leq}  \lambda_n^{-1} \EE_{u_j} \sum_{a=1}^q \sum_{h=1}^{\bar h} \left\{ \mu_{ah} -  \mu^2_{ah} \Phi_{jah}^{\T} \tilde Z_{ja}^\T \left(  \tilde Z_j \tilde R_{jj} \tilde Z_j^\T + \frac{\tau_0^2\lambda_n}{k} I_{m} \right)^{-1} \tilde Z_{ja} \Phi_{jah} \right\}  \nonumber \\
&\quad + \lambda_n^{-1} \Tr(\rho,\bar h) \nonumber \\
& =  \lambda_n^{-1} \EE_{u_j} \tr\left\{ M - M \Phi_{j}^{\T} \tilde Z_{j}^\T \left(  \tilde Z_j \tilde R_{jj} \tilde Z_j^\T + \frac{\tau_0^2\lambda_n}{k} I_{m} \right)^{-1} \tilde Z_{j} \Phi_{j} M \right\} + \lambda_n^{-1} \Tr(\rho,\bar h),
\end{align}
where $\Phi_{jah}$ denotes the $h$th column of the matrix $\Phi_{ja}$ defined in \eqref{eq:ba4}, and $(i)$ follows because we dropped the last negative term to make it larger.

If we let
$$\tilde M = M - M \Phi_{j}^{\T} \tilde Z_{j}^\T \left(  \tilde Z_j \tilde R_{jj} \tilde Z_j^\T + \frac{\tau_0^2\lambda_n}{k} I_{m} \right)^{-1} \tilde Z_{j} \Phi_{j} M ,$$
then \eqref{eq:v2a} has shown that
\begin{align} \label{eq:defvj}
&\EE_{u^*}\EE_{y_j,u_j} \left[{\var}_{\nu_j\mid y_j,u_j}\{\nu_j(u^*)\}\right] \leq \lambda_n^{-1} \EE_{u_j} \tr(\tilde M) + \lambda_n^{-1} \Tr(\rho,\bar h) .
\end{align}
For $j=1,\ldots,k$, $a=1,\ldots,q$, and $h=1,2,\ldots$, we define the following matrices
\begin{align*}
& M_a^{\uparrow} = \diag \left\{\mu_{a(\bar h+1)}, \ldots, \mu_{a \infty}\right\}, \quad M^{\uparrow} = \diag \left\{M_1^{\uparrow}, \ldots, M_q^{\uparrow}\right\}, \\
& \Phi_{jah} = \left\{\varphi_{ah}(u_{j1}),\ldots,\varphi_{ah}(u_{jm})\right\}^\T, \\
& \Phi_{ja}^{\uparrow} = \left\{\Phi_{ja(\bar h+1)},\ldots,\Phi_{ja\infty}\right\}, \quad \Phi_j^{\uparrow}= \diag \left\{\Phi_{j1}^{\uparrow},\ldots,\Phi_{jq}^{\uparrow}\right\}, \\
& \tilde R_{jj}^{\uparrow} = \Phi_j^{\uparrow} M^{\uparrow}  \Phi_j^{\uparrow}, \quad \tilde R_{jj} = \Phi_j M \Phi_j + \tilde R_{jj}^{\uparrow},
\end{align*}
where $\Phi_j \in \RR^{qm\times q\bar h}$ is defined in \eqref{eq:ba4}. Then the Woodbury formula \citep{Har97} and the definition of $Q$ imply that
\begin{align} \label{eq:v2b}
& \tilde M = \left\{ M^{-1} + \Phi_j^{\T} \tilde Z_j^{\T} \left( \tilde Z_j \tilde R_{jj}^{\uparrow} \tilde Z_j^{\T} + \frac{\tau_0^2\lambda_n}{k} I_{m} \right)^{-1} \tilde Z_j \Phi_j   \right\}^{-1} \nonumber\\
&=  \frac{\tau_0^2\lambda_n}{km} \left\{\Omega + \frac{\tau_0^2\lambda_n} {km} M^{-1} + \frac{1}{m}\Phi_j^{\T} \tilde Z_j^{\T} \left(  \tfrac{k}{\tau_0^2\lambda_n} \tilde Z_j \tilde R_{jj}^{\uparrow} \tilde Z_j^{\T} + I_{m} \right)^{-1} \tilde Z_j \Phi_j - \Omega   \right\}^{-1} \nonumber \\
&=  \frac{\tau_0^2\lambda_n}{km} Q^{-1} \left[ I_{q\bar h}  +  Q^{-1} \Omega^{-1} \left\{ \frac{1}{m}\Phi_j^{\T} \tilde Z_j^{\T} \left(  \tfrac{k}{\tau_0^2\lambda_n} \tilde Z_j \tilde R_{jj}^{\uparrow} \tilde Z_j^{\T} + I_{m} \right)^{-1} \tilde Z_j \Phi_j - \Omega  \right\} Q^{-1} \right]^{-1} Q^{-1}\Omega^{-1}.
\end{align}
For $j=1,\ldots,k$, define the event $\Ecal_{j2} = \left\{ \tfrac{k}{\tau_0^2\lambda_n} \tilde Z_j \tilde R_{jj}^{\uparrow} \tilde Z_j^{\T} \preceq \frac{1}{4} I_{m} \right\}$. Since the matrix $ \tilde R_{jj}^{\uparrow}$ is semi-positive definite, we have the relation that
\begin{align*}
\left\{\tr\left(\tfrac{k}{\tau_0^2\lambda_n} \tilde Z_j \tilde R_{jj}^{\uparrow} \tilde Z_j^{\T} \right) \leq \frac{1}{4} \right\} \subseteq \left\{\text{s}_{\max} \left(\tfrac{k}{\tau_0^2\lambda_n} \tilde Z_j \tilde R_{jj}^{\uparrow} \tilde Z_j^{\T} \right) \leq \frac{1}{4} \right\} \subseteq \Ecal_{j2},
\end{align*}
$\text{s}_{\max}(A)$ is the maximum eigenvalue of the square matrix $A$. Therefore, by Markov's inequality and Assumption \ref{z_assumption}, we have that
\begin{align}\label{traceerr1}
& \PP_{u_j} \left( \Ecal_{j2}^c \right) \leq \PP_{u_j} \left\{ \tr\left(\tfrac{k}{\tau_0^2\lambda_n} \tilde Z_j \tilde R_{jj}^{\uparrow} \tilde Z_j^{\T} \right) > \frac{1}{4}\right\} \leq 4 \EE_{u_j} \tr\left(\tfrac{k}{\tau_0^2\lambda_n} \tilde Z_j \tilde R_{jj}^{\uparrow} \tilde Z_j^{\T} \right)  \nonumber \\
& = \frac{4k}{\tau_0^2\lambda_n} \EE_{u_j} \tr\left(\tilde Z_{j} \Phi_j^{\uparrow} M^{\uparrow}  \Phi_j^{\uparrow} \tilde Z_j^{\T}\right) \nonumber \\
& = \frac{4k}{\tau_0^2\lambda_n} \sum_{a=1}^q \sum_{h=\bar h+1}^{\infty} \mu_{ah} \EE_{u_j} \left\{ \sum_{i=1}^m  Z_a(u_{ji})^2 \varphi_{ah}(u_{ji})^2 \right\} \nonumber\\
&\leq \frac{4\tilde C_Z^2 k}{\tau_0^2\lambda_n} \sum_{a=1}^q \sum_{h=\bar h+1}^{\infty} \sum_{i=1}^m \mu_{ah} \EE_{u_j} \left\{ \varphi_{ah}(u_{ji})^2 \right\} \nonumber\\
& = \frac{4\tilde C_Z^2 km}{\tau_0^2\lambda_n} \Tr(\rho,\bar h).
\end{align}
On the event $\Ecal_{j1}\cap \Ecal_{j2}$ (with $\Ecal_{j1}$ defined in \eqref{eset1}), we have that
\begin{align}\label{Bmatbound1}
& \quad I_{q\bar h}  +  Q^{-1} \Omega^{-1} \left\{ \frac{1}{m}\Phi_j^{\T} \tilde Z_j^{\T} \left(  \tfrac{k}{\tau_0^2\lambda_n} \tilde Z_j \tilde R_{jj}^{\uparrow} \tilde Z_j^{\T} + I_{m} \right)^{-1} \tilde Z_j \Phi_j - \Omega  \right\} Q^{-1} \nonumber \\
& \stackrel{(i)}{\succeq} I_{q\bar h}  +  Q^{-1} \Omega^{-1}  \left\{ \frac{1}{m}\Phi_j^{\T} \tilde Z_j^{\T} \left( \frac{1}{4}I_{m} + I_{m} \right)^{-1} \tilde Z_j \Phi_j - \Omega  \right\} Q^{-1} \nonumber \\
& = I_{\bar h} - \frac{1}{5} Q^{-2} + \frac{4}{5} Q^{-1} \Omega^{-1} \left\{ \frac{1}{m}\Phi_j^{\T} \tilde Z_j^{\T} \left( \frac{1}{4}I_{m} + I_{m} \right)^{-1} \tilde Z_j \Phi_j - \Omega \right\} Q^{-1} \nonumber \\
& \stackrel{(ii)}{\succeq} I_{q\bar h} - \frac{1}{5} I_{q\bar h} - \frac{4}{5}\cdot \frac{1}{2}I_{q\bar h} = \frac{2}{5}I_{q\bar h},
\end{align}
where $(i)$ follows on the event $\Ecal_{j2}$, and $(ii)$ holds on the event $\Ecal_{j1}$ and from the fact $Q^{-2}\preceq I_{q\bar h}$. Therefore, from \eqref{eq:v2b} and \eqref{Bmatbound1}, we can obtain that
\begin{align} \label{eq:v2c}
& \EE_{u_j} \left\{\tr(\tilde M)1\left(\Ecal_{j1}\cap \Ecal_{j2}\right)\right\}
\leq  \EE_{u_j} \left\{\frac{\tau_0^2\lambda_n}{km} Q^{-1} \cdot \frac{5}{2}I_{q\bar h} \cdot Q^{-1} \Omega^{-1} \right\} \nonumber \\
={}& \frac{5\tau_0^2\lambda_n}{2km} \tr\left(Q^{-2}\Omega^{-1}\right)
\leq \frac{5\tau_0^2\lambda_n}{2 km} \tr\left\{\left(\Omega + \frac{\tau_0^2 \lambda_n}{km} M^{-1}\right)^{-1}\right\} \nonumber \\
\leq{}& \frac{5\tau_0^2\lambda_n}{2 km} \sum_{a=1}^q \sum_{h=1}^{\bar h} \frac{1}{\underline c_{\Omega} +\frac{\tau_0^2 \lambda_n}{km \mu_{ah}}}
= \frac{5\tau_0^2\lambda_n}{2\underline c_{\Omega} km} \Lambda\left(\tau_0^2\lambda_n/(\underline c_{\Omega} km),\bar h\right).
\end{align}
Therefore, by combining \eqref{eq:defvj}, \eqref{traceerr1}, \eqref{eq:v2c}, and Lemma \ref{lem:zhanglem10}, we obtain that
\begin{align} 
&\quad \EE_{u^*}\EE_{y_j,u_j} \tr \left[{\var}_{\nu_j\mid y_j,u_j}\{\nu_j(u^*)\}\right] \nonumber \\
&\leq \lambda_n^{-1} \EE_{u_j} \tr(\tilde M) + \lambda_n^{-1} \Tr(\rho,\bar h) \nonumber \\
&\leq \lambda_n^{-1} \EE_{u_j} \left\{\tr(\tilde M)1\left(\Ecal_{j1} \cap \Ecal_{j2} \right)\right\} + \lambda_n^{-1}\EE_{u_j} \left\{\tr(\tilde M) 1\left(\Ecal_{j1}^c\right)\right\} \nonumber \\
&\quad + \lambda_n^{-1}\EE_{u_j} \left\{\tr(\tilde M)1\left(\Ecal_{j2}^c\right)\right\} + \lambda_n^{-1} \Tr(\rho,\bar h) \nonumber \\
&\leq \frac{5\tau_0^2\lambda_n}{2\underline c_{\Omega} km} \Lambda\left(\tau_0^2\lambda_n/(\underline c_{\Omega} km),\bar h\right) + \lambda_n^{-1} \Tr(\rho) \PP_{u_j} \left(\Ecal_{j1}^c\right) \nonumber \\
&\quad \frac{4\tilde C_Z^2 km}{\tau_0^2\lambda_n^2} \Tr(\rho,\bar h) \Tr(\rho)  + \lambda_n^{-1} \Tr(\rho,\bar h) \nonumber \\
&\leq \frac{5\tau_0^2\lambda_n}{2\underline c_{\Omega} km} \Lambda\left(\tau_0^2\lambda_n/(\underline c_{\Omega} km),\bar h\right) + \frac{4\tilde C_Z^2 km}{\tau_0^2\lambda_n^2} \Tr(\rho,\bar h) \Tr(\rho)  + \lambda_n^{-1} \Tr(\rho,\bar h)  \nonumber \\
&\quad +  2\lambda_n^{-1}q\bar h\Tr(\rho) \exp\left\{-\frac{m}{8(B^2+B)}\right\}, \nonumber
\end{align}
where $B=C_{\varphi}^2 \tilde C_Z^2 \underline c_{\Omega}^{-1} q \Lambda\left(\tau_0^2\lambda_n/(\underline c_{\Omega}km),\bar h\right)+ 1$. This completes the proof of Lemma \ref{lem:var2_bound}.
\end{proof}

\vspace{8mm}

\begin{lemma}\label{lem:zhanglem10}
For the event $\Ecal_{j1}$ defined in \eqref{eset1}, the probability of the event $\Ecal_{j1}^c$ is upper bounded by
\begin{align} \label{eset1:prob}
& \PP_{u_j} \left(\Ecal_{j1}^c \right) \leq  2q\bar h \exp\left\{-\frac{m}{8(B^2+B)} \right\},
\end{align}
where $B=C_{\varphi}^2 \tilde C_Z^2 \underline c_{\Omega}^{-1} q \Lambda\left(\tau_0^2\lambda_n/(\underline  c_{\Omega}km),\bar h\right)+ 1$.
\end{lemma}

\begin{proof}[Proof of Lemma \ref{lem:zhanglem10}]
For $j=1,\ldots,k$ and $i=1,\ldots,m$, let
\begin{align*}
W_{ji}={}&\big\{\tilde Z_1(u_{ji})\varphi_{11}(u_{ji}),\ldots, \tilde Z_1(u_{ji})\varphi_{1\bar h}(u_{ji}),\ldots, \nonumber\\
&~~\tilde Z_q(u_{ji})\varphi_{q1}(u_{ji}),\ldots, \tilde Z_q(u_{ji})\varphi_{q\bar h}(u_{ji})\big\}^\T \in \RR^{q \bar h},
\end{align*}
similar to $W(u)$ in Assumption \ref{z_assumption}. With some linear algebra, the matrix in the definition \eqref{eset1} can be rewritten as
\begin{align} \label{eq:ecal:1}
& Q^{-1} \Omega^{-1} \left( \frac{1}{m}\Phi_j^\T \tilde Z_j^\T \tilde Z_j \Phi_j - \Omega \right) Q^{-1} = Q^{-1} \Omega^{-1} \left( \frac{1}{m} \sum_{i=1}^m W_{ji} W_{ji}^\T - \Omega \right) Q^{-1} \nonumber\\
={}& \frac{1}{m} \sum_{i=1}^m Q^{-1} \Omega^{-1}\left(W_{ji} W_{ji}^\T - \Omega \right) Q^{-1}.
\end{align}
Using Assumptions \ref{eigen_assumption}, \ref{z_assumption}, and the fact that $Q\succeq I_{q\bar h}$, we can obtain that for every $j=1,\ldots,k$ and $i=1,\ldots,m$,
\begin{align} \label{eq:ecal:2}
&\quad ~\vertiii{Q^{-1} \Omega^{-1}\left(W_{ji} W_{ji}^\T - \Omega \right) Q^{-1}}  \nonumber \\
&\leq  \vertiii{Q^{-1}  \Omega^{-1} W_{ji} W_{ji}^\T Q^{-1}} + \vertiii{Q^{-2}} \nonumber \\
&\leq  W_{ji}^\T Q^{-2}\Omega^{-1} W_{ji}   +  1\nonumber\\
&\leq W_{ji}^\T \left(\Omega + \frac{\tau_0^2 \lambda_n}{km} M^{-1}\right)^{-1} W_{ji}  + 1   \nonumber \\
&\leq \underline c_{\Omega}^{-1} \sum_{a=1}^q\sum_{h=1}^{\bar h}  \frac{\mu_{ah}}{\mu_{ah} + \frac{\tau_0^2\lambda_n}{\underline c_{\Omega} km}} \tilde Z_a(u_{ji})^2 \varphi_{ah}(u_{ji})^2 + 1 \nonumber \\
&\leq C_{\varphi}^2 \tilde C_Z^2 \underline c_{\Omega}^{-1} q \Lambda\left(\tau_0^2\lambda_n/(\underline c_{\Omega} km),\bar h\right)+ 1 \equiv B.
\end{align}
Furthermore,
\begin{align} \label{eq:ecal:3}
&\vertiii{\left\{Q^{-1} \Omega^{-1}\left(W_{ji} W_{ji}^\T - \Omega \right) Q^{-1}\right\}^2} \nonumber\\
\leq{}& \vertiii{Q^{-1} \Omega^{-1}\left(W_{ji} W_{ji}^\T - \Omega \right) Q^{-1}}^2 \leq B^2.
\end{align}
Now from \eqref{eq:ecal:1}, \eqref{eq:ecal:2} and \eqref{eq:ecal:3}, we apply the matrix Bernstein inequality (Theorem 6.1.1 of \citealt{Tro15}) to the sequence of $\left\{Q^{-1} \Omega^{-1}\left(W_{ji} W_{ji}^\T - \Omega \right) Q^{-1} \right\}_{i=1}^m$ to obtain that
\begin{align}
&\quad \PP_{u_j} \left(\Ecal_{j1}^c \right)  = \PP_{u_j} \left(\vertiii{  Q^{-1} \Omega^{-1} \left( \frac{1}{m}\Phi_j^\T \tilde Z_j^\T \tilde Z_j \Phi_j - \Omega \right) Q^{-1} } > 1/2\right) \nonumber\\
&\leq \PP_{u_j} \left(\vertiii{  \frac{1}{m} \sum_{i=1}^m Q^{-1} \Omega^{-1}\left(W_{ji} W_{ji}^\T - \Omega \right) Q^{-1} } > 1/2\right) \nonumber\\
&\leq 2q\bar h \exp\left\{-\frac{(m/2)^2/2}{m B^2 + mB /6 } \right\} \nonumber \\
&\leq 2q\bar h \exp\left\{-\frac{m}{8(B^2+B) } \right\}, \nonumber
\end{align}
which completes the proof.
\end{proof}

\vspace{8mm}

\begin{proof}[Proof of Theorem \ref{thm:main}]
We prove part (i) and part (ii), respectively. We simplify the upper bounds in Lemmas \ref{lem:bias_bound}, \ref{lem:var1_bound}, and \ref{lem:var2_bound} with the choice $\lambda_n=1$ and $\lambda_n\asymp n^{-(2\vv)/(2\vv+d)}$.
\vspace{2mm}

\noindent (i) With $\lambda_n=1$, we first derive a bound for the quantity $\Lambda(\tau_0^2\lambda_n/(ckm),\bar h)$ for a generic constant $c>0$ (which will be replaced by $\overline c_{\Omega}$ and $\underline c_{\Omega}$ in the upper bounds in Lemmas \ref{lem:bias_bound}, \ref{lem:var1_bound}, and \ref{lem:var2_bound}). Using Assumption \ref{eigen_assumption}, there exists some constant $c_{\mu}>0$ such that $\mu_{i*} \leq c_{\mu} i^{-2\vv/d}$ for $i=1,2,\ldots$. Therefore, given the fact that $2\vv>d$, we have that if $\lambda_n=1$,
\begin{align} \label{Lambda:order1}
& \Lambda(\tau_0^2\lambda_n/(ckm),\bar h) \leq \sum_{h=1}^{\infty} \left(1 + \frac{\tau_0^2}{ckm\mu_{h*}}\right)^{-1} \leq \sum_{h=1}^{\infty} \left(1 + \frac{c_{\mu}c_2\tau_0^2}{cn}h^{2\vv/d}\right)^{-1} \nonumber \\
& \leq \sum_{h\leq n^{d/(2\vv)}} \left(1 + \frac{c_{\mu}c_2\tau_0^2}{cn}h^{2\vv/d}\right)^{-1} + \sum_{h>n^{d/(2\vv)}} \left(1 + \frac{c_{\mu}c_2\tau_0^2}{cn}h^{2\vv/d}\right)^{-1} \nonumber \\
& \leq n^{d/(2\vv)} + \sum_{h>n^{d/(2\vv)}} \left( \frac{c_{\mu}c_2\tau_0^2}{cn}h^{2\vv/d}\right)^{-1} \nonumber \\
& \leq n^{d/(2\vv)} + \left( \frac{c_{\mu}c_2\tau_0^2}{cn}\right)^{-1}\sum_{h>n^{d/(2\vv)}}  \int_{h}^{h+1} x^{-2\vv/d} dx \nonumber \\
& \leq n^{d/(2\vv)} + \left( \frac{c_{\mu}c_2\tau_0^2}{cn}\right)^{-1} n^{-\frac{d}{2\vv}\cdot \left(\tfrac{2\vv}{d}-1\right)} \nonumber \\
& \asymp n^{d/(2\vv)} + n^{1-\frac{2\vv-d}{2\vv}} \nonumber \\
& \asymp n^{d/(2\vv)}.
\end{align}
Given the condition $m \gtrsim n^{(d/\vv)+\eta}$, in the exponent of \eqref{eset1:prob}, we have
\begin{align*}
& \frac{m}{8(B^2+B) } \asymp \frac{m}{ \Lambda(\tau_0^2\lambda_n/(\underline c_{\Omega}km),\bar h)^2 } \gtrsim \frac{n^{(d/\vv)+\eta}}{n^{d/\vv}} = n^{\eta}.
\end{align*}
This implies that for some positive constant $c'>0$,
\begin{align}\label{exp:bound1}
& \exp\left\{- \frac{m}{8(B^2+B)  } \right\} \leq \exp(-c'n^\eta).
\end{align}
With the choice $\bar h=\lceil n^{3d/(2\vv-d)} \rceil$, we have that
\begin{align} \label{trace:dstar}
& \Tr(\rho,\bar h) = \sum_{a=1}^q \sum_{h\geq \bar h} \mu_{ah} \leq q \sum_{h\geq \bar h } \mu_{h*} \leq q \sum_{h\geq \bar h } c_{\mu} h^{-2\vv/d} \nonumber \\
&\leq q c_{\mu} \sum_{h\geq \bar h} \int_{h}^{h+1} x^{-2\vv/d} dx \leq c_{\mu} \int_{n^{3d/(2\vv-d)}}^{\infty} x^{-2\vv/d} dx \nonumber \\
&= q c_{\mu} n^{-\frac{3d}{2\vv-d}\cdot \frac{2\vv-d}{d}} \asymp n^{-3}.
\end{align}

Note that Assumption \ref{eigen_assumption} (ii) and Assumption \ref{w0_assumption} imply that $\Tr(\rho)=O(1)$ and $\|\nu_0\|_{\HH}=O(1)$. Using the orders in \eqref{Lambda:order1}, \eqref{exp:bound1}, and \eqref{trace:dstar}, the order of the upper bound in Lemma \ref{lem:bias_bound} can be quantified as
\begin{align}\label{bias_bound1}
&\quad \frac{4}{k} \sum_{j=1}^k \EE_{u^*} \EE_{u_j} \left\| \EE_{\nu_{j}, y_j\mid u_j}\{\nu_{j}(u^*)\} - \nu_{0}(u^*) \right\|_2 ^2   \nonumber \\
&\lesssim n^{-1}  + n \cdot n^{-3} +  n^{-3} + n^{\lceil 3d/(2\vv-d)\rceil } \cdot \exp (-c'n^\eta)  \nonumber \\
&\lesssim n^{-1}.
\end{align}
Similarly, the order of the upper bound in Lemma \ref{lem:var1_bound} can be quantified as
\begin{align}\label{var1_bound1}
&\quad \frac{4}{k^2} \sum_{j=1}^k \EE_{u^*}\EE_{u_j} \tr\left( {\var}_{y_j\mid u_j} \left[\EE_{\nu_{j} \mid y_j,u_j}\{\nu_{j}(u^*)\} \right]\right)  \nonumber  \\
&\lesssim \frac{1}{kn} + \frac{n}{k}\cdot n^{-3}\cdot \left(n + 1 \right) \nonumber \\
&\quad + \frac{1}{km} \cdot n^{d/(2\vv)}+ \frac{1}{k}\cdot n^{-3} \left(n + 1 \right)  \nonumber \\
&\quad + \frac{1}{k}\cdot n^{\lceil 3d/(2\vv-d)\rceil} \left( n + 1\right) \cdot \exp(-c'n^\eta)  \nonumber \\
&\lesssim k^{-1} n^{-1} + n^{-(2\vv-d)/(2\vv)} + k^{-1} n^{-2} + k^{-1} n^{\lceil 3d/(2\vv-d)\rceil+1}\cdot \exp(-c'n^\eta)  \nonumber \\
&\lesssim n^{-(2\vv-d)/(2\vv)}.
\end{align}
The order of the upper bound in Lemma \ref{lem:var2_bound} can be quantified as
\begin{align}\label{var2_bound1}
&\quad \frac{\overline c}{k} \sum_{j=1}^k \EE_{u^*}\EE_{y_j,u_j}\tr\left({\var}_{\nu_{j} \mid y_j,u_j}\{\nu_{j}(u^*)\}\right) \nonumber \\
&\lesssim n^{-1} \cdot n^{d/(2\vv)} + \left(n+1\right) \cdot n^{-3}  + n^{\lceil 3d/(2\vv-d)\rceil} \cdot \exp(-c'n^\eta)  \nonumber \\
&\lesssim n^{-(2\vv-d)/(2\vv)}.
\end{align}
Finally, we combine \eqref{bias_bound1}, \eqref{var1_bound1}, \eqref{var2_bound1}, \eqref{L2_risk1}, and \eqref{combine_subset_bias_var2} to obtain that
\begin{align*}
\LL(\overline \Pi) &\lesssim n^{-1} + n^{-(2\vv-d)/(2\vv)} + n^{-(2\vv-d)/(2\vv)} + n^{-1} \lesssim n^{-(2\vv-d)/(2\vv)}.
\end{align*}
The rate for $\overline w(\cdot)$ follows trivially from the inequality
\begin{align*}
& |\overline w(u^*)-w_0(u^*)| = \left| Z(u^*) \overline \beta(u^*) - Z(u^*) \beta_0(u^*)\right| \leq C_Z\left\|\overline \beta(u^*) - \beta_0(u^*)\right\|_2.
\end{align*}
This proves the conclusion of Theorem \ref{thm:main} (i).

\vspace{5mm}

\noindent (ii) When $\lambda_n\asymp n^{d/(2\vv+d)}$, similar to part (i), we first derive a bound for the quantity $\Lambda_n(\tau_0^2\lambda_n/(ckm),\bar h)$ with $c>0$ being a generic constant. Using Assumption \ref{eigen_assumption}, there exists some constant $c_{\mu}>0$ such that $\mu_j \leq c_{\mu} j^{-2\vv/d}$. Therefore, we have that
\begin{align} \label{Lambda:order}
& \Lambda(\tau_0^2\lambda_n/(km),\bar h) \leq \sum_{h=1}^{\infty} \left(1 + \frac{\tau_0^2\lambda_n}{ckm\mu_{h*}}\right)^{-1} \leq \sum_{h=1}^{\infty} \left(1 + \frac{c_{\mu}c_2\tau_0^2\lambda_n}{cn}h^{2\vv/d}\right)^{-1} \nonumber \\
& \leq \sum_{h\leq n^{d/(2\vv+d)}} \left(1 + \frac{c_{\mu}c_2\tau_0^2\lambda_n}{n}h^{2\vv/d}\right)^{-1} + \sum_{h>n^{d/(2\vv+d)}} \left(1 + \frac{c_{\mu}c_2\tau_0^2\lambda_n}{cn}h^{2\vv/d}\right)^{-1} \nonumber \\
& \leq n^{d/(2\vv+d)} + \sum_{h>n^{d/(2\vv+d)}} \left( \frac{c_{\mu}c_2\tau_0^2\lambda_n}{cn}h^{2\vv/d}\right)^{-1} \nonumber \\
& \leq n^{d/(2\vv+d)} + \left( \frac{c_{\mu}c_2\tau_0^2\lambda_n}{cn}\right)^{-1}\sum_{h>n^{d/(2\vv+d)}}  \int_{h}^{h+1} x^{-2\vv/d} dx \nonumber \\
& \leq n^{d/(2\vv+d)} + \left( \frac{c_{\mu}c_2\tau_0^2\lambda_n}{cn}\right)^{-1} n^{-\frac{d}{2\vv+d}\cdot \left(\tfrac{2\vv}{d}-1\right)} \nonumber \\
& \asymp n^{d/(2\vv+d)} + \lambda_n^{-1} n^{1-\frac{2\vv-d}{2\vv+d}} \nonumber \\
& \asymp n^{d/(2\vv+d)},
\end{align}
where the last step follows because $\lambda_n \asymp n^{d/(2\vv+d)}$.

Given our condition $m \gtrsim n^{2d/(2\vv+d)+\eta}$, in the exponent of \eqref{eset1:prob}, we have
\begin{align*}
& \frac{m}{8(B^2+B) }  \asymp \frac{m}{ \Lambda(\tau_0^2\lambda_n/(ckm),\bar h)^2 } \gtrsim \frac{n^{2d/(2\vv+d)+\eta}}{n^{2d/(2\vv+d)}} = n^{\eta}.
\end{align*}
This implies that for some positive constant $c'>0$,
\begin{align}\label{exp:bound}
& \exp\left\{\frac{m}{8(B^2+B) } \right\} \leq \exp(-c'n^\eta).
\end{align}

With the choice $\bar h=\lceil n^{3d/(2\vv-d)} \rceil$, $\Tr(\rho,\bar h)$ is upper bounded by $n^{-3}$ as in \eqref{trace:dstar}.

Using the orders in \eqref{Lambda:order}, \eqref{exp:bound}, and \eqref{trace:dstar}, the order of the upper bound in Lemma \ref{lem:bias_bound} can be quantified as
\begin{align}\label{bias_bound2}
&\quad \frac{4}{k} \sum_{j=1}^k \EE_{u^*} \EE_{u_j} \left\| \EE_{\nu_{j}, y_j\mid u_j}\{\nu_{j}(u^*)\} - \nu_{0}(u^*) \right\|_2 ^2 \nonumber \\
&\lesssim \frac{n^{d/(2\vv+d)}}{n}  + \frac{n}{n^{d/(2\vv+d)}}\cdot n^{-3} +  n^{-3} + n^{\lceil 3d/(2\vv-d) \rceil} \cdot \exp (-c'n^\eta)  \nonumber \\
&\lesssim n^{-2\vv/(2\vv+d)}.
\end{align}
Similarly, the order of the upper bound in Lemma \ref{lem:var1_bound} can be quantified as
\begin{align}\label{var1_bound2}
&\quad \frac{4}{k^2} \sum_{j=1}^k \EE_{u^*}\EE_{u_j} \tr\left( {\var}_{y_j\mid u_j} \left[\EE_{\nu_{j} \mid y_j,u_j}\{\nu_{j}(u^*)\} \right]\right)  \nonumber  \\
&\lesssim \frac{n^{d/(2\vv+d)}}{kn} + \frac{n}{kn^{d/(2\vv+d)}}\cdot n^{-3}\cdot \left(\frac{n}{n^{d/(2\vv+d)}} + 1 \right) \nonumber \\
&\quad + \frac{1}{km} \cdot n^{d/(2\vv+d)}+ \frac{1}{k}\cdot n^{-3} \left( \frac{2n}{n^{d/(2\vv+d)}} + 1 \right)  \nonumber \\
&\quad + \frac{1}{k}\cdot n^{\lceil 3d/(2\vv-d) \rceil} \left( \frac{n}{n^{d/(2\vv+d)}} + 1\right) \cdot \exp(-c'n^\eta)  \nonumber \\
&\lesssim k^{-1} n^{-2\vv/(2\vv+d)} + k^{-1} n^{4\vv/(2\vv+d)-3} + n^{-2\vv/(2\vv+d)} + k^{-1} n^{2\vv/(2\vv+d)-3} \nonumber \\
&\quad + k^{-1} n^{\lceil 3\vv/(2\vv-d)\rceil +2\vv/(2\vv+d)}\cdot \exp(-c'n^\eta)  \nonumber \\
&\lesssim n^{-2\vv/(2\vv+d)}.
\end{align}
The order of the upper bound in Lemma \ref{lem:var2_bound} can be quantified as
\begin{align}\label{var2_bound2}
&\quad \frac{\overline c}{k} \sum_{j=1}^k \EE_{u^*}\EE_{y_j,u_j}\tr\left({\var}_{\nu_{j} \mid y_j,u_j}\{\nu_{j}(u^*)\}\right) \nonumber \\
&\lesssim \frac{1}{n} \cdot n^{d/(2\vv+d)} + \left(n\cdot n^{-2d/(2\vv+d)} +n^{-d/(2\vv+d)}\right) \cdot n^{-3} \nonumber \\
&\quad + n^{-d/(2\vv+d)}\cdot n^{\lceil 3d/(2\vv-d) \rceil} \cdot \exp(-c'n^\eta)  \nonumber \\
&\lesssim n^{-2\vv/(2\vv+d)}.
\end{align}
Finally, we combine \eqref{bias_bound2}, \eqref{var1_bound2}, and \eqref{var2_bound2}, \eqref{L2_risk1}, and \eqref{combine_subset_bias_var2} to obtain that
\begin{align*}
\LL(\overline \Pi) &\lesssim n^{-2\vv/(2\vv+d)} + n^{-2\vv/(2\vv+d)} + n^{-2\vv/(2\vv+d)} + n^{-1} \lesssim n^{-2\vv/(2\vv+d)}.
\end{align*}
The rate for $\overline w(\cdot)$ follows similarly. This proves the conclusion of Theorem \ref{thm:main} part (ii).
\end{proof}

\bibliographystyle{plainnat}
\bibliography{papers}

\end{document}